\documentclass{report}
\ifdefined\directlua
  \usepackage{fontspec}
\else
  \usepackage[utf8x]{inputenc} 
  \usepackage[nomath]{lmodern}
\fi

\usepackage[english]{babel} 

\usepackage{longtable}

\usepackage[T1]{fontenc}
\usepackage{graphicx}
\usepackage[pagebackref=false,colorlinks,linkcolor=blue,citecolor=red]{hyperref}
\usepackage{caption}

\usepackage{amssymb}

\usepackage{amsthm}
\usepackage{amsmath}
\usepackage{mathtools}

\usepackage{csquotes}
\usepackage{bm}

\usepackage{xcolor}
\usepackage{algorithm,algorithmic}
\usepackage{caption}

\newcommand{\INDSTATE}[1][1]{\STATE\hspace{#1\algorithmicindent}}

\usepackage{eqparbox,array}

 \usepackage{parskip}
 \usepackage[a4paper, total={6in, 8in}]{geometry}

 \def\code#1{\texttt{#1}}
 
 \usepackage{enumitem}
\newlist{steps}{enumerate}{1}
\setlist[steps, 1]{label = Step \arabic*:}

\usepackage{float}

\usepackage{lipsum}
\usepackage{mwe}
\usepackage{graphicx}

\usepackage{multirow}
\usepackage{makecell}

\usepackage{titlesec}
\newcommand{\nocontentsline}[3]{}
\newcommand{\tocless}[2]{\bgroup\let\addcontentsline=\nocontentsline#1{#2}\egroup}

\usepackage{pdfpages}

\usepackage{fancyhdr}
\title{Asset Price Forecasting using Recurrent Neural Networks}

\newcommand{\thesisauthor}{Hamed Vaheb} 
\newcommand{\thesistype}{Master Thesis} 
\newcommand{\supervisor}{Dr. Erfan Salavati, Dr. Mina Aminghafari}

\begin{document}

\newcommand\norm[1]{\left\lVert#1\right\rVert}
\newcommand\myeq{\stackrel{\mathclap{\normalfont\mbox{def}}}{=}}

\newcommand{\oldoptimal}[1]{{#1}^*}
\newcommand{\newoptimal}[1]{#1^*}


\makeatletter
\begin{titlepage}
    \begin{center}
         \begin{figure}[htbp]
             \centering
             \includegraphics[width=.4\linewidth, scale=1.5]{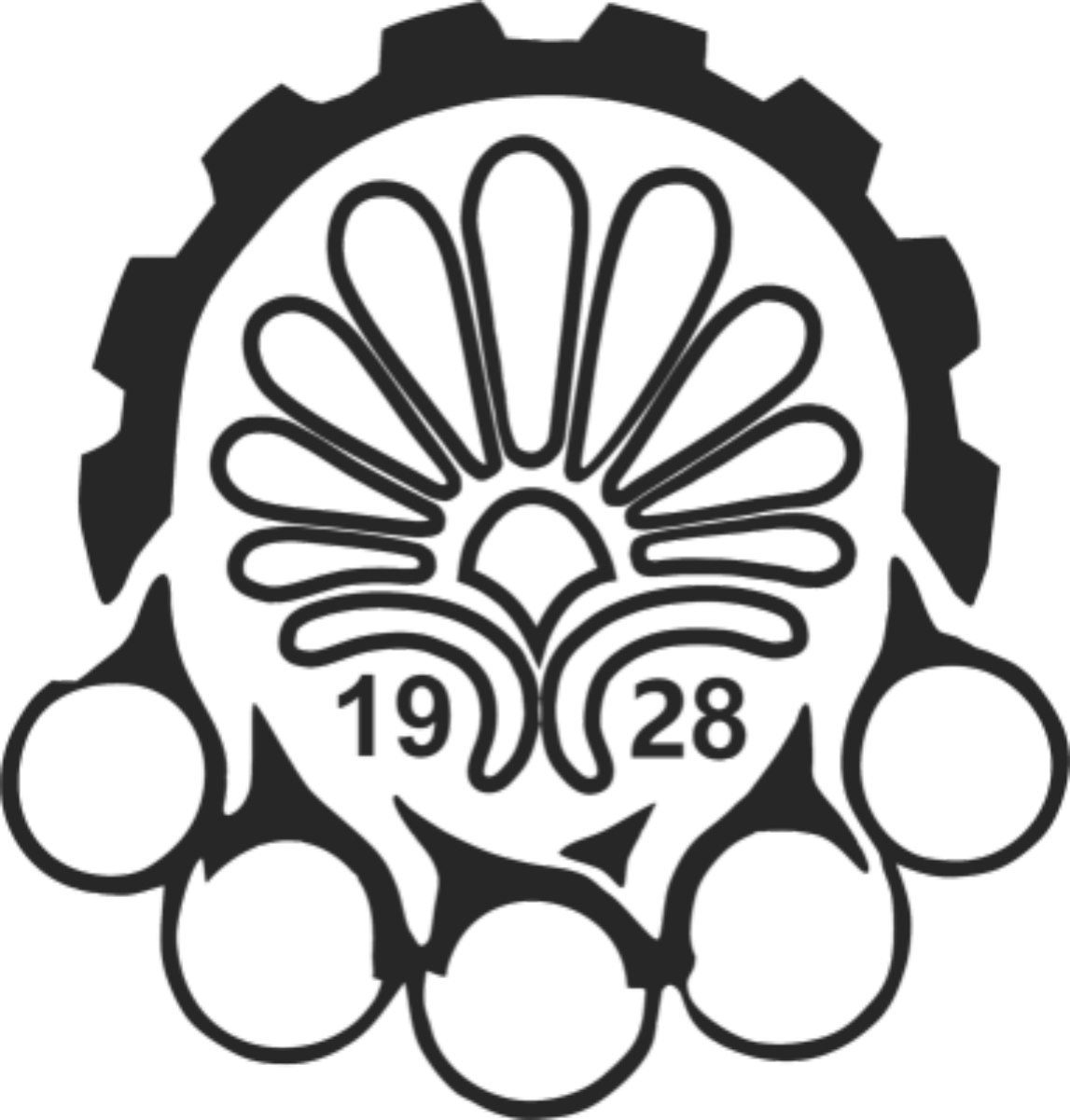}
        \end{figure}
        
        \huge
        Amirkabir University of Technology
        \\ 
        \smallskip
        
        \Large
        (Tehran Polytechnic) \\
        
        \smallskip
        \smallskip
        \Large
        Department of Mathematics and Computer Science\\
        
        \bigskip
        \bigskip
        
        \Large
        \thesistype{}
        
        \Large
        Field: Financial Mathematics \\
        
        \bigskip
        \bigskip
        \bigskip
        \bigskip
        
        \LARGE
        Title\\
        \@title

        \bigskip
        \bigskip
        
        \large
        Author \\
        \Large{}
        \thesisauthor{} \\ 
        
        \bigskip
        \bigskip
        
        \large
        Supervisor \\ 
        \Large
        \supervisor{}\\

        \bigskip
        \bigskip
        \bigskip
        \bigskip
        
        \Large
        July, 2020

    \end{center}
\end{titlepage}
\makeatother



\clearpage

\pagenumbering{gobble}
\clearpage
\thispagestyle{plain}
\includepdf{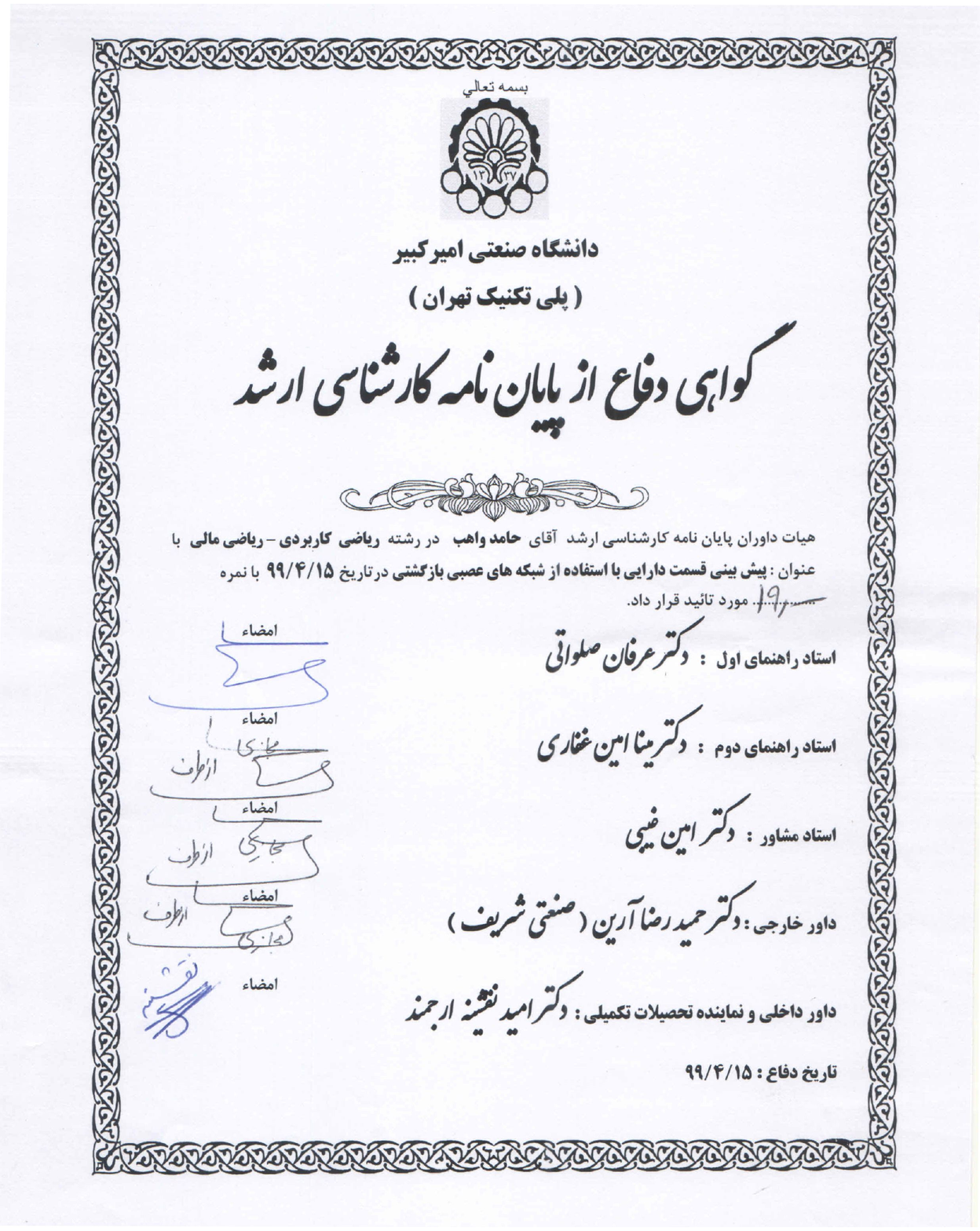}
\clearpage
\clearpage
\thispagestyle{empty}
\section*{Statement of Authorship}
\label{sec:SOOA}


I hereby declare that I am the sole author of this master thesis and that I have not used any sources other than those listed in the
bibliography and identified as references. I further declare that I have not submitted
this thesis at any other institution in order to obtain a degree.

\vspace{1cm}

\vspace{3cm}
\begin{center}

\textbf{\thesisauthor{}} 
\end{center}
\vspace{0.5cm}

\clearpage
\thispagestyle{plain}
\pagenumbering{Roman}

\medskip
\medskip
\medskip
\begin{center}
\Large
\textbf{Abstract}
\end{center}
\vspace{1cm}
This thesis serves three primary purposes, first of which is to forecast two prices from S\& P500,  Goldman Sachs (GS) and General Electric (GE). \\
In recent years, financial market dynamics forecasting has been a focus of economic research. 
Prediction of stock market prices is considered a monumental task and has garnered considerable attention since if it is done accurately, it leads to lucrative decisions such as constructing a portfolio or investing in profitable stocks. Stock market prediction is quite challenging due to non-stationary, non-linear blaring, and chaotic data. Artificial Neural Networks (ANN) are models that are currently used in a wide variety of fields thanks to their superior performance; the field of finance and the task of time series forecasting and particularly, stock prices forecasting is no exception.
In this thesis, in order to forecast stock prices, we used a certain type of Recurrent Neural Network (RNN), called long short-term memory (LSTM) model in which we inputted the prices of two other stocks that lie in rather close correlation with GS. On the other hand, GE was inputted with its historical data designated as the sole feature. Other models such as ARIMA were used as benchmark. Empirical results manifest the practical challenges when using LSTM for forecasting stocks. One of the main upheavals was a recurring delay in forecast, which we explained under the term of forecasting lag. We concluded that even though LSTM is theoretically perhaps the best candidate among machine learning and neural network models for modeling and predicting stock prices, it seems quite cumbersome to approach prices with LSTM, in a practical sense. \\ \\
The second purpose of this thesis goes beyond attempting a single time series forecasting task. Rather, it is to develop a more general and objective perspective on the task of time series forecasting so that it could be applied to assist in an arbitrary that of forecasting by ANNs. Thus, attempts are made for distinguishing previous works by certain criteria (introduced by a review paper written by Ahmed Tealab) so as to summarise those including effective information. The summarised information is then unified and expressed through a common terminology that can be applied to different steps of a time series forecasting task. Hopefully, this information can serve as a methodical approach that requires less subjective experience and intervention and aid of the researcher.\\ \\ 
The last but not least purpose of this thesis is to elaborate on a mathematical framework on which ANNs are based. At this point we do not have a full rigorous understanding of why ANNs perform well, and how exactly to construct neural networks that works out for a specific problem. We are going to use the framework introduced in the book "Neural Networks in Mathematical Framework" by Anthony L. Caterini in which the structure of a generic neural network is introduced and the gradient descent algorithm (which incorporates backpropagation) is introduced in terms of their described framework. In the end, we use this framework for a specific architecture, which is recurrent neural networks on which we concentrated and our implementations are based. The book proves its theorems mostly for classification case. Instead, we proved theorems for regression case, which is the case of our problem.
\vfill

\textit{Keywords}: Neural Networks, Recurrent Neural Networks, LSTM, Long short-term memory, ARIMA, Forecasting, Prediction, Time Series, Stock Price, Sequential Data

\clearpage
\tableofcontents
\clearpage
\listoffigures
\clearpage
\listoftables
\clearpage
\newpage

\pagenumbering{arabic}

\clearpage
\chapter{Literature Review}
\clearpage
In this chapter, after introducing the terminology and definitions used in the problem of time series forecasting, previous works are presented from two different perspectives, one of which accounts for the studies that have been carried out specifically for the task of stock market prediction, from which we will concentrate on those including recurrent neural networks (RNNs). Another perspective delves into works that have employed Artificial Neural Networks (ANNs) methods for the task of time series forecasting and distinguish papers based on certain criteria introduced by Ahmed Tealab. At the end of the section, based on the introduced criteria, essential information of the distinguished papers are extracted for facilitating further research.

\section{Terminology}

In what follows, the problem of time series forecasting is explained thoroughly in addition to its technical terms.\\ \\

\subsection{Concept of Time} \label{subsec:timeconcept}

There are a plethora of phenomena in the universe that behave in accordance with passage of time, such as lunar phases. But the concept of time can be a nebulous one. There is a matter of contention whether time is a subjective illusion that is only perceived through humans (and some other animals) or it is an objective and inherent quality within the universe. Perhaps it is merely our mind that sees the world in chronological order rather than the world itself incorporating the role of time as a factor of its existence. Many people, including philosophers and scientists, support this rather odd latter viewpoint.
Kant, for instance, maintained that time is a pure \textit{form of intuition}, a term he coined to describe sensory aspects (e.g., space is another aspect) of our experience of the world, i.e. the way in which we passively receive data through sensibility. This passiveness implies that we can not be certain how much the reality outside of our mind resembles our perception of it. Another scientist who had a somewhat similar perspective is Einstein. Many physicists developed his theory of “the block universe” with which he came up soon after his special theory of relativity. The block universe theory states that the present is not the sole form of time that exists from which past and present derive. In effect, not only has the past not been gone and forgotten, but the future has not yet to be decided. Instead, all times coexist. Therefore, it may not matter how adept or inept one would be in considering or predicting future outcomes since every single outcome that he picks already exists and will do so perpetually. There is a growing consensus among physicists that this theory is really the case and the way things really are. All in all, despite the fact the time might not be a factor contributing to the formation of reality or that one's intention of predicting future events would not necessarily defy our ineffectiveness and alter the possible fact that our future might be set in stone, the sense of prediction and hence control restore a man's majesty to his life. As a great philosopher once said, "To a man devoid of blinders, there is no finer sight than that of intelligence at grips with a reality that transcends it."

\subsection{Time Series}
Whether one intends to predict the trend in financial markets or electricity consumption, time is an important factor that should be considered in models since trends and salient patterns in data can be detected by their proportion of change over time. It would be desirable that not only would one be aware of the amount that a stock will increase, but also when it will move up. Time series is a suitable entity for modelling data with respect to its alteration over time whether these data represent a deterministic phenomenon (such as the dynamic of an object after a push within an environment with no external force) or a stochastic one which is called a stochastic process (such as stock prices). The attribution of a stochastic process to a time series resembles that of a random variable to a number.\\
\begin{quote}
"A \textit{time series} is a sequence of observations taken sequentially in time." \cite{tsar3}
\end{quote}
\begin{quote}
    "There is, however, a key feature that distinguishes financial time series analysis from other time series analysis. Both financial theory and its empirical time series contain an element of uncertainty. For example, there are various definitions of asset volatility, and for a stock return series, the volatility is not directly observable." \cite{fts}
\end{quote}

\subsection{Describing vs Predicting}
One would have different goals depending on whether he is interested in understanding and describing a dataset or making predictions about future observations.

Understanding a dataset, called \textit{time series analysis}, may assist in refining future predictions. However, it is not a necessity and also it may consume time and energy and require expertise that is not directly in line with the desired task, which is forecasting the future. \\

\begin{quote}
In \textit{descriptive modelling}, or \textit{time series analysis}, a time series is modeled to determine its components in terms of seasonal patterns, trends, relation to external factors, and the like. … In contrast, time series forecasting uses the information in a time series (perhaps with additional information) to forecast future values of that series \cite{ptsr}
\end{quote}

\subsection{Time Series Analysis}
In classical statistics, the primary concern is the analysis of time series. \textit{Time series analysis} involves developing models that best capture or describe an observed time series so as to appreciate the underlying causes. This field of study seeks the \textbf{“why”} behind a time series dataset, the motives that influence its dynamics. 

This often involves making assumptions about the form of the data and decomposing the time series into constitution components.

The quality of a descriptive model is determined by how well it describes all available data and the interpretation it provides to better inform the problem domain.

\begin{quote}
"The primary objective of time series analysis is to develop mathematical models that provide plausible descriptions from sample data." \cite{tsar}
\end{quote}

\subsection{Time Series Forecasting}
Attempting to predict the future is called \textit{extrapolation} in the classical statistics. More modern fields concentrate on this task and refer to it as \textit{time series forecasting}. Forecasting involves attributing models that fit on historical data and using them to predict future observations.

An important difference between analyzing and forecasting is that in the latter the future is unavailable and one has no choice other than estimating it from what has already occurred. To state this in the machine learning context, let us put it this way: In supervised learning, estimating future can only be done by dividing the dataset into training and testing. And the model tries to learn the training data while capturing test data. Evaluating the performance on the test data ensures our precision in forecasting the future. 

\begin{quote}
"The purpose of time series analysis is generally twofold: to understand or model the stochastic mechanisms that gives rise to an observed series and to predict or forecast the future values of a series based on the history of that series." \cite{tsar2}
\end{quote} 

The adeptness of a time series forecasting model is determined by its performance at predicting the future. This is often at the expense of being able to explain \textbf{why} a specific prediction was made, (or what are the influencing factors (report confidence intervals and appreciating the underlying causes behind the problem.

The primary goal of this thesis is \textit{forecasting} stock prices and this is done by adding extra information other than the target stock.

Although we distinguished forecasting from describing, based on \ref{subsec:timeconcept} and imposing a philosophical perspective regarding this matter, it is not far-fetched to deduce that describing and forecasting are actually the same task because past, present, and future are inextricably intertwined and they all coexist.

\clearpage
\section{Time Series: A model space odyssey}

\begin{quote}
    "Artificial intelligence in finance has been a very popular topic for both academia and financial industry in the last few decades. Numerous studies have been published resulting in various models. Meanwhile, within the Machine Learning (ML) field, Deep Learning (DL) started getting a lot of attention recently, mostly due to its outperformance over the classical models. Lots of different implementations of DL exist today, and the broad interest is continuing. Stock market forecasting, algorithmic trading, credit risk assessment, portfolio allocation, asset pricing and derivatives market are among the areas where ML researchers focused on developing models that can provide real-time working solutions for the financial industry." \cite{rev00}
\end{quote}

If you seek a comprehensive and state-of-the-art categorization of Deep Learning models developed for financial applications, see \cite{rev00}. \\

We are, on the other hand, concerned with the usage of ANN models (which almost underlies the category of deep learning) as well as the task of stock price forecasting. Consequently, in what follows, previous studies will be presented from two perspectives, one of which accounts for the studies that have been carried out specifically for the task of stock market prediction and another one delves into works that have employed ANN methods for the task of time series forecasting. For the former perspective, which is introduced in \ref{sec:rev1}, \cite{rev0} was used as the reference and for the latter one, which is introduced in \ref{sec:rev2}, \cite{srev0} was used as that.

\subsection{Review of Stock Market Prediction Techniques}
\label{sec:rev1}

\begin{quote}
    "The advancements in stock price prediction have gained significant importance among expert analysts and investors. The stock market prediction for analyzing the trends is complicated due to intrinsic noisy environments and large volatility with respect to the market trends. The complexities of the stock prices adapt certain factors that involve quarterly earnings' reports, market news, and varying changing behaviors. The traders depend on various technical indicators that are based on the stocks, which are collected on a daily basis. Even though these indicators are used to analyze the stock returns, it is complicated to forecast daily and weekly trends in the market." \cite{rev1}
\end{quote}

\begin{quote}
    "The accurate prediction of stock trends is interesting and a complex task in the changing industrial world. Several aspects, which affect the behavior of stock trends, are non-economic and economic factors and which are taken into consideration. Thus, predicting the stock market is considered a major challenge for increasing production." \cite{rev2}
\end{quote}

\begin{quote}
    "Traditional techniques reveal that the stock market earnings are predicted from previous stock returns and other financial variables and macroeconomics. The prediction of stock market revenues directed the investors towards examining the causes of predictability. The forecasting of stock trends is a difficult process as it is influenced by several aspects, which involve trader's expectations, financial circumstances, administrative events, and certain aspects related to the market trends. Moreover, the list of stock prices is usually dynamic, complicated, noisy, nonparametric, and nonlinear by nature. \cite{rev3}"

\end{quote}

\begin{quote}
    The forecasting of financial time series becomes an issue due to certain complex features, like volatility, irregularities, noise, and changing trends \cite{rev4}.
\end{quote}

\begin{quote}
    "Various models applied for predicting the stock prices are managed using the time series models that involve Auto-Regressive Conditional Heteroscedastic (ARCH) model, Generalized Auto-Regressive Conditional Heteroskedasticity (GARCH), and Auto-Regressive Moving Average (ARMA). However, these models entail historical data and hypotheses like normality postulates. Several methods used for stock market prediction are based on conventional time series, such as fuzzy time series data, real numbers, and design of fuzzy sets. The fuzzy time series data are implemented for stock market prediction for handling linguistic value data for producing precise predicting results. These methods are widely used for forecasting nonlinear and dynamic datasets in the changing domains, such as tourism demand and stock markets." \cite{rev5}
\end{quote}

Many intelligent techniques, namely soft computing algorithms, Neural Network (NN), backpropagation algorithm, and Genetic Algorithm (GA), are applied for predicting the stock market returns. In \cite{rev6}, a prediction model was designed for predicting the stock trends with time series models.

In \cite{rev7}, GA and
NN is integrated for designing hybrid expert systems to make
the investment decisions. A technique based on GA is designed in \cite{rev8}, for feature discretization and determining the weights of Artificial Neural Networks (ANNs) \cite{rev9} for predicting the index of the stock price. time series model and NN were combined to predict the variability of the stock price in \cite{rev10}. The Artificial Intelligence (AI) techniques, like ANN, were devised for predicting the stock market prices. Many networks used feedforward neural networks for predicting the stock trends and evaluated multiple parametric and non-parametric models to forecast the stock market returns \cite{rev11}. Soft computing methods are utilized to deal with the AI for making the decisions using the profit and loss criterion. The techniques employed are fuzzy logic \cite{rev12}, Particle Swarm Optimization (PSO) \cite{rev13}, ANN \cite{rev14, rev15}, and Support Vector
Machine (SVM) \cite{rev16, rev17}. Several researchers tried to employ fuzzy-based techniques and randomness for optimizing the pricing models \cite{rev18,rev19}. In \cite{rev20}, the fuzzy-based techniques are employed for analyzing the market trends, and in \cite{rev21}, the performance of the fuzzy forecast is derived for estimating the initial values of stock price \cite{rev22}.

\cite{rev0} pinpoints main techniques and their corresponding papers that are used for stock prediction. Figure \ref{fig:stktec} illustrates the categorization of distinct stock market prediction techniques.

\begin{figure}[H]
    \centering
    \includegraphics[height=15cm]{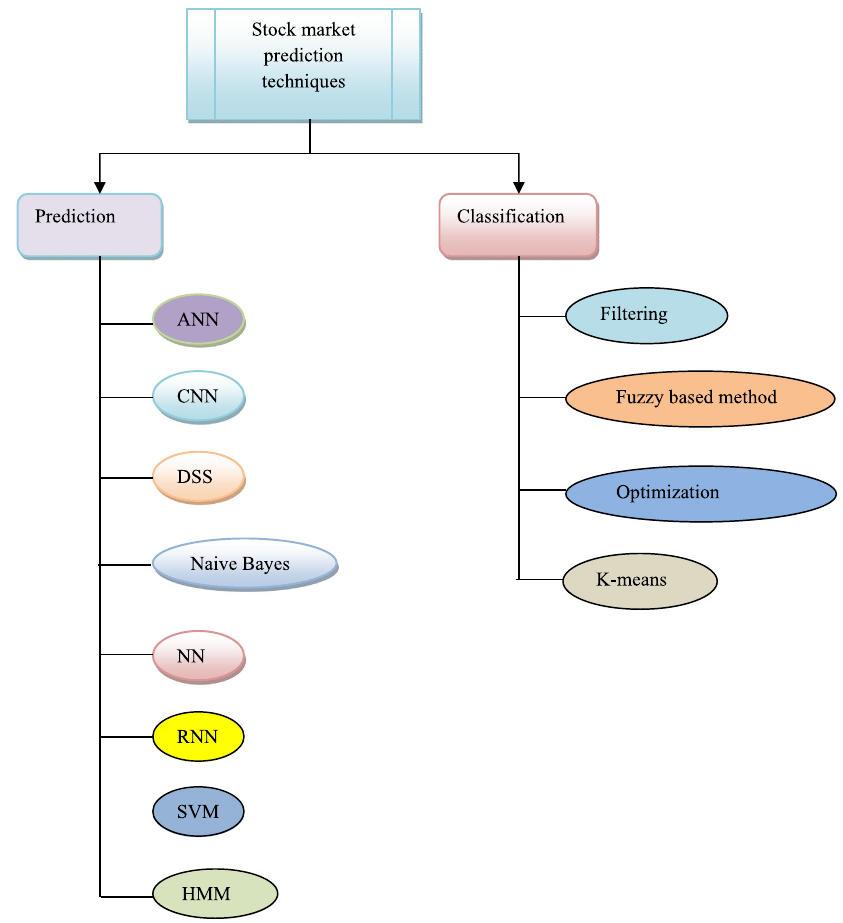}
    \caption{Categorization of distinct stock market prediction techniques by \cite{rev0}}
    \label{fig:stktec}
\end{figure}

Among techniques delineated in Figure \ref{fig:stktec}, works involving RNNs are the only ones chosen to be brought in the following since they are in line with this thesis's work. 

\subsubsection{RNN Based Prediction Techniques}
A recurrent neural network (RNN) is a class of artificial neural networks where connections between nodes form a directed graph along a temporal sequence. This allows it to exhibit temporal dynamic behavior. They are explained thoroughly in \ref{subsec:ann-rnn}.

The research works performing the stock market prediction by employing the RNNs are as follows: Hsieh, T.J et al. \cite{rev58} designed a united system, in which RNN based Artificial Bee Colony (ABC-RNN) algorithm are integrated for forecasting the stock prices. The model contains three phases, in which, initially, Haar wavelet is employed for decomposing the stock price time series data and to remove artifacts and noises. Secondly, the RNN was used for constructing the input features using Step-wise Regression-Correlation Selection (SRCS). Thirdly, ABC was adapted for optimizing the RNN weights and biases while designing the parameters. However, the method lacks advanced pattern selection mechanism for retrieving essential patterns from the data. Xie, X.K., and Wang, H \cite{rev59} designed RNN for analyzing the time series data. The dataset consists of the huge amount of intraday data from China Shanghai Shenzhen 300 Index. The RNN was used to classify the daily features using intraday data. The performance of the model was analyzed using precision and average profit. Chen, W et al. \cite{rev60} designed a model on the basis of RNN with Gated Recurrent Units (GRUs) for predicting the volatility of stocks in the Chinese stock market. The multiple price related features are subjected as an input to the model. However, the method failed to consider advanced machine learning method, like Interdependent Latent Dirichlet Allocation (ILDA), for making an accurate prediction. \cite{rev0}

\clearpage
\subsection{Review of Time Series Forecasting using Artificial Neural Networks Methods} \label{subsec:rvw-ts-ann}
\label{sec:rev2}
In what follows, previous works including Artificial Neural Networks for the special problem of time series forecasting will be investigated along with their deficiencies and requirements for further research. Thereafter, criteria containing questions will be introduced on which distinguishing papers that are applicable for further research are based. \\

\begin{quote}
    "Box and Jenkins, in the late 70s, made an important work in studying applications composed of mathematical linear models. These models represent Autoregressive (AR) and Moving Averages (MA) processes. While in the AR processes it is assumed that the current value of the time series is a linear combination of its past values, processes in the MA are supposed that the current value is a function of random interference passes or perturbations that have affected the series. Many practical experiments have demonstrated that this approach makes it possible to represent the dynamics of many real time series. That class models is popularized in both the academic and professional fields.\\
    However, it has also been found that many real time series seem to follow non-linear behavior and the approach of Box and Jenkins is insufficient to represent their dynamics \cite{srevv1}, \cite{srevv4}. Thus, in the most relevant literature have been presented a wide range of models that suggest different mathematical representations of the non-linearity present in the data \cite{srevv4}, \cite{srevv5}, such as the models based on schemes \cite{srevv4} and different types of Artificial Neural Networks (ANN) \cite{srevv6}, \cite{srevv7}, \cite{srevv8}. Some other literature reviews focused on one type of forecasting, one-step or multi-step ahead forecasting, and comparing the proposed strategies in theoretical and practical terms \cite{srevv33}." \cite{srev0}
\end{quote}

\begin{quote}
    "Particularly, the ANN has received considerable attention by the scientific community, which has been translated into a major investigative effort. Evidence of this is the large number of publications on this topic; it is as well, as a simple search in SCOPUS chains with “neural networks” and ”time series” found more than 4,000 reported documents; many of them might be irrelevant to the needs of the researcher. \\Despite the success of the ANN, and its persistence since the 90s, there exists several problems in proper model specification. This is evidenced by the fact that most of the decisions in the process specifications are subjective and are based on the experience of the modeler. Thus, seemingly, has not been fully reached systematic procedure for formal specification of ANN models" \cite{srevv10}
\end{quote}

To put it differently, most of the papers are solely confined to report result and performance without pursuing and being faithful to a general strategy. Scarcely have they pointed out the procedure followed for selecting methods or parameters provided that there was a non-empirical and methodical approach (either automated or not) in doing so. 

In addition to lack of a strategy for parameter selection or designing a network, another existing issue is an undeniable gap between theory and practice. Seldom, if ever, has a study provided a comprehensive theoretical and mathematical explanation for describing the neural network that has been used and this is somewhat natural and not expected since Neural Networks are still partly considered black boxes that are yet to be unraveled, formalized and modeled through a mathematical framework. In order to take first steps for removing this gap, attempts are made in section \ref{sec:nnmf} for explanation of the mathematical framework on which Neural Networks are based. 


\subsubsection{Quality Questions}
\cite{srev0} introduces criteria from which future works can inspire and apply to distinguish papers that satisfy the requirements and overcome deficiencies mentioned earlier to some extent (depending on the paper) and also, use the selected papers to reach the desired general strategy for designing a neural network model in the special case of the time series forecasting problem. The criteria comprise 7 quality questions which are brought subsequently:

\begin{enumerate}
\item Is there explicit mathematical formulation of the model?

\item Is the process of
estimating parameters for the new model defined?

\item Does the study
specify criteria for selecting the relevant variables?

\item Does the
study present a method for determining the appropriate complexity
(hyperparameters e.g., number of layers or neurons) of the model?

\item Is there an evaluation or
diagnosis for the model?

\item Does it examine the feasibility of the
proposed methodology by means of an application to a real case?

\item Does it specify the procedure for training the model (data
transformation, initial values of the parameters, stop criterion, etc.)?
\end{enumerate}

When using the search strings in the SCOPUS system, it
automatically recovered a total of 4021 publications. Then, \cite{srev0} manually applied the criteria of inclusion and exclusion, which led to the selection finally, a total of 17 articles .

the obtained results for 17 articles are presented
in table \ref{tbl1}, in which each row represents one of the selected studies. In Table \ref{tbl1}, columns 3-9 show the degree of compliance with the criteria of quality questions defined by the Q1-Q7 (Y stands for yes, P stands for partially and N stands for no); the column 10 collects the total score obtained in relation to the questions to Q1-Q7. \\



\begin{table}[H]
\centering
\begin{tabular}{ |c|c|c|c|c|c|c|c|c| }
 \hline
 ID & Q1 & Q2 & Q3 & Q4 & Q5 & Q6 & Q7 & Score\\
 \hline
 S1 & Y & Y & Y & Y & Y & Y & Y & 7.0\\
 S2 & P & Y & N & Y & N & Y & P & 4.0\\
 S3 & Y & N & N & N & Y & Y & Y & 5.0\\
 S4 & Y & N & N & N & Y & P & P & 3.0\\
 S5 & Y & Y & P & N & Y & Y & P & 5.0\\
 S6 & Y & Y & Y & Y & Y & Y & Y & 7.0\\
 S7 & Y & P & N & N & P & Y & N & 3.0\\
 S8 & Y & P & N & N & P & P & P & 3.0\\
 S9 & Y & P & Y & N & Y & Y & Y & 4.5\\
 S10 & P & N & P & N & Y & Y & P & 3.5\\
 S11 & Y & Y & P & N & Y & Y & Y & 4.5\\
 S12 & P & P & N & N & Y & P & P & 3.0\\
 S13 & Y & Y & N & N & Y & Y & Y & 5.0\\
 S14 & N & N & N & N & Y & Y & P & 2.5\\
 S15 & P & N & P & N & Y & Y & P & 3.5\\
 S16 & Y & Y & P & N & Y & Y & P & 5.0\\
 S17 & Y & P & P & Y & Y & Y & P & 5.5\\
 
\hline
\end{tabular}
\caption{Quality Assessment by \cite{srev0}}
\label{tbl1}
\end{table}

\cite{srev0} listed a table of 17 articles along with their conclusions which we present here in table \ref{tbl2}. \\

What we did was to extract information from the selected articles and summarise them. The extracted information are closely tied to steps of time series forecasting and future works can apply it to an arbitrary problem. It should be noted that the information we extracted is in line with the quality questions. In order to accomplish the extraction, at the end of the section essential information of each paper is stated in the following manner: \textbf{Data, Preprocessing, Parameters and Hyperparameters, Loss Measures} \\
The aforementioned terms are pivotal concepts in steps of a time series forecasting task.

\clearpage


\begin{center}
\begin{longtable}{|p{0.7cm}|p{1cm}|p{2cm}|p{9.5cm}|}
\caption{Description of Selected Methods by \cite{srev0}}
\label{tbl2} \\

\hline \multicolumn{1}{|c|}{\textbf{ID}} & \multicolumn{1}{c|}{\textbf{Year}} & 
\multicolumn{1}{c|}{\textbf{Method}} & \multicolumn{1}{c|}{\textbf{Conclusion}} \\ \hline 
\endfirsthead

\multicolumn{4}{c}%
{{\tablename\ \thetable{} -- continued from previous page}} \\
\hline \multicolumn{1}{|c|}{\textbf{ID}} &
\multicolumn{1}{c|}{\textbf{Year}} 
& \multicolumn{1}{c|}{\textbf{Method}} & 
\multicolumn{1}{c|}{\textbf{Conclusion}} \\ \hline 
\endhead

\hline \multicolumn{4}{|c|}{{Continued on next page}} \\
\hline
\endfoot

\hline \hline
\endlastfoot
 
S1 \cite{srev1} & 2006 & AR-NN & 
It presents and tests a non-linear disturbance model based on neural networks. This is successfully applied to the simulation of rainfall in a watershed.\\
\hline
S2 \cite{srev2} & 2007 & NLPMANN & 
It presents and test a non-linear disturbance model based on neural networks. This is successfully applied to the simulation of rainfall in a watershed.\\
\hline
S3 \cite{srev3} & 2007 & - & 
It proposed a hybrid model based on the hidden Markov model (HMM), ANN and Genetic Algorithms (GA), to predict the behavior of the financial market.\\
\hline
S4 \cite{srev4} & 2008 & AR*-GRNN &
It considered a combined model, consisting of an AR* model and a generalized regression neural networks model (GRNN, by its initials). The results indicate that the method is effective to combine the time-series models with models of neural networks, taking the advantages of decade model.\\
\hline
S5 \cite{srev5} & 2008 & - & 
It proposes a new hybrid method based on the basic concepts of ANN and the fuzzy regression models, which produces more precise results for incomplete datasets.\\
 \hline
S6 \cite{srev6} & 2009 & - & 
It is considered an Evolutionary Artificial Neural Network model (EANN) to automatically construct the architecture and the connections of the weights of the neural network.\\
\hline
S7 \cite{srev7} & 2010 & - & 
It presents a new hybrid ANN model, u sing an ARIMA model to obtain forecasts more accurate than the model of neural networks. In the first stage of the hybrid methodology, fits an ARIMA model and, in the second, is taken as network inputs the residuals of the ARIMA model and the original data.\\
\hline
S8 \cite{srev8} & 2010 & ADNN & 
It is considered a new ANN model with adaptation metrics for the entries, and with a mixing mechanism for the outputs of the network.\\
\hline
S9 \cite{srev9} & 2010 & AWNN Time series & 
Propose a hybrid model based on an adaptive Wavelet Neural Network (AWNN) and time series models, such as the ARMAX and GARCH, to predict the daily value of electricity in the market. They show better results in the forecasts than those reported in the literature.\\
\hline
S10 \cite{srev10} & 2011 & GEFTS-GRNN &
It features and ensembles a regression neural network model to predict widespread time series, which is a hybrid of different algorithms for machine learning. With this model combines the advantages joint that represent the algorithms, but has a high computational cost.\\
\hline
S11 \cite{srev11} & 2011 & ANN-SA & 
It is considered an ANN model whose training is done by means of the simulated annealing algorithm (SA). This model is proposed to study some characteristics related to seismic events, using experimental data.\\
\hline
S12 \cite{srev12} & 2012 & NWESN, BAESN, MESN & 
It Proposes three new neural network models based on Echo State Network (ESN) and with theory of complex networks. Although these models have a structure more complex than the original model ESN. They show that they yield more accurate forecasts.
\\
\hline
S13 \cite{srev13} & 2013 & L and NL-ANN & 
It presents a new model that considers at the same time the ANN of linear and nonlinear time series structures. The network is trained using the Particles Swarm Optimization algorithm (PSO). The authors show that the proposed model allows you to obtain better results compared to some traditional models.\\
\hline
S14 \cite{srev14} & 2013 & ARIMA-ANN & 
It is considered a hybrid model between ANN and the ARIMA model, in order to integrate the advantages of both models. The new model was tested in three databases, yielding good results.\\
\hline
S15 \cite{srev15} & 2014 & SVR-NN-GA & 
Proposes a hybrid approach wavelet denoising (WD) techniques, in conjunction with artificial intelligence optimization based SVR and NN model. The computational results reveal that cuckoo search (CS) outperforms both PSO and GA with respect to convergence and global searching capacity, and the proposed CS-based hybrid model is effective and feasible in generating more reliable and skillful forecasts.\\
\hline
S16 \cite{srev16} & 2015 & 3D hydrodynamic, ANN, ARMAX & 
A 3D hydrodynamic model, ANN model (BPNN), AR with exogenous inputs, ARMAX) model, and a combined hydrodynamic and ANN model integrated to more accurately predict water level fluctuation. The proposed concept of combining three-dimensional hydrodynamic model used in conjunction with an ANN is a novel one model which has shown improved prediction accuracy for the water level fluctuation.\\
\hline
S17 \cite{srev17} & 2016 & ERNN-STNN & 
Propose a hybrid model based on Elman recurrent neural networks (ERNN) with stochastic time effective function (STNN), the empirical results show that the proposed neural network displays the best performance between linear regression, complexity invariant distance (CID), and multi-scale CID (MCID) analysis methods and compared with different models such as the back-propagation neural network (BPNN), in financial time series forecasting.

\end{longtable}
\end{center}

\clearpage


\subsubsection{Information Extraction}
 \begin{enumerate}
     \item \textbf{Data} \\
         \textit{S17 \cite{srev1} (2016)}: Stock Market Indices
         
         \textit{S16 \cite{srev16} (2015)}: Water Level Fluctuation
         
         \textit{S15 \cite{srev15} (2014)}: Wind Speed
         
         \textit{S14 \cite{srev14} (2013)}: Economic
         
         \textit{S13 \cite{srev13} (2013)}: A mount of carbon dioxide measured monthly in Ankara capital of Turkey (ANSO). In addition, it uses well-known datasets such as Australian beer consumption or Logarithmic Canadian Lynx data
         
         \textit{S12 \cite{srev12} (2012)}: Simulated
         
         \textit{S11 \cite{srev11} (2011)}: Peak time-domain characteristics of strong ground-motions
         
         \textit{S10 \cite{srev10} (2011)}: 30 real datasets
         
         \textit{S9 \cite{srev9} (2010)}: Electricity Prices
         
         \textit{S8 \cite{srev8} (2010)}: Chaotic time-series based on the solution of Duffing equation
         
         \textit{S7 \cite{srev7} (2010)}: Three well-known data sets---the Wolf's sunspot data, the Canadian lynx data, and the British pound/United States dollar exchange rate data

     \item \textbf{Preprocessing} \\
         \textit{S17 \cite{srev1} (2016)}: data should be properly adjusted and normalized at the beginning of the modelling in $[0,1]$ range via the following formula:
            
            $$S(t)' = \frac{S(t) - min S(t)}{max S(t) - min S(t)}$$

         \textit{S15 \cite{srev15} (2014)}: empirical mode decomposition (EMD) method and wavelet transform (WT) or wavelet denoising (WD) method.
         Furthermore, two techniques were used:
            \begin{enumerate}
                \item \textit{5-3 Hanning filter (5-3H)}:
                first, a 5-point moving median average smoothing.
                second, a 3-point moving average smoothing
                and finally Hanning moving average smoothing is done by the following formula:
               $$u(i) = \frac{z(i-1) + 2z(i) + z(i+1)}{4}$$
                
                \item \textit{Wavelet Denoising}: 
                Wavelet transform method has been extensively applied recently in analyzing a nonstationary and high fluctuant series. It decomposes the original complicated data into several components of wavelet transform, one of which is smooth and reflects the inherent and real information.
            \end{enumerate}
            
         \textit{S14 \cite{srev14} (2013)}: Data Normalization between $[0,1]$ via the following formula: 
            $$x\mbox{*}_t = (x_t - x_{min}) / (x_{max} - x_{min})$$
            
         \textit{S11 \cite{srev11} (2011)}: Excluded some data from the analysis based on some of the filtering strategies proposed by Boore and Atkinson. Duplicate and missing information were removed. Also, the following normalization method was applied:
         $x_n = ax + b$ where $a = \frac{U - L}{X_{max} - X_{min}}, b = U - a X_{max}$. In the study, $L = 0.05$ and $U = 0.95$
            
         \textit{S10 \cite{srev10} (2011)}: Detrending and deseasonalization - make time series stationary by log or square and also differencing. 
    
         \textit{S9 \cite{srev9} (2010)}: Make time series stationary by differencing. Use logarithmic return instead of price
             
         \textit{S7 \cite{srev7} (2010)}: Make time series stationary by differencing and power transformation

    \item \textbf{Parameter or Hyperparameter selection procedure} \\
         \textit{S17 \cite{srev1} (2016)}:
         Experiments done repeatedly on the different index data, different number of neural nodes in the hidden layer were chosen as the optimal (manually).
         
         The details of principles of how to choose the hidden number are as follows: If the number of neural nodes in the input layer is $N$, the number of neural nodes in the hidden layer is set to be nearly $2N + 1$, and the number of neural nodes in the output layer is 1.
         
         \textit{S16 \cite{srev16} (2015)}:
         The best network architecture (number of hidden nodes, number of iterations, learning rate, and momentum coefficient) was obtained by trial and error based on RMSE in the training and validation phases.
         
         \textit{S15 \cite{srev15} (2014)}: ARIMA parameter tuning: Particle Swarm Optimization (PSO)
         SVR hyperparameter tuning: cuckoo search which is compared with grid search (GS) and two conventional optimal algorithms (GA and PSO).
        
         \textit{S14 \cite{srev14} (2013)}: 
         \textit{ARMA}: ACF and PACF \\
         \textit{Neural Network}: in practice, experiments are often conducted to select the appropriate values p (dimension of input vector - the lagged observations) and q (number of hidden nodes). selected the ‘best' models according to delineated criteria, which include the root mean squared error (RMSE) and MAE, AIC and BIC.
    
         \textit{S13 \cite{srev13} (2013)}: NN: Particle Swarm Optimization 
         
         \textit{S11 \cite{srev11} (2011)}: Optimization Algorithm: Annual sealing \\ Recursive equation for updating the weights
    
         \textit{S10 \cite{srev10} (2011)}: BIC, uses PSO and SAGA for optimization. SAG (stochastic average gradient) was used for selecting the best feature subspace of univariate time series

         \textit{S9 \cite{srev9} (2010)}: ARMA: ACF and PACF and ordinary least squares (OLS) 
             
         \textit{S8 \cite{srev8} (2010)}: Optimization Algorithm: Levenberg - Marquardt
         
         \textit{S7 \cite{srev7} (2010)}: ARMA: The parameters are estimated such that an overall measure of errors is minimized.
        
         \textit{S6 \cite{srev6} (2009)}:
         SCGA algorithm \\
        
     \item \textbf{Loss Measure} \\
         \textit{S17 \cite{srev1} (2016)}: MAP, RMSE, MAPE
        
         \textit{S16 \cite{srev16} (2015)}: MAE, RMSE, R, and SS
         
         \textit{S15 \cite{srev15} (2014)}: MSE, MAE, MAPE, SMPAE
         
         \textit{S14 \cite{srev14} (2013)}: MD, SD, MAD, MSE, MAPE
    
         \textit{S13 \cite{srev13} (2013)}: MSE, RMSE, MAPE, MdAPE, DA 
        
         \textit{S12 \cite{srev12} (2012)}: NRMSE
         
         \textit{S11 \cite{srev11} (2011)}: R (Pearson), MAE, MAPE, MSE
         
         \textit{S10 \cite{srev10} (2011)}: MAPE
        
         \textit{S9 \cite{srev9} (2010)}: AMAPE, used MSE to avoid overfitting
         
         \textit{S8 \cite{srev8} (2010)}: MAPE, MSE, NMSE
         
         \textit{S7 \cite{srev7} (2010)}:
         MAE, MSE
     
     \item \textbf{Additional} \\
         \textit{S17 \cite{srev1} (2016)}: Horizon: 1 year
         
         \textit{S14 \cite{srev14} (2013)}: 
         Decomposes time series in the following manner; Assume two models to analyse such a time series, an additive model $(L + N)$ and a multiplicative model $(L * N)$. multiplicative model was superior to the ARIMA model, the ANN model and the additive model.
         
         \textit{S13 \cite{srev13} (2013)}: Decompose time series into linear and nonlinear by addition. \\
         SARIMA and Winters' multiplicative exponential smoothing (WMES) approaches are the linear conventional methods used in the implementation.
         
         \textit{S12 \cite{srev12} (2012)}: uses NARMA and NARX as linear models
        
         \textit{S9 \cite{srev9} (2010)}: ARMAX model is used to catch the linear relationship between price return series and explanatory variable load series. GARCH model is used to unveil the heteroscedastic character of residuals. AWNN is used to present the nonlinear, nonstationary impact of load series on electricity prices
         
         \textit{S6 \cite{srev6} (2009)}: Includes the genetic algorithm and the scaled conjugate gradient algorithm

\end{enumerate}

\vfill

\section{Conclusion and Thesis Outline}

In this chapter, at first we reviewed the works done for the task of stock price prediction and specifically, works that use RNN networks for this task. Then, we pinpointed the research gaps in the papers that use ANNs for the task of time series forecasting. We subsequently distinguished practical articles based on introduced criteria and extracted information form them so as to facilitate future works. The extracted information is pivotal in the steps that should be taken for a time series forecasting task. Hopefully, this will help to achieve a general strategy so that time series forecasting can be done more objectively instead of a problem-oriented approach which is solely based on the researcher's subjective experience. The rest of this thesis is organized as follows:

In chapter \ref{chap2}, we will provide a rather concise introduction to deep learning and neural networks, and then we will narrow our scope to recurrent neural networks (RNNs) and more specifically, Gated RNNs. Eventually, we will elaborate on a particular kind of Gated RNN, namely, LSTM, which is the primary focus of this thesis both in practice (chapter \ref{subsec:implmnt}) as well as theory (chapter \ref{sec:nnmf}).

In chapter \ref{subsec:implmnt}, we attempted forecasting one-step ahead of Goldman Sachs's (GS) prices as well as General Electric (GE) by implementing LSTM. Two other stocks that are correlated with GS, namely, JPMorgan and Morgan Stanley were added as feature to GS, Also, auxiliary features were added to improve the model's accuracy. In addition, the ARMA (auto-reggressive integrated moving average) model is also applied so as to serve as a benchmark. Our empirical results indicated the challenges when attempting to forecast by LSTM. One imperative upheaval is the forecasting lag problem.

In chapter \ref{sec:nnmf}, we introduce the formalization of neural networks through a mathematical framework, which is introduced by \cite{dnnmf}. At first, the framework is developed for a generic Neural Network and gradient descent algorithm is expressed within that framework. Then, this framework is extended to specific architectures of Neural Networks, from which we chose RNNs to explain and express gradient descent within their structure. \cite{dnnmf} proves theorems regarding RNNs for the case of classification case and cite some of the theorems to those used in feedforward networks. On the other hand, we proved the theorems for the regression case (which includes forecasting), and proved theorems of RNNs as well as expressing gradient descent algorithm independently from feedforward networks.
\clearpage

\clearpage
\chapter{Introduction to Deep Learning and Neural Networks} \label{chap2}
\clearpage
In the previous chapter, we reviewed the works done on Stock Market Forecasting as well as those that employed artificial neural network (ANN) methods for time series Forecasting. In this chapter, we will provide a concise introduction to Deep Learning. At first, we explain the general notion of \textit{deep learning} by pinpointing two primary aspects of it, which are the representation of data and the depth of model. In the former, we stress the profound extent that performance of the models rely on the representation of data and in the latter, we explain the effect of depth in the models. \\
Afterward, we delineate the position of neural networks in deep learning.\\
In the last section, after providing a description of neural networks, we will narrow our scope to recurrent neural networks (RNNs) so as to explain a special kind of them, which is called Gated RNNs. Eventually, we will elaborate on a particular case of gated RNNs, namely, LSTM through a step-by-step walk-through.

\section{Representation of Data} \label{subsec:rep}

\begin{quote}
"The performance of simple machine learning algorithms depends heavily on the representation of the data they are given. For example, when logistic regression is used to recommend cesarean delivery, the AI system does not examine the patient directly. Instead, the doctor tells the system several pieces of relevant information, such as the presence or absence of a uterine scar. Each piece of information included in the representation of the patient is known as a \textit{feature}. Logistic regression learns how each of these features of the patient correlates with various outcomes. However, it cannot influence how features are defined in any way. If logistic regression were given an MRI scan of the patient, rather than the doctor's formalized report, it would not be able to make useful predictions. Individual pixels in an MRI scan have negligible correlation with any complications that might occur during delivery. This dependence on representations is a general phenomenon that appears throughout computer science and even daily life. In computer science, operations such as searching a collection of data can proceed exponentially faster if the collection is structured and indexed intelligently. People can easily perform arithmetic on Arabic numerals but find arithmetic on Roman numerals much more time-consuming. It is not surprising that the choice of representation has an enormous effect on the performance of machine learning algorithms. For a simple visual example, see figure \ref{fig:descaresvspolar}. Many artificial intelligence tasks can be solved by designing the right set of features to extract for that task, then providing these features to a simple machine learning algorithm. For example, a useful feature for speaker identification from sound is an estimate of the size of the speaker’s vocal tract. This feature gives a strong clue as to whether the speaker is a man, woman, or child." \cite{dl}
\end{quote}

Detecting the right set of features is not a straightforward task though. For the purpose of illustration, consider the task of car detection in photographs. if we are intent on considering the presence of the wheel of the car as feature, we will confront difficulty describing the wheel in terms of pixel values. Even though a wheel has a simple geometric shape, its corresponding image may be distorted due to falling shadows and sun glares or the wheel's visibility might be limited because of the car's fender or another object, and so forth.

One way to deal with this problem is the approach of \textit{representation learning} in which the ML algorithm finds not only the mapping from representation to output but also the representation itself (a typical example of this approach is \textit{autoencoders}).

In order to learn the features, one should distinguish factors of variations, i.e. sources that influence data that are either observed or unobserved. They can also be abstract concepts in the human mind that helps him to make sense of the data and infer the underlying causes of it. When analyzing a speech recording, the factors of variation include the speaker's age, their sex, their accent and the words they are speaking. When analyzing an image of a car, the factors of variation include the position of the car, its color, and the angle and brightness of the sun.

A predicament that arises in real-world applications is that factors of variation may influence every single piece of data and so disentangling them and discarding those we do not care about will be quite challenging. For instance, the individual pixels in an image of a red car might be very close to black at night. The shape of the car's silhouette depends on the viewing angle. Another example of this challenge is in recognizing facial expression. Two images of different individuals with the same visage are separated effectively in pixel space. On the other hand, two images of the same individuals manifesting different visages may lie in a very close position in pixel space.
In the mentioned scenario there are two contributing factors, one of which is the identity of the individual and another one is the facial expression (visage). The former, however, is irrelevant to the desired task and may as well be discarded. However, it influences the representation of the image and hinders the process of disentanglement.

\begin{figure}
    \centering
    \includegraphics[width=.5\linewidth]{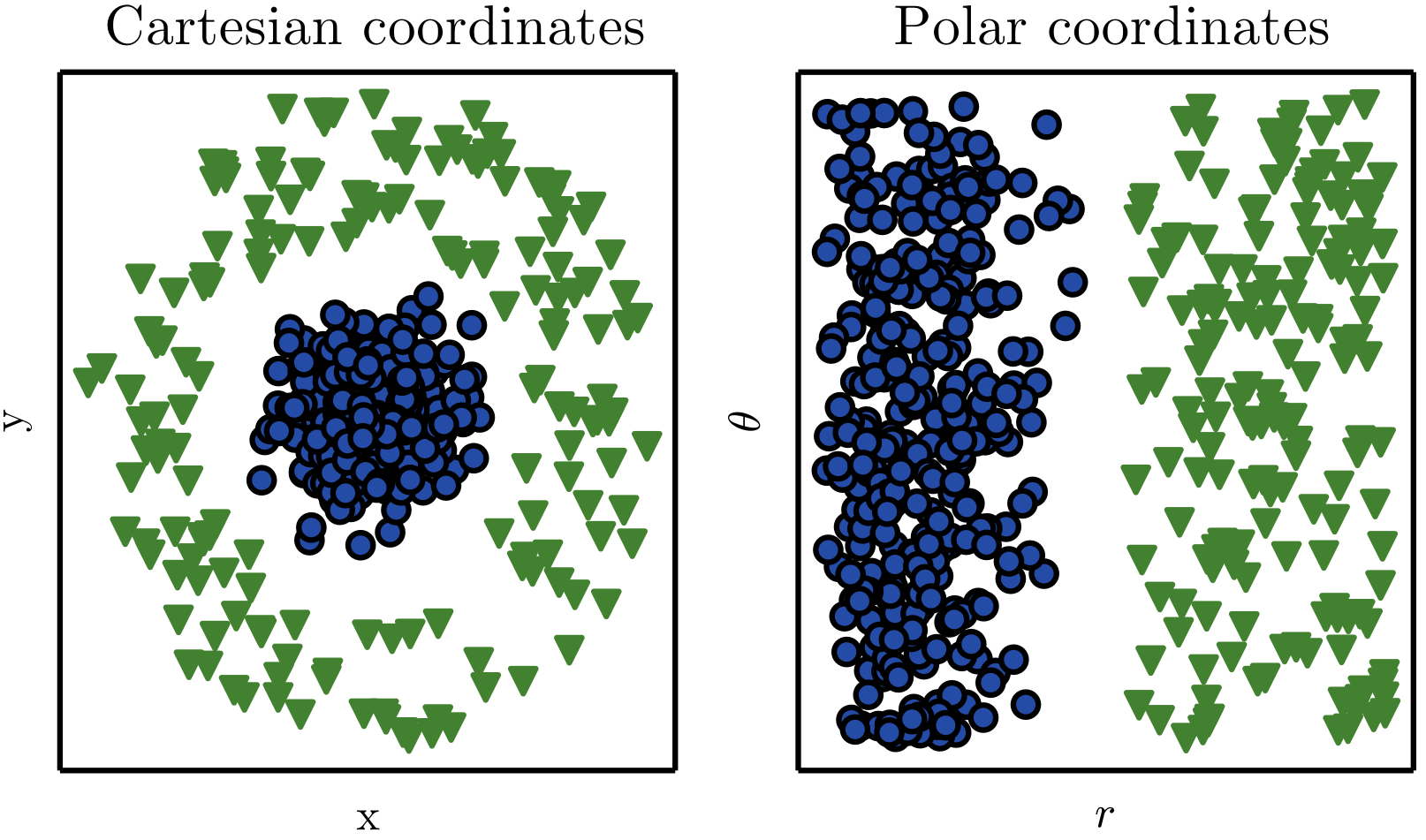}
    \caption{Example of different representations by \cite{dl}: Suppose we want to separate two categories of data by drawing a line between them in a scatterplot. In the plot on the left, we represent some data using Cartesian coordinates, and the task is impossible. In the plot on the right, we represent the data with polar coordinates and the task becomes simple to solve with a vertical line. Figure is from \cite{dl}}
    \medskip
    \small
    \label{fig:descaresvspolar}
\end{figure}

\begin{quote}
"\textit{Deep Learning} solves this central problem in representation learning by introducing representations that are expressed in terms of other, simpler representations. Deep learning enables the computer to build complex concepts out of simpler concepts." \cite{dl}
\end{quote}

By gathering knowledge from experience and improving this experience by data, not only does this approach avoids the need for human operators to explicitly state all the knowledge that the computer needs (advantage of machine learning over knowledge-based approach), but also it enables the model to learn about features on its own (advantage of deep learning over representation learning). The hierarchy of concepts enables the computer to learn complicated concepts by building them out of simpler ones. The graph that shows the way these concepts are stacked up on the top of each other, would be a deep one, consisting of many layers. For this reason, we call this approach to AI deep learning.

The quintessential example of a deep learning model is the feedforward deep network, or multilayer perceptron (MLP) which is explained in \ref{subsec:ann-fnn}. Figure \ref{fig:MLP} shows an MLP model that indicates how a deep learning system can represent the concept of an image of a person by combining simpler concepts, such as corners and contours, which are in turn defined in terms of edges.

\begin{figure}[H]
    \centering
    \includegraphics[width=.5\linewidth]{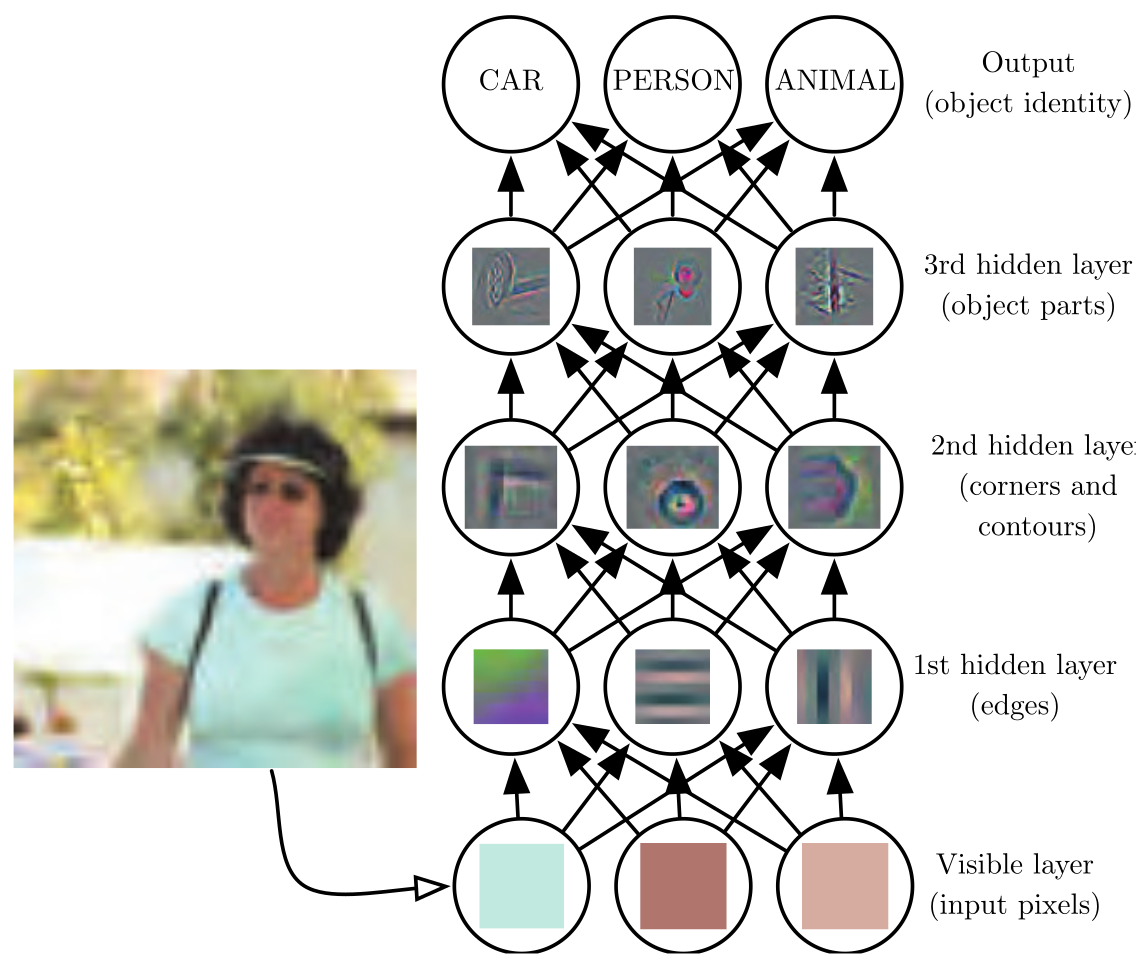}
    \caption{Illustration of a deep learning model by \cite{dl}. It is difficult for a computer to understand the meaning of raw sensory input data, such as this image represented as a collection of pixel values. The function mapping from a set of pixels to an object identity is very complicated. Learning or evaluating this mapping seems insurmountable if tackled directly. Deep learning resolves this difficulty by breaking the desired complicated mapping into a series of nested simple mappings, each described by a different layer of the model. The input is presented at the visible layer, so named because it contains the variables that we are able to observe. Then a series of hidden layers extract increasingly abstract features from the image. These layers are called "hidden" because their values are not given in the data; instead, the model must determine which concepts are useful for explaining the relationships in the observed data. The images here are visualizations of the kind of feature represented by each hidden unit. Given the pixels, the first layer can easily identify edges, by comparing the brightness of neighboring pixels. Given the first hidden layer's description of the edges, the second hidden layer can easily search for corners and extended contours, which are recognizable as collections of edges. Given the second hidden layer's description of the image in terms of corners and contours, the third hidden layer can detect entire parts of specific objects, by finding specific collections of contours and corners. Finally, this description of the image in terms of the object parts it contains can be used to recognize the objects present in the image. Images from \cite{dl}.}
    \medskip
    \small
    \label{fig:MLP}
\end{figure}

To wrap it up, we contend that machine learning is the only viable approach to building AI systems that can operate in complicated real-world environments. Deep learning is a particular kind of machine learning that achieves great power and flexibility by representing the world as a nested hierarchy of concepts, with each concept defined in relation to simpler concepts, and more abstract representations computed in terms of less abstract ones. Figure \ref{fig:DL} illustrates the relationship between these different AI disciplines which is cited from \cite{dl}.

\begin{figure}[H]
    \centering
    \includegraphics[width=.6\linewidth]{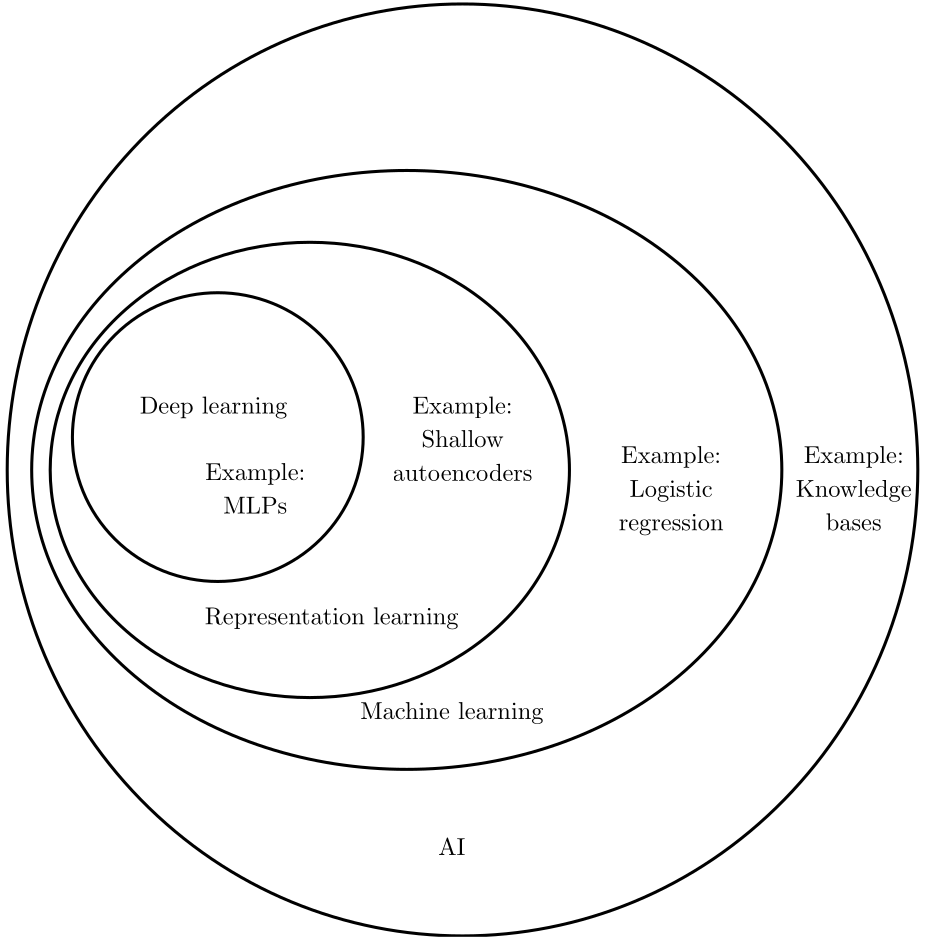}
    \caption{}
    \label{fig:DL}
\end{figure}

\subsection{Depth}
\begin{quote}
    
"It is noteworthy to state that the idea of learning the right representation for the data provides one perspective on deep learning. Another perspective on deep learning is that depth
enables the computer to learn a multistep computer program. Each layer of the representation can be thought of as the state of the computer's memory after executing another set of instructions in parallel. Networks with greater depth can execute more instructions in sequence. Sequential instructions offer great power because later instructions can refer back to the results of earlier instructions. according to this view of deep learning, not all the information in a layer's activations necessarily encodes factors of variation that explain the input. The representation also stores state information that helps to execute a program that can make sense of the input. This state information could be analogous to a counter or pointer in a traditional computer program. It has nothing to do with the content of the input specifically, but it helps the model to organize its processing." \cite{dl}
\end{quote}

\section{Position of Neural Networks in Deep Learning} \label{subsec:dl role}
A comprehensive history of deep learning is beyond the scope of this thesis. See \cite{dl} for more details. However, we will pinpoint some concise remarks.
Some of the earliest learning algorithms we recognize today were intended to be computational models of biological learning, that is, models of how learning happens or could happen in the brain. As a result, one of the names that deep
learning has gone by is artificial neural networks (ANNs). The corresponding perspective on deep learning models is that they are engineered systems inspired by the biological brain (whether the human brain or the brain of another animal).

Among the algorithms that are used nowadays in the machine learning universe, some stem from models that served as computational models of biological learning, i.e. models of how the learning process occurs in the brain. Consequently, an alternative term that encompasses deep learning is \textit{artificial neural networks} (ANNs). This term roots in the inspiration of the biological brain on learning models which is another perspective of deep learning.

\begin{quote}
    "The neural perspective on deep learning is motivated by two main ideas. One idea is that the brain provides a proof by example that intelligent behavior is possible, and a conceptually straightforward path to building intelligence is to reverse engineer the computational principles behind the brain and duplicate its functionality. Another perspective is that it would be deeply interesting to understand the brain and the principles that underlie human intelligence, so machine learning models that shed light on these basic scientific questions are useful apart from their ability to solve engineering applications." \cite{dl}
\end{quote}

One may wonder why deep learning has only recently become well-known even though the first experiments with artificial neural networks were conducted in the 1950s. There are two major contributing factors:

\begin{enumerate}
    \item {\textbf{Increasing Dataset Size}: The amount of skill required for achieving satisfactory performance on a deep learning problem reduces as the amount of training data increases. On the other hand, people are now spending more time on digital devices (laptops, mobile devices). Their digital activities generate huge amounts of data that we can feed to our learning algorithms.}

    \item {\textbf{Increasing Model Sizes}: One of the main insights within brain structure is that animals become intelligent when many of their neurons work together. ANNs have grown larger in size throughout in history. This increase in model size per time is due to the availability of faster CPUs, the advent of general purpose GPUs, faster network connectivity and better software infrastructure for distributed computing.}
\end{enumerate}

If you are more drawn to grasp the advantages of deep learning over traditional machine learning models, see challenges motivating deep learning in [\cite{dl}, Chapter5, Section11].

\section{Deep Learning and Prior Knowledge}

It turns out that the incorporation of prior knowledge, biasing the learning process, is inevitable for the success of learning algorithms. To fully understand this topic, we refer the reader to two well-known experiments, bait shyness and pigeon's superstition.
\begin{quote}
"Any time we choose a specific machine learning algorithm, we are implicitly stating some set of prior beliefs we have about what kind of function the algorithm should learn. Choosing a deep model encodes a very general belief that the function we want to learn should involve composition of several simpler functions.
This can be interpreted from a representation learning point of view as saying that we believe the learning problem consists of discovering a set of underlying factors of variation that can in turn be described in terms of other, simpler underlying
factors of variation. Alternately, we can interpret the use of a deep architecture as expressing a belief that the function we want to learn is a computer program consisting of multiple steps, where each step makes use of the previous step's output. These intermediate outputs are not necessarily factors of variation but can instead be analogous to counters or pointers that the network uses to organize its internal processing. Empirically, greater depth does seem to result in better generalization."
\cite{dl}
\end{quote}

The fewer beliefs we impose on the data, the more freedom  we will gain when performing tasks on data and the more complex data is, the fewer beliefs can be correctly imposed on it. In case of our concern, which is stock data, it may become complex to the point that it would be hard to distinguish a white noise (e.g., a random walk) from a stock's price return. Therefore, this fact suggests that we require a model with as less conditions and presumptions as possible. Therefore, one of the most suitable candidates would be neural networks, which can represent an arbitrary function with almost no assumption (according to universal approximation theorem.)

\clearpage
\section{Neural Networks} \label{subsec:nn-desc}

To embark on the path to introduce neural networks, we will pose beforehand the notion of \textit{perceptron}, which is a primary type of artificial neural network and then we will pose \textit{sigmoid neurons}. The relation of these notions and the formation of neural networks is explained thereafter. However, the precise mathematical formalization of neural networks is explained in \ref{sec:nnmf}.\\
In order to explain RNNs, we will elaborate on their recurrent nature (which is closely tied to involvement of self-loops) and the requirement for the existence of memory in model for modeling certain phenomena in real-world, such as predicting words or stocks. Afterward, we will pose the long term dependency challenge and propose approaches to overcome this challenge, such as leaky units, which would then bring us to the LSTM model as one of the most suitable approaches. \\
At the end of the chapter, we will describe the structure of an LSTM by introducing the concepts of cell-state and gates and will end this introduction by providing a step-by-step walk-through of LSTM.

\subsubsection{Perceptron} 
Perceptrons were developed in the 1950s and 1960s by scientist Frank Rosenblatt. Nowadays, it is more common to use alternative models of artificial neurons, perhaps the most common of which is the \textit{sigmoid neuron}. However, introducing perceptron beforehand will provide insight for understanding a sigmoid neuron.\\

A perceptron takes several binary inputs, $x_1,x_2,...,$ and produces a single binary output. See figure \ref{fig:ann-perceptron}.

\begin{figure}[H]
    \centering
    \includegraphics[width=.3\linewidth]{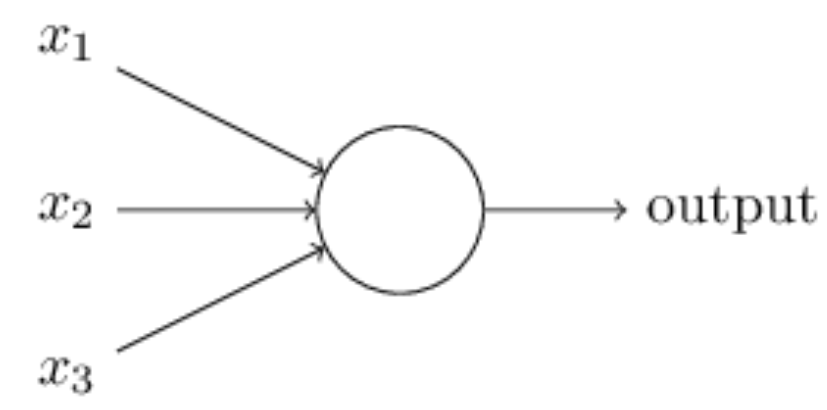}    \caption{A single perceptron with three inputs}
    \medskip
    \small
    \label{fig:ann-perceptron}
\end{figure}

Rosenblatt proposed a simple rule to compute the output. He introduced weights, $w_1,w_2,...$, real numbers manifesting the effect of the respective inputs the the output. The nueron's output, 0 or 1, is determined by whether $\sum_{j} w_j x_j$ is less than or greater than some \textit{threshold value}. More precisely:

\[
    \text{output}= 
\begin{cases}
    0, & \text{if} \sum_{j} w_j x_j \leq \text{threshold}\\
    1, & \text{if} \sum_{j} w_j x_j > \text{threshold}
\end{cases}
\]

An intuitive interpretation of perceptron is to consider perceptron a simplified binary decision-making process in which influencing factors are denoted by $x_1,x_2,...$, and the impact of each factor is expressed by $w_1,w_2,...$ respectively. If the expression $\sum_{j} w_j x_j$ exceeds the amount of threshold then the model outputs 1 and the decision will be made. Otherwise, the model outputs 0, and no decision is made since after weighing influences, they were not adequate for making the decision. 

In the relatively same manner, a complex network of perceptions would be capable of making more subtle decisions. For instance, in figure \ref{fig:ann-perceptron-network}, the network consists of three columns of perceptrons, also called \textit{layers}, each of which serves as a decision-maker by weighing up outputs of the previous layer. Adding layers lead to an increase in ability to make more complicated decisions. As mentioned in \ref{subsec:dl role}, larger networks are capable of achieving higher accuracy on more complex tasks due to involvement of more neurons as well as more interactions.

\begin{figure}[H]
    \centering
    \includegraphics[width=.7\linewidth]{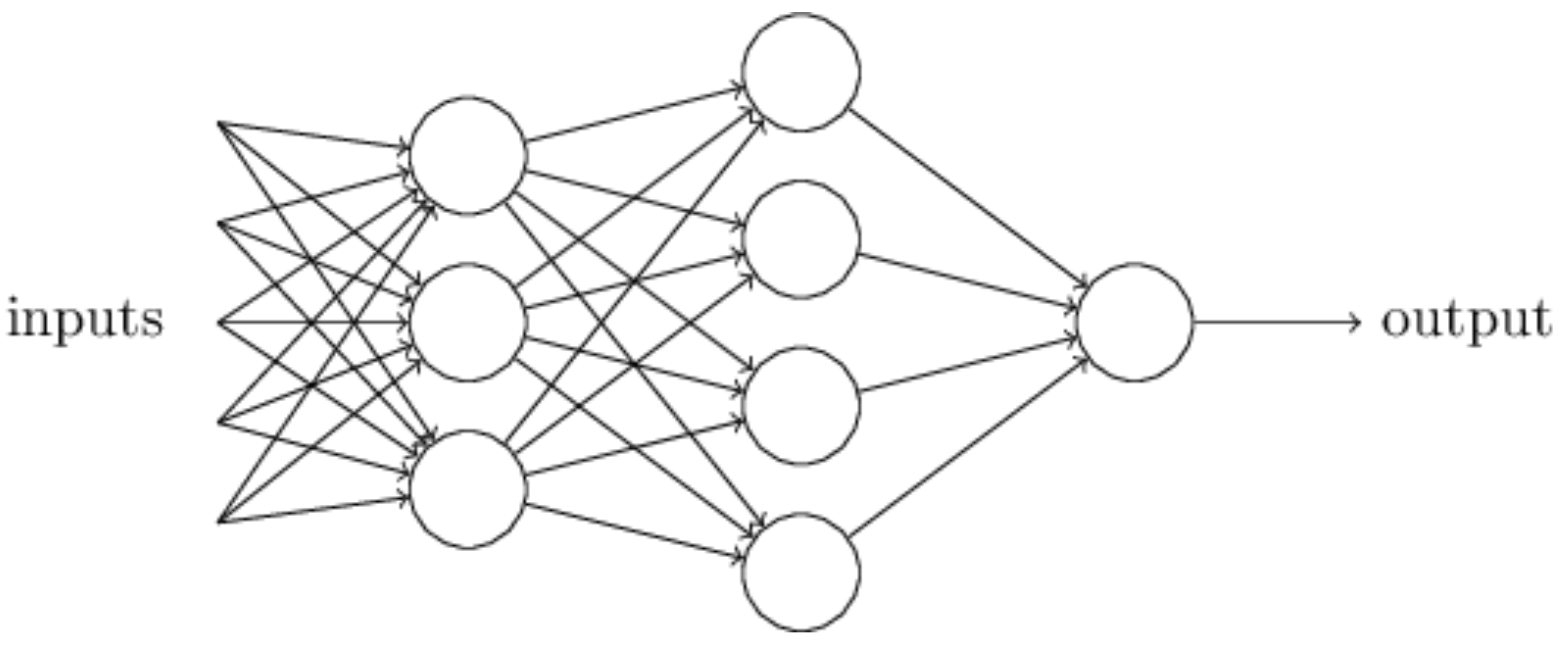}    \caption{A complex network of perceptrons}
    \medskip
    \small
    \label{fig:ann-perceptron-network}
\end{figure}

Even though in figure \ref{subsec:dl role} perceptrons are depicted as if they yield multiple outputs, in fact each perceptron outputs a single value that will serve as inputs for multiple perceptrons. However, it is confusing to draw a straight line that will then divide into multiple lines; so with bearing that in mind, we will still maintain this illustration.

An equivalent form of perceptron is this:
\[
    \text{output}= 
\begin{cases}
    0, & w \cdot x + b \leq 0 \\
    1, & w \cdot x + b > 0
\end{cases}
\]

In biological terms, the bias $b$ is a measure of the extent to which the perceptron is close to \textit{firing}. A big bias leads to the occurrence of outputting 1 being more likely and a small bias leads to the opposite.

\begin{quote}
"Another way perceptrons can be used is to compute
the elementary logical functions we usually think of as underlying
computation, functions such as \code{AND, OR}, and \code{NAND}. we can use perceptrons to compute simple logical functions. In fact, we can use networks of perceptrons to compute any logical function at all. The reason is that the \code{NAND} gate is universal for computation, that is, we can build any computation up out of \code{NAND} gates. It turns out that we can devise \textit{learning algorithms} which can automatically tune the weights and biases of a network of artificial neurons. This tuning happens in response to external stimuli, without direct intervention by a programmer. 
conventional logic gates." \cite{dl2} 
\end{quote}

In order to transform a perceptron to a model that would learn from data and solve problems, what required is tuning the weights and biases in the network in a way that produces results close to our desired network. However, we would prefer that the network would be not sensitive to perturbation, i.e. small changes in inputs so that the output does not differ much. Otherwise, if the output undergoes a huge change, then we will not be able to readily improve the network's learning since we will lose sight of the impact of each variable (factor) on the output. See figure \ref{fig:ann-perceptron-network-perturb}.

\begin{figure}[H]
    \centering
    \includegraphics[width=.7\linewidth]{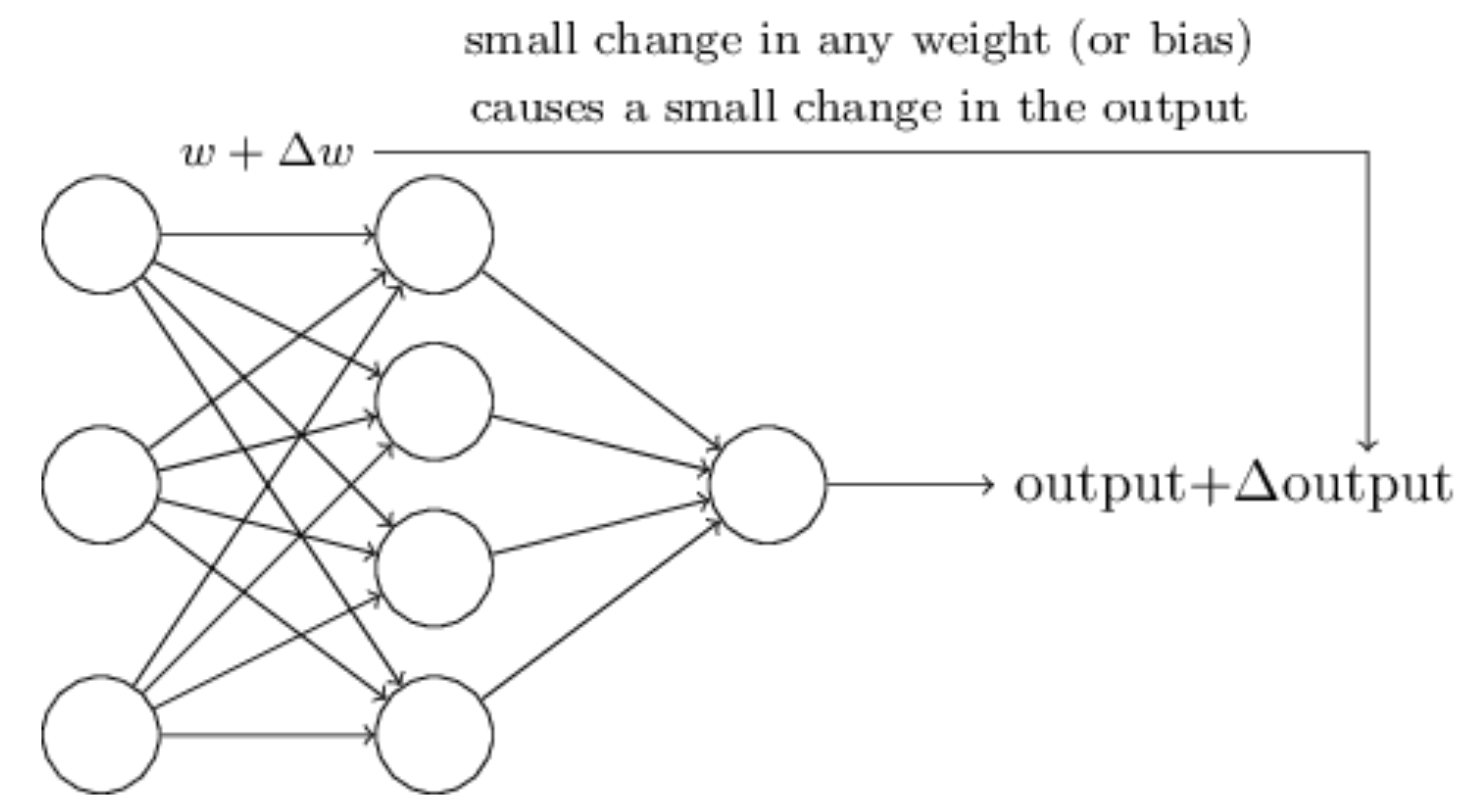}    \caption{A desired network does not manifest sensitivity to perturbation.}
    \medskip
    \small
    \label{fig:ann-perceptron-network-perturb}
\end{figure}

The main drawback of perceptrons is that they are intrinsically sensitive to input and a small change may sometimes shift their output from 0 to 1 and backward. One way to overcome this issue is to use \textit{sigmoid neurons}.

\subsubsection{Sigmoid Neurons} 
Akin to perceptrons, a sigmoid neurons receive inputs with the difference that its output can take any value between 0 and 1. In effect, the output is $\sigma(w \cdot x + b)$ where $\sigma$ is the sigmoid function which is equal to $$\sigma (z) = \frac{1}{1 + e^{-z}} = \frac{1}{1 + exp(- \sum_{j} w_j x_j -b)}.$$

Although perceptron and sigmoid neuron appear different they behave not that dissimilar. To elucidate this, suppose $z = w \cdot x + b$ is a large positive number. Then, $e^{-z}$ approaches zero and so does $\sigma(z)$. Therefore, when the input($z$) is large and positive, the output of the sigmoid neuron will be approximately 1 which was also the case in perceptron. Likewise, if $z$ is a large negative number, the output will be 0. In conclusion, perceptron and sigmoid neurons behave similarly in extreme numbers. 

We can detect the smoothness of the sigmoid function in \ref{rnn-sigmoid-func}. The smoothness of the sigmoid function would affect the structure of output such that that small perturbation ($\Delta w_j$/ in the weights and $\Delta b$/ in the bias) will result in small change $\Delta output$ in output. More precisely:

$$\Delta output = \sum_{j} \frac{\partial output}{\partial w_j} \Delta w_j + \frac{\partial output}{\partial b} \Delta b,$$

where the sum is over all weights $w_j$, and $\partial output / \partial w_j$ and $\partial output / \partial b$ denote partial derivatives of the output with respect to $w_j$ and $b$, respectively. 

\begin{figure}[H]
    \centering
    \includegraphics[width=.7\linewidth]{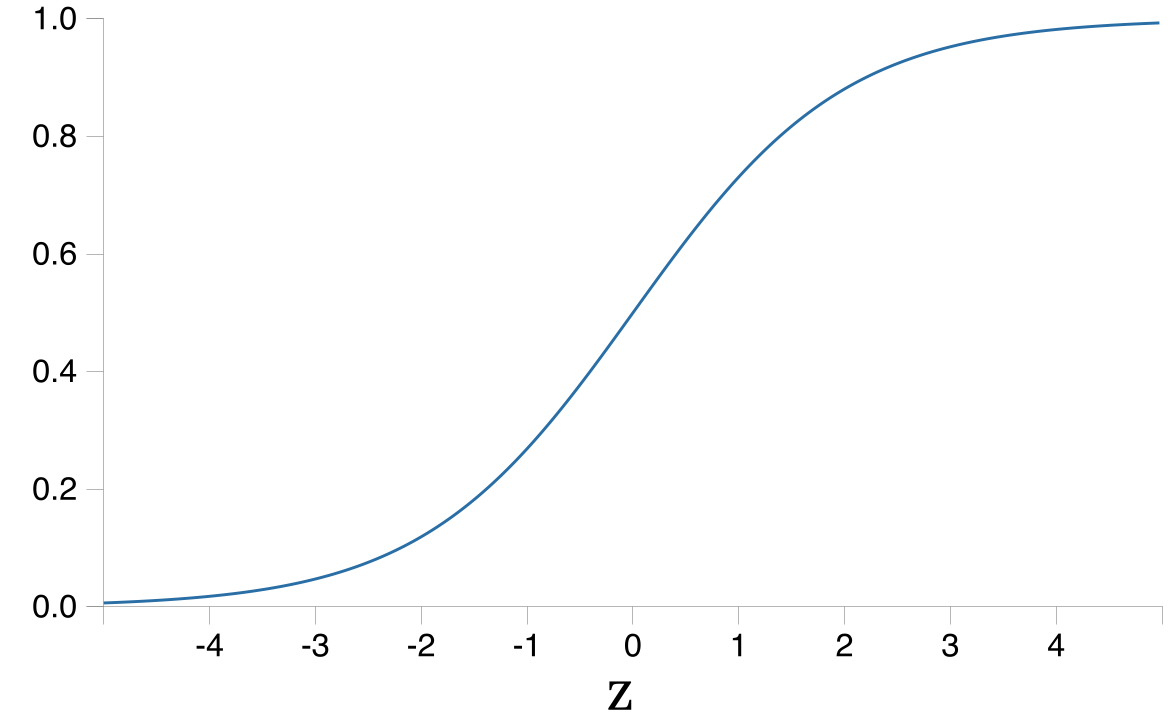}    \caption{Sigmoid Function}
    \medskip
    \small
    \label{rnn-sigmoid-func}
\end{figure}

One may wonder how to interpret the output of a sigmoid neuron. In the case of regression, the output can take any number. In the case of classification (detecting handwritten numbers for instance), however, the output is binary and either 0 or 1. In order to deal with this problem, it is conventional that if the output of the sigmoid neuron is less than 0.5, then it will be considered 0 and otherwise, it will be considered 1.

The smooth function in a sigmoid neuron should not be necessarily a sigmoid function. This called an activation function and it is the entity that makes neural network nonlinear. If no activation function was used, neural networks would be nothing more than over-parametrized linear models.

To wrap it up, we first introduced the concept of perceptron and sigmoid nuerons from which neural networks can be created. We stated the preference of sigmoid nuerons over perecptrons by stressing the role of the sigmoid function. The sigmoid function which serves as the activation function in neural network, yields two important conclusions: Firstly, it makes network less sensitive to perturbation and secondly, it makes the network nonlinear which is crucial for approximating nonlinear models. Even a simple XOR function requires an activation function if we attempt to model it by a neural network. For more detail of this, see [\cite{dl}, Chapter 6].

Now that we are familiar with the general notion of a neural network, we can introduce networks with specific architectures. We will first describe Feedforward Neural Networks in ,\ref{subsec:ann-fnn} and then Recurrent Neural Networks in \ref{subsec:ann-rnn}.

\subsection{Feedforward Neural Network} \label{subsec:ann-fnn}
\begin{quote}
"The goal of a feedforward network is to approximate some function $f\mbox{*}$. For example, for a classifier, $y = f\mbox{*}$ maps an input $x$ to a category $y$. A feedforward network defines a mapping $y = f(x;\theta)$ and learns the value of the parameters $\theta$ that results in the best possible function approximate." \cite{dl}
\end{quote}

We will scrutinize the meaning behind the relevant terms of a feedforward neural network and explain its structure.

MLPs are called \textit{feedforward} since information flows from input $x$ through intermediate computations composing $f$, and finally to the output $y$. The MLPs are devoid of \textit{feedback connections}, i.e. connections in which outputs of the model are fed back into itself. If we incorporate feedback connections in an MLP, they will become \textbf{recurrent neural networks}. These models are discussed in the next subsection (\ref{subsec:ann-rnn}).

Feedforward neural networks are called networks because they are often represented by composing together many different functions. For example, we might have three functions $f(x) = f^{(3)}(f^{(2)}(f^{(1)}(x))).$. In this case, $f^{(1)}$ is called the \textit{first layer} of the network, $f^{(2)}$ is called the \textit{second layer}, and so on. The overall length of the chain gives the \textit{depth} of the model. The name "\textit{deep learning}" arose from this terminology. The final layer of a feedforward network is called the \textit{output layer}. 

During neural network training, we propel $f(x)$ to match $f\mbox{*}(x).$ This matching is done by comparing the label $y$ associated with each point of training data $x$ with generated value of the model and then improve the model by eliminating the disparity between them. The improvement of the model is done by tuning parameters and minimizing the cost function of the model which is explained with more detail at the end of this subsection.

The intermediate layers are called \textit{hidden layers} since the training data does not reveal the desired output for each of them. Finally, these networks are called \textit{neural} because they are loosely inspired by neuroscience. Each hidden layer of the network is typically vector-valued. The dimension size of these hidden layers determines the \textit{width} of the model. See figure \ref{fig:ffn}

\begin{quote}
    "Each element of the vector may be interpreted as playing a role analogous to a neuron. Rather than thinking of the layer as representing a single vector-to-vector function, we can also think of the layer as consisting of many units that act in parallel, each representing a vector-to-scalar function." \cite{dl}
\end{quote}

\begin{figure}[H]
    \centering
    \includegraphics[width=.7\linewidth]{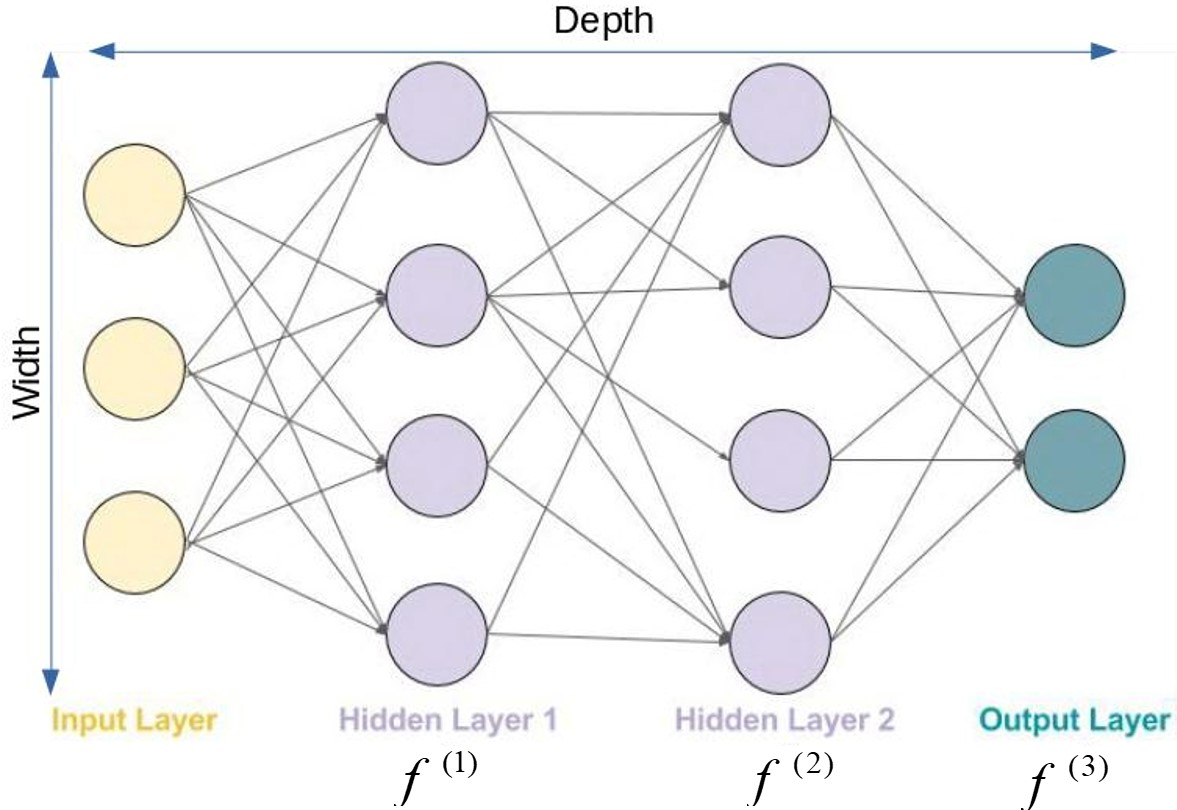}    \caption{Feedforward Neural Network}
    \medskip
    \small
    \label{fig:ffn}
\end{figure}

Layered neural networks conventionally contain a nonlinear activation function operating on individual coordinates---also known as \emph{elementwise nonlinearity}--- placed at the end of each layer. Without these, neural networks would be nothing more than over-parametrized linear models;

The process of minimizing the cost function and tuning parameters is often done by gradient descent which is a numerical algorithm for minimization. This algorithm includes a certain step called backpropagation. A thorough mathematical description of gradient descent and backpropagation algorithm is described in \ref{subsec:bptt-gen}.


\clearpage
\subsection{Recurrent Neural Network} \label{subsec:ann-rnn}

In feedforward networks and traditional ANNs in general, it is often presumed that all the inputs and outputs are independent of each other. But there exist real-world scenarios in which learning an instance of the data requires previous inputs or outputs. Even now that you are reading this thesis, you comprehend each word based on your understanding of previous words. Similarly, if one intends to predict the next word of a sentence, he has to remember the previous word of that sentence. It is as if human thought has persistence and retains information of previous observations and conclusions over time. MLPs are not capable of performing this ability and this would be an issue for learning phenomena that require memory and persistence of information. RNNs, on the other hand, address this issue since in them, the behavior of hidden neurons might not just be determined by the activationss in previous hidden layers, but also by the activations at earlier times. Furthermore, the activations of hidden and output neurons will not be determined just by the current input to the network, but also by earlier inputs. The persistence of information in RNNs was made possible by sharing parameters across different positions (index of time). Each member of the output is produced using the same update rule applied to the previous outputs. Such updating rule is often a (same) neural network layer, as the “A” in the figure below.

\begin{figure}[H]
    \centering
    \includegraphics[width=.15\linewidth]{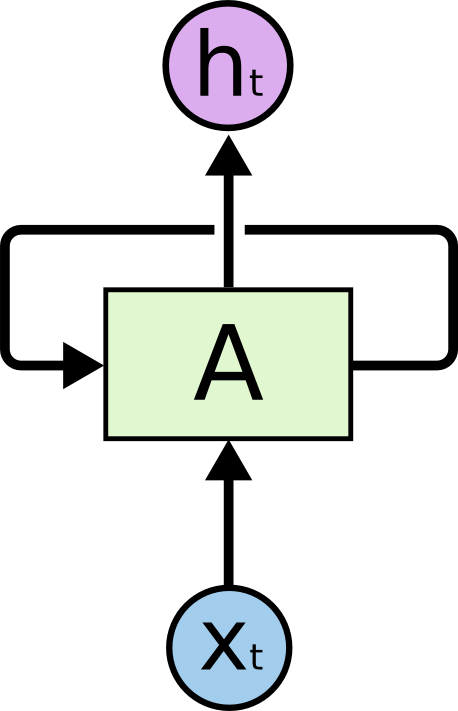}    \caption{Recurrent Neural Networks incorporate self-loops.}
    \medskip
    \small
    \label{fig:rnn1}
\end{figure}

The recurrent presentation of RNNs might make them seem obscure. However, RNNs can be perceived as multiple copies of the same networks, each passing a message to a successor. Unrolling this structure is conceivable as in figure \ref{fig:rnn1}. This chain-wise nature suggests that recurrent neural networks are closely tied with sequences and lists. 
\begin{figure}[H]
    \centering
    \includegraphics[width=.9\linewidth]{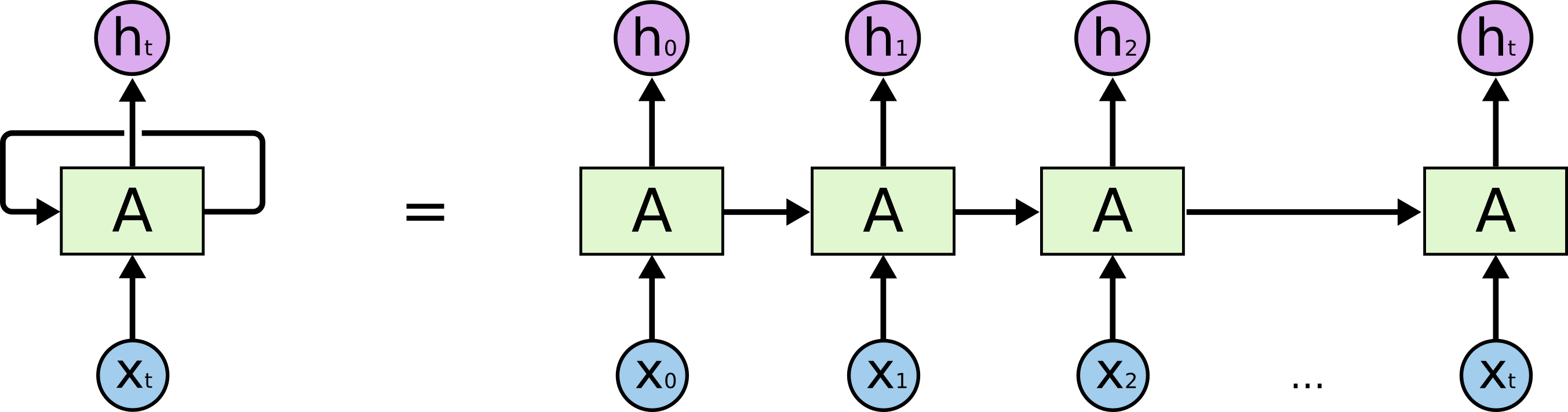}    \caption{An unrolled recurrent neural network}
    \medskip
    \small
    \label{fig:rnn2}
\end{figure}

A proper mathematical formalization of Recurrent Neural Networks is explained in \ref{sec:nnmf-rnn}. However, in order to express neural networks in a slightly more mathematical manner, consider the classical form of a dynamical system:
\begin{equation} \label{eq:ds}
    s^{(t)} = f(s^{(t-1)}; \theta),
\end{equation}

where $s^{(t)}$ is called the state of the system. If we unfold equation \eqref{eq:ds} for $\tau = 3$ time-steps, we obtain 

\begin{subequations}
\begin{align}
s^{(3)} & = f\big( s^{(2)}; \theta\big) \label{eq:folded} \\ 
& = f\big( f(s^{(1)}; \theta); \theta). \label{eq:unfolded} 
\end{align}
\end{subequations}

The unfolded equation yielded an expression devoid of recurrence that can be depicted by a traditional directed acyclic computational graph. The unfolded computational graph of equation \eqref{eq:ds} and equation \eqref{eq:unfolded} is illustrated in figure \ref{fig:rnn-unfold}.\\

\begin{figure}[H]
    \centering
    \includegraphics[width=.5\linewidth]{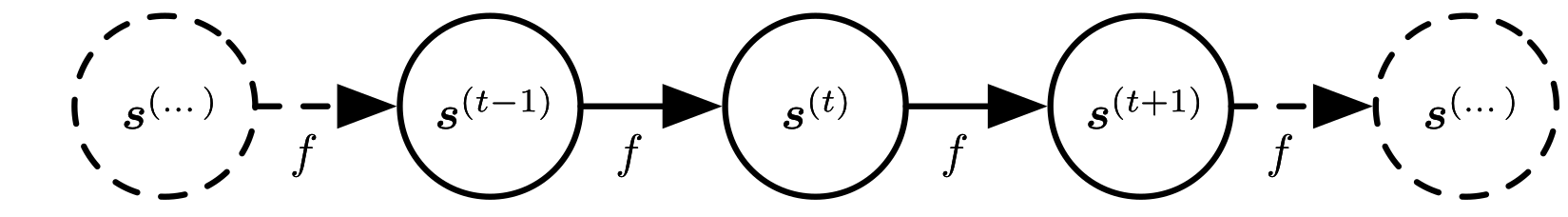}    \caption{The classical dynamical system described by equation \eqref{eq:ds}, illustrated as an unfolded computational graph. Each node represents the state at some time $t$, and the function $f$ maps the state at $t$ to the state at $t+1$. The same parameters (the same value of $\theta$ used to parameterize $f$) are used for all time-steps. Figure is from \cite{dl}}
    \medskip
    \small
    \label{fig:rnn-unfold}
\end{figure}

To incorporate the input of each step, consider a dynamical system driven by an external signal $x^{t},$

\begin{equation} \label{eq:rnn-ds}
    s^{(t)} = f(s^{(t-1)}, x^{(t)}; \theta),
\end{equation}

where we see that the state now contains information about the whole past sequence. 

\begin{quote}
    "Recurrent neural networks can be built in many different ways. Much as almost any function can be considered a feedforward neural network [according to universal approximation theorem], essentially any function involving recurrence can be considered a recurrent neural network." \cite{dl}

\end{quote}
Many recurrent neural networks use equation \eqref{eq:rnn-ds-h} or a similar equation to define the values of their hidden units. To stress the fact that the state is the hidden neurons of the network, we now rewrite equation \eqref{eq:rnn-ds} using the variable $h$ to represent the state,

\begin{equation} \label{eq:rnn-ds-h}
    h^{(t)} = f(h^{(t-1)}, x^{(t)}; \theta),
\end{equation}

illustrated in figure \ref{fig:rnn-unfold-hidden}.

\begin{figure}[H]
    \centering
    \includegraphics[width=.5\linewidth]{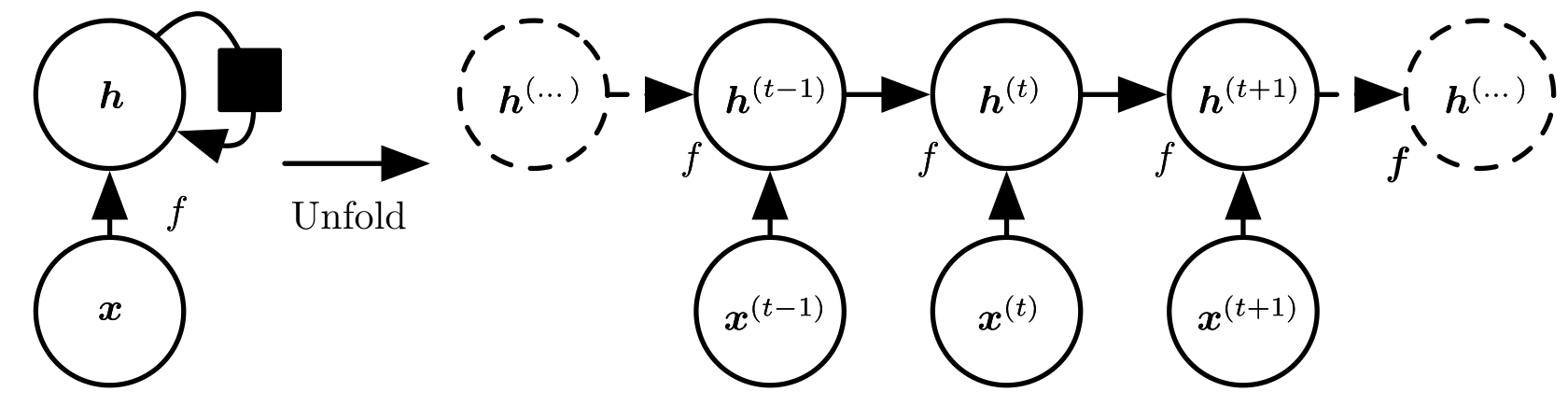}    \caption{A recurrent network with no outputs. This recurrent network just processes information from the input $x$ by incorporating it into the state $h$ that is passed forward through time. (\textit{Left}) Circuit diagram. The black square indicates a delay of a single time-step. \textit{Right} The same network seen as an unfolded computational graph, where each node is now associated with one particular time instance.}
    \medskip
    \small
    \label{fig:rnn-unfold-hidden}
\end{figure}

\begin{quote}
    "Some examples of important design patterns for recurrent neural networks include the following:
\begin{enumerate}
    \item Recurrent networks that produce an output at each time-step and have recurrent connections between hidden units.
    \item Recurrent networks that produce an output at each time-step and have recurrent connections only from the output at one time-step to the hidden units at the next time-step.
    \item Recurrent networks with recurrent connections between hidden units, that read an entire sequence and then produce a single output." \cite{dl}
\end{enumerate}
\end{quote}

The main difference between the first and second design is that in the former outputs from previous steps might be stored in the hidden state and may persist through later steps. In the latter, however, at each state, the only step from which its value used is the previous one. The type of design that is mathematically formalized in \ref{sec:nnmf-rnn} is the first one. 

\clearpage
\subsubsection{Challenge of long-term dependencies} 
When using data from the past, in some problems recent information (perhaps only data from previous time-step) suffices for learning. In other problems, however, information from more distant past is required. For example, consider a language model attempting to predict the next word in a passage based on previous ones. If one tries to predict the last word in the sentence "the clouds are in the \textit{sky}," he does not need further context---it's quite obvious that the next word is going to be sky. RNNs are adept in such cases, where the gap between the positions of information that are required to interplay in learning process is small. 

\begin{figure}[H]
    \centering
    \includegraphics[width=.5\linewidth]{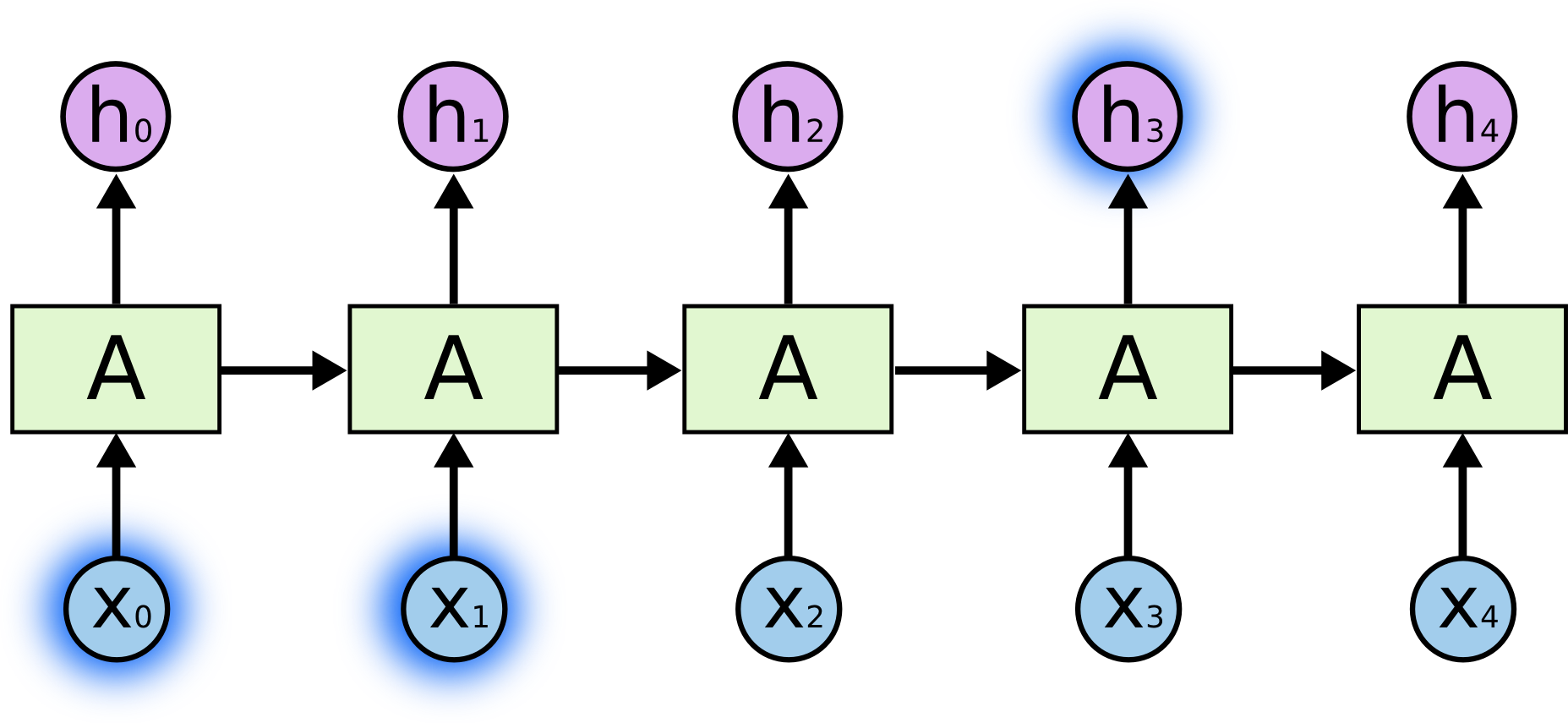}    \caption{Dependency on recent inputs from early past}
    \medskip
    \small
    \label{fig:long-term0}
\end{figure}

But there exist cases in which more context is inevitable. Consider attempting to predict the last word in the sentences "I grew up in France… I speak fluent \textit{French}." Recent information suggests that the next word is probably the name of a language, but if we try to narrow down which language, we need the context of France, from further back. There is a wide variety of similar cases where the gap between position of required information grow significantly large. Unfortunately, as that gap grows, RNNs become inept in learning to link the scattered positions of information.

\begin{figure}[H]
    \centering
    \includegraphics[width=.6\linewidth]{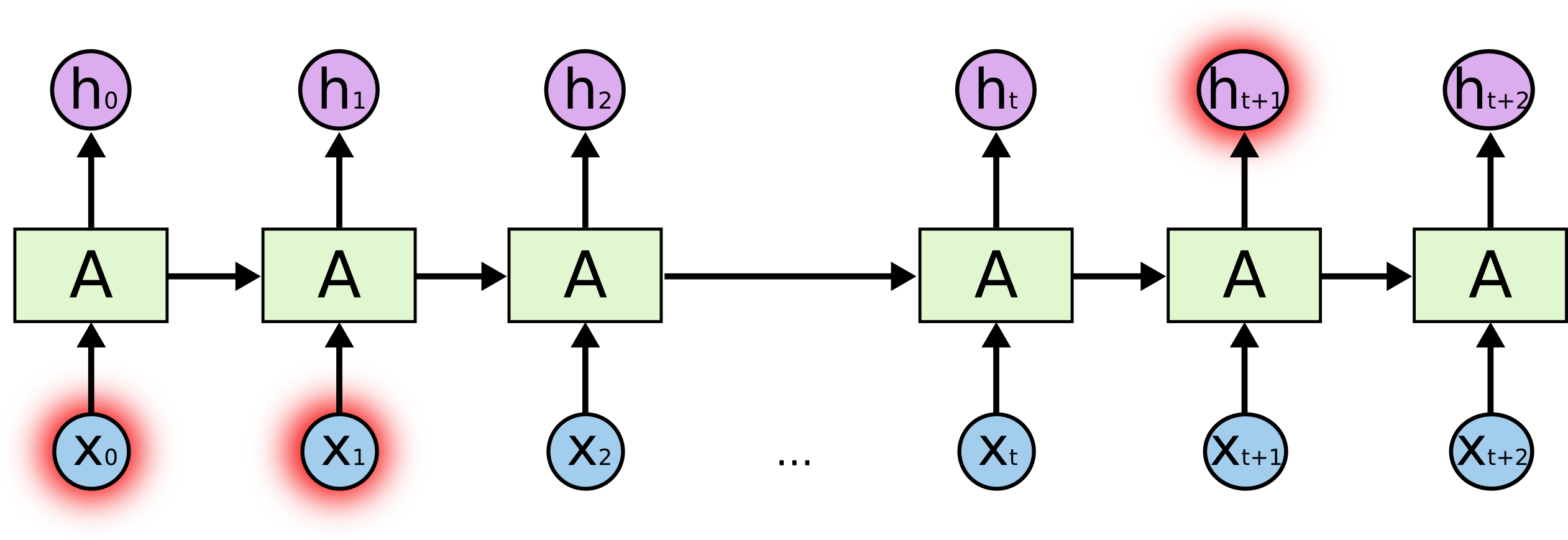}    \caption{Dependency on distant outputs from far past}
    \medskip
    \small
    \label{fig:long-term00}
\end{figure}

The reason that this fundamental issue occurs will be explained.

\begin{quote}
    "A particular difficulty that neural network optimization algorithms must overcome arises when the computational graph becomes extremely deep. Feedforward networks with many layers have such deep computational graphs. So do recurrent networks. Repeated application of the same parameters gives rise to especially pronounced difficulties." \cite{dl}
\end{quote}

\begin{figure}[H]
    \centering
    \includegraphics[width=.7\linewidth]{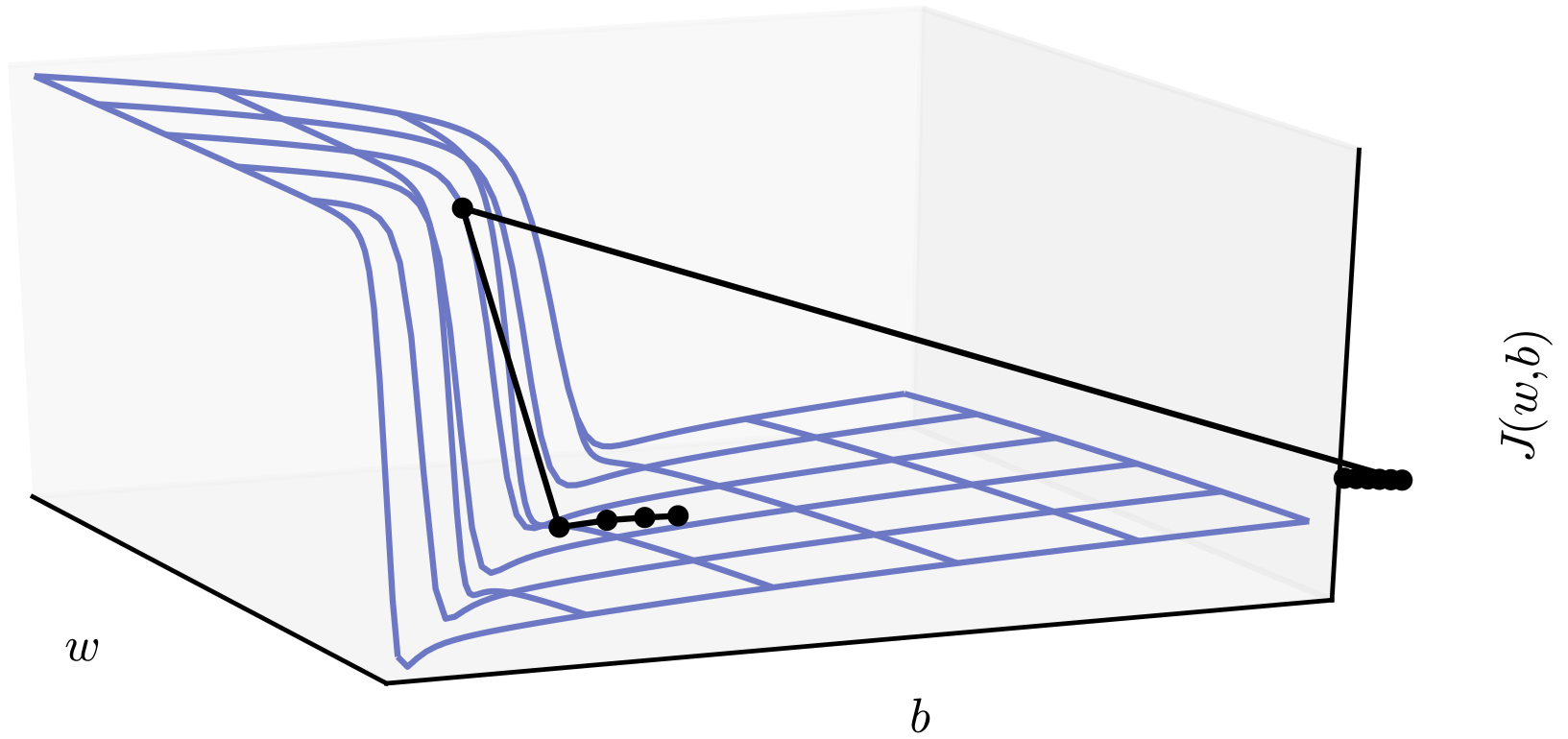}    \caption{The loss function for highly nonlinear deep neural networks or for recurrent neural networks often contains sharp nonlinearities in parameter space resulting from the multiplication of several parameters. These nonlinearities give rise to very high derivatives in some places. When the parameters get close to such a cliff region, a gradient descent update can catapult the parameters very far, possibly losing most of the optimization work that has been done. Figure is from \cite{dl}}
    \medskip
    \small
    \label{fig:long-term}
\end{figure}

Recurrent networks involve the composition of the same function multiple times, once per time-step. These compositions may lead to extremely nonlinear behavior.

In particular, the function composition employed by recurrent neural networks sort of resembles matrix multiplication. We can think of the recurrence relation

$$h^{(t)} = W^\intercal h^{(t-1)}$$

as a very simple recurrent neural network lacking a nonlinear activation function, and lacking inputs $x$  (this recurrence relation essentially describes the power method). It can be simplified to

$$h^{(t)} = (W^t)^\intercal h^{(0)},$$

and if $W$ has an eigendecomposition of the form

$$W = Q \Lambda Q^\intercal,$$

with orthogonal $Q$, the recurrence may be simplified further to

$$h^{(t)} = Q^\intercal \Lambda^t Q h^{(0)}.$$

The eigenvalues are raised to the power of $t$, causing eigenvalues with magnitude less than one to decay and vanish to zero and eigenvalues with magnitude greater than one to
explode. Any component of $h^{(0)}$ that is not aligned with the largest eigenvector will eventually be discarded. The \textbf{vanishing and exploding gradient problem} stresses the fact that gradients through such a graph are also scaled according to $\Lambda^t$. Vanishing gradients make it difficult to know which direction the parameters should move to improve the cost function, while exploding gradients can make learning unstable. In figure \ref{fig:long-term} a cliff structure is depicted that motivate gradient clipping. It is an example of the exploding gradient phenomenon.

\begin{quote}
    "One may hope that the problem can be avoided simply by staying in a region of parameter space where the gradients do not vanish or explode. Unfortunately, in order to store memories in a way that is robust to small perturbations, the RNN must enter a region of parameter space where gradients vanish. Specifically, whenever the model is able to represent long-term dependencies, the gradient of a long-term interaction has exponentially smaller magnitude than the gradient of a short-term interaction. This means not that it is impossible to learn, but that it might take a very long time to learn long-term dependencies, because the signal about these dependencies will tend to be hidden by the smallest fluctuations arising from short-term dependencies. 
    
    One way to deal with long-term dependencies is to design a model that operates at multiple time scales, so that some parts of the model operate at fine-grained time scales and can handle small details, while other parts operate at coarse time scales and transfer information from the distant past to the present more efficiently. Various strategies for building both fine and coarse time scales are possible." \cite{dl}

\end{quote}

\subsubsection{Leaky Units} 
One approach to designing a model with coarse and fine time scale is to obtain paths on which the product of derivatives is close to one. And an effective way to accomplish this is to have units with \textit{linear} self-connections and a weight near one on these connections.

When we accumulate a running average $\mu^{(t)}$ of some value $v^{(t)}$ by applying the update $\mu^{(t)} \leftarrow \alpha \mu^{(t-1)} + (1- \alpha)v^{(t)}$, the $\alpha$ parameter is an example of a linear self-connection from $\mu^{(t-1)}$ to $\mu^{(t)}$. When $\alpha$ is near one, the running average remembers information about the past for a long time, and when $\alpha$ is near zero, information about the past is rapidly discarded. Hidden units with linear self-connections can behave similarly to such running averages. Such hidden units are called \textit{leaky units}. The use of a linear self-connection with a weight near one is a way of ensuring that the unit can access values from the past.

There are two basic strategies for setting the time constants used by leaky units. One strategy is to manually fix them to values that remain constant, for example, by sampling their values from some distribution once at initialization time. Another strategy is to make the time constants free parameters and learn them. 

\subsubsection{Gated Recurrent Networks}
The most effective sequence models used in practical applications
are called \textbf{gated RNNs}. These include the \textbf{long short-term memory} and networks based on the \textbf{gated recurrent unit}.

Like leaky units, gated RNNs are based on the idea of creating paths through time that have derivatives that neither vanish nor explode. Leaky units made this possible via connection weights that were either manually chosen constants or were parameters. Gated RNNs generalize this to connection weights that may alter at each time-step.

\begin{quote}
    Leaky units allow the network to \textit{accumulate} information (such as evidence for a particular feature or category) over a long duration. Once that information has been used, however, it might be useful for the neural network to \textit{forget} the old state. For example, if a sequence is made of subsequences and we want a leaky unit to accumulate evidence inside each sub-subsequence, we need a mechanism to forget the old state by setting it to zero. Instead of manually deciding when to clear the state, we want the neural network to learn to decide when to do it. This is what gated RNNs do. 

\end{quote}
\subsubsection{LSTMs} 
LSTMs were introduced by Hochreiter and Schmidhuber (1997), and one of the core contributions of their initial versions was self-loops that produced paths along gradient can flow and persist for a long duration.
Despite having the same chain-wise nature of repeating modules in RNNs, the notorious difference of LSTMs is having a different module structure. Instead of having the same layer, such as a single tanh layer (Figure \ref{fig:rnn}), four layers interplay in their modules (Figure \ref{fig:rnn-lstm-1}). 

\begin{figure}[H]
    \centering
    \includegraphics[width=.7\linewidth]{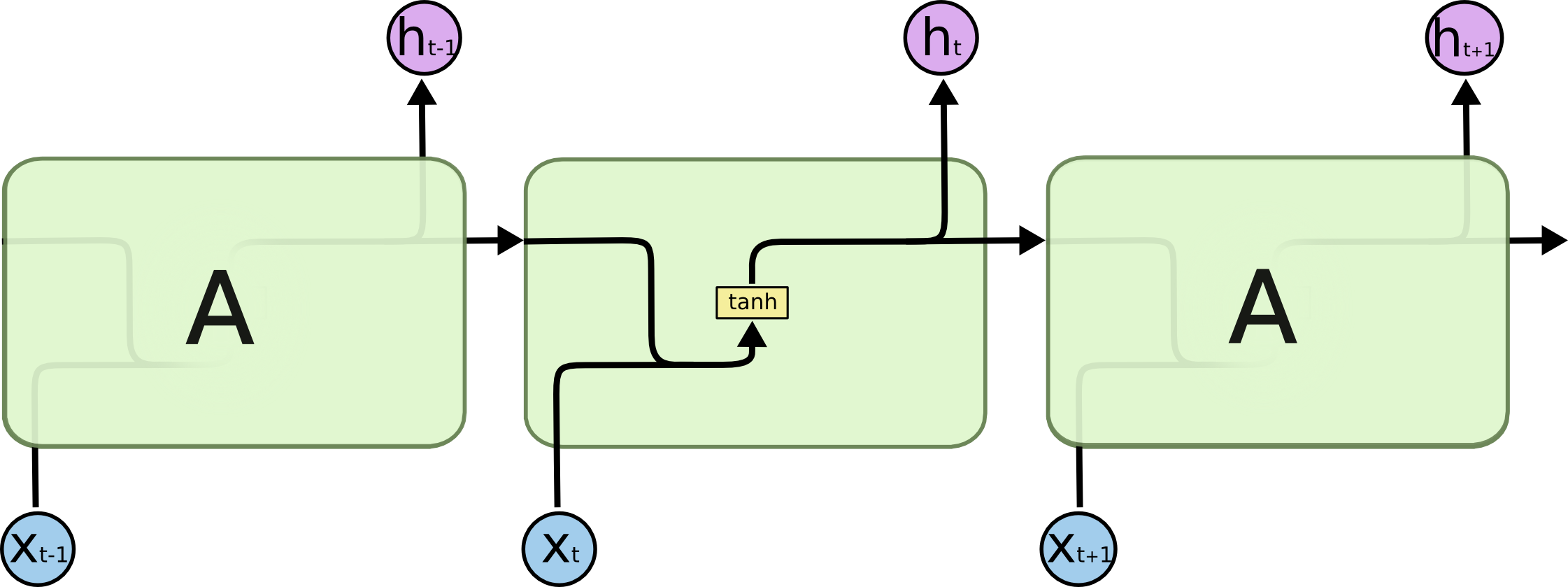}    \caption{The repeating module in a standard RNN contains a single layer.}
    \medskip
    \small
    \label{fig:rnn}
\end{figure}

\begin{figure}[H]
    \centering
    \includegraphics[width=.7\linewidth]{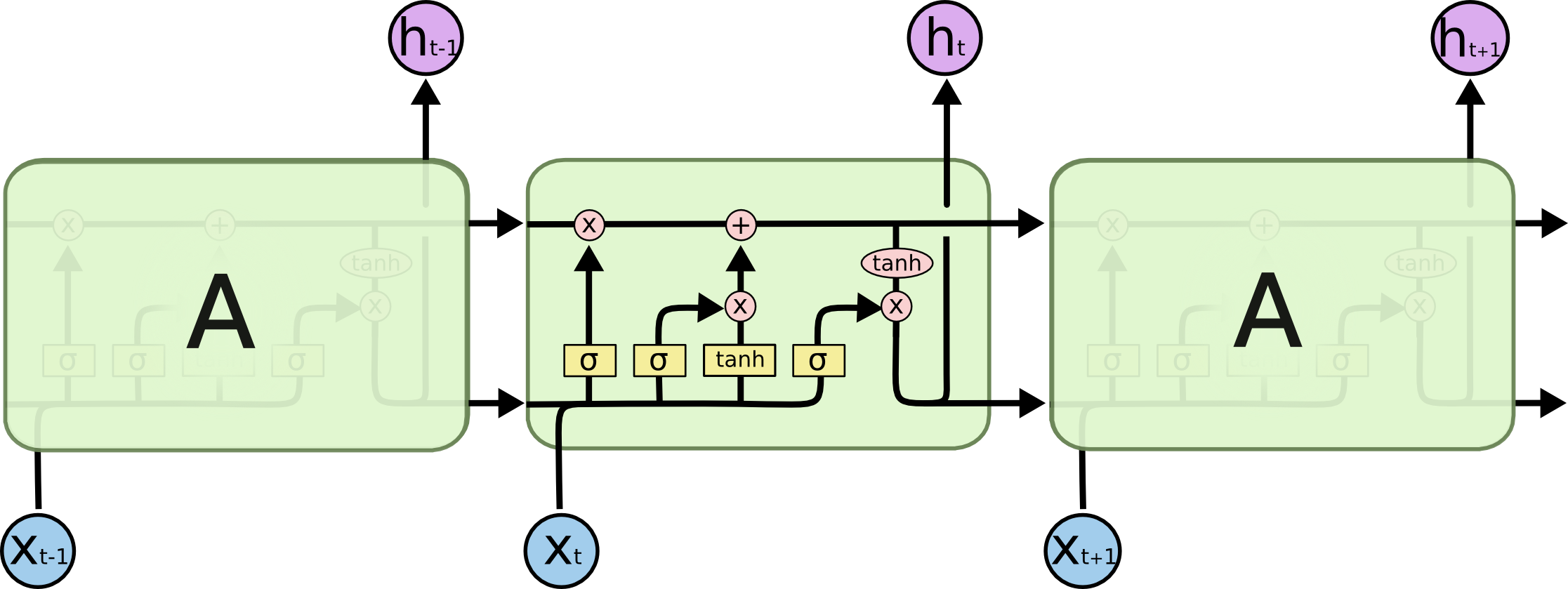}    \caption{The repeating module in an LSTM contains four interacting layers.}
    \medskip
    \small
    \label{fig:rnn-lstm-1}
\end{figure}

We will delve deeper into LSTMs by providing a walk-through shortly. Before doing so, we will specify our notation beforehand and also elucidate the core idea behind LSTMs.

\begin{figure}[H]
    \centering
    \includegraphics[width=.7\linewidth]{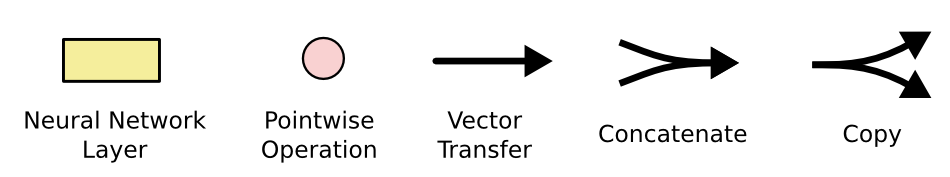}    \caption{Introducing Notation}
    \medskip
    \small
    \label{fig:rnn-lstm-not}
\end{figure}

In diagram \ref{fig:rnn-lstm-not}, each line carries an entire vector, from the output of one node to the inputs of others. The pink circles represent pointwise operations, such as vector addition, while the yellow boxes are learned neural network layers. Merging lines denote concatenation, while forking lines denote its content being copied and the copies are going to different locations.

\subsubsection{The Core Idea Behind LSTMs} 
One of the pivotal roles in an LSTM is the \textit{cell state}, which is the horizontal line placed on top of the figure \ref{fig:rnn-lstm-cellstate} The cell state serves as a passive carrying medium, not dissimilar to a conveyor belt. It surfs through the entire chain, including only some minor linear interactions. The cell state facilitates the flowing of information unaltered.

\begin{figure}[H]
    \centering
    \includegraphics[width=.9\linewidth]{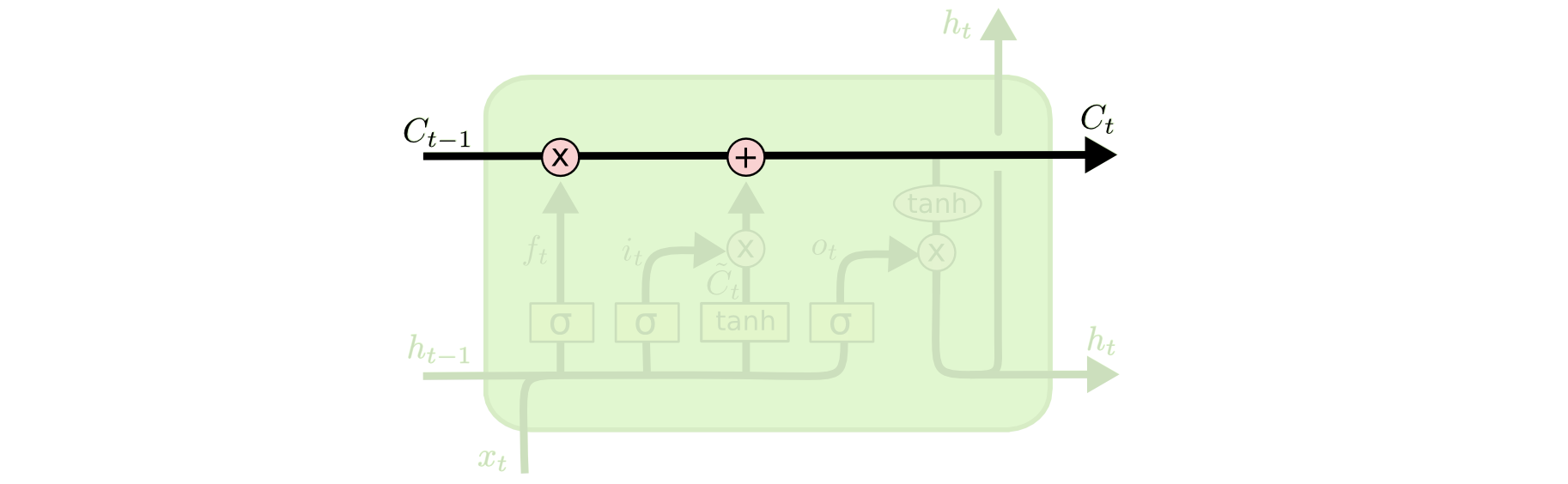}    \caption{Cell State in LSTM}
    \medskip
    \small
    \label{fig:rnn-lstm-cellstate}
\end{figure}

LSTMs are capable of adding and removing information stored in the cell state. This process is regulated by certain structures, called \textit{gates}. Gates allow in information optimally and they comprise a sigmoid neural net layer and a pointwise multiplication operator. See figure \ref{fig:rnn-lstm-gate}.

\begin{figure}[H]
    \centering
    \includegraphics[width=.1\linewidth]{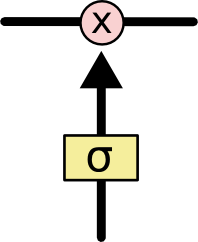}    \caption{Gate}
    \medskip
    \small
    \label{fig:rnn-lstm-gate}
\end{figure}

The sigmoid layer outputs numbers between zero and one. It acts akin to $\alpha$ in leaky units by describing how much of each component should be let through. A value of zero means "allow nothing in," while a value of one means "allow everything in."

An LSTM has three of these gates, to protect and control the cell state. 

\subsubsection{Step-by-Step LSTM Walk-through}
\begin{steps}
    \item First we have to decide which information should be discarded from the cell state. This decision is made by a sigmoid layer, called the \textit{forget gate} layer. The forget gate looks into $h_{t-1}$ and $x_t$, and outputs a value between 0 and 1 for each number of the cell state $C_{t-1}$ accordingly. See figure \ref{fig:lstm-wt-1}.\\ 
    \begin{figure}[H]
    \centering
    \includegraphics[width=.9\linewidth]{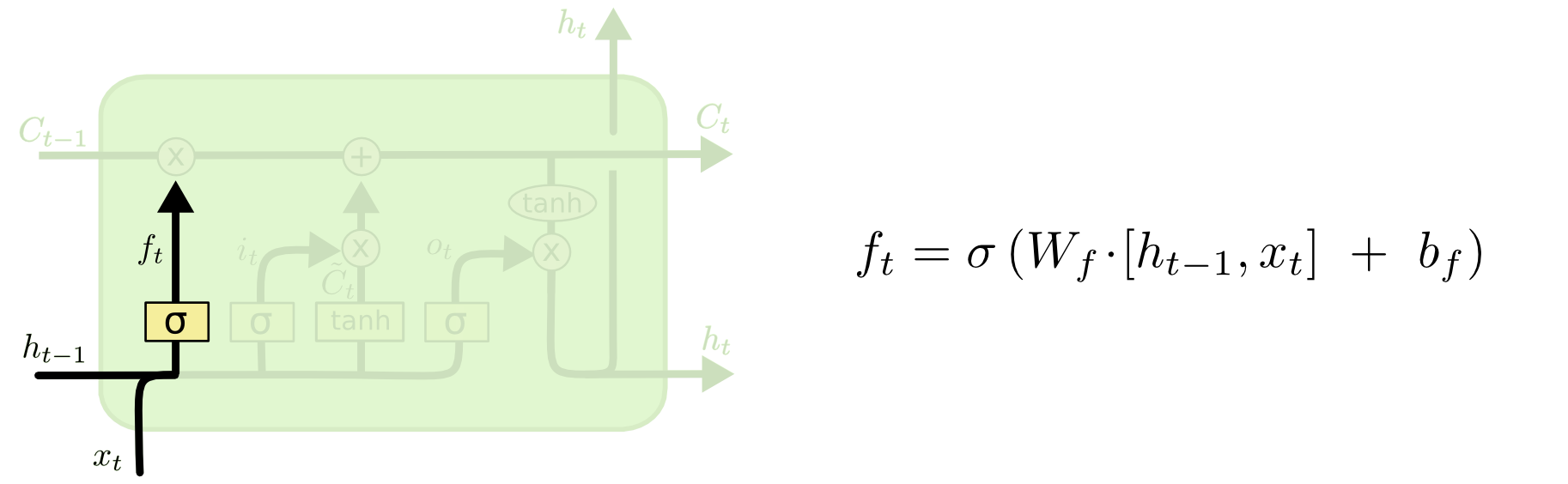}    \caption{Forget gate}
    \medskip
    \small
    \label{fig:lstm-wt-1}
    \end{figure}
    By referring to our earlier example of a language model predicting the next word based on all of the previous ones, we realize that in this problem, the cell state might include the gender of the present subject so as to use the correct pronouns. When a new subject is observed, we prefer to forget the gender of the old subject as it does not bear any useful information anymore.

    \item Second thing to do is to decide which information to store in the cell state. This step includes two tasks. Firstly, a sigmoid layer called the \textit{input gate} decides which values to update. Secondly, a tanh layer creates a vector of new candidate values $\tilde{C_t},$ that may append the state but the extent of update is scaled by the input gate. Afterward, the two aforementioned tasks will be combined to finalize updating of the state. See figure \ref{fig:lstm-wt-2}. \\
    
    \begin{figure}[H]
    \centering
    \includegraphics[width=.9\linewidth]{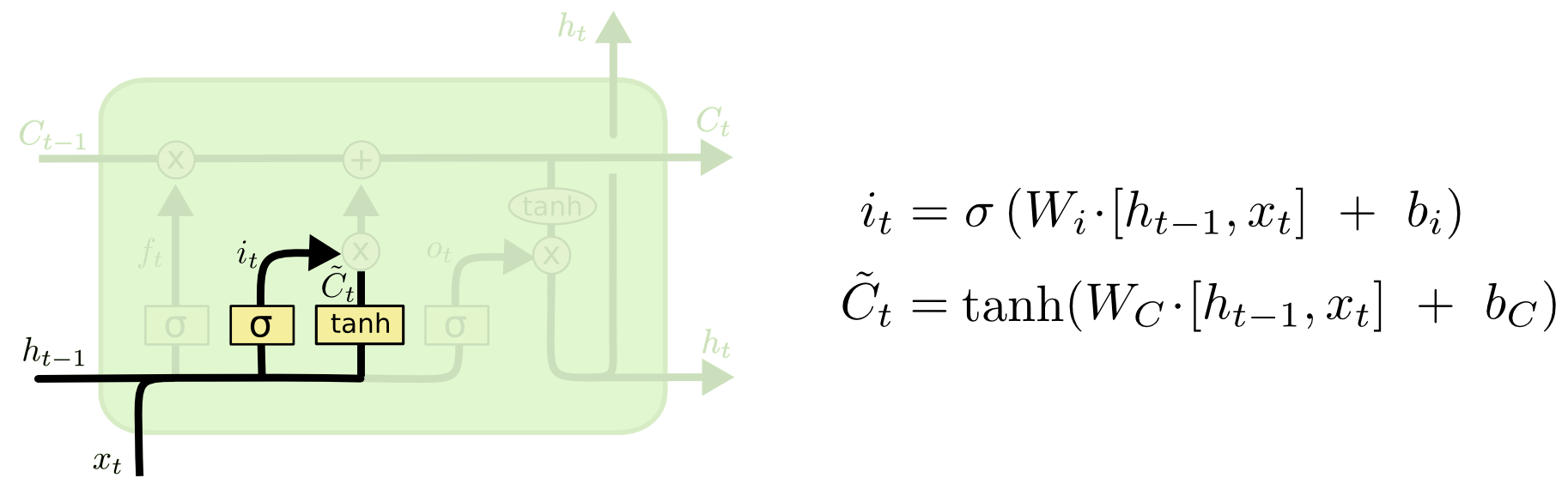}    \caption{Input Gate}
    \medskip
    \small
    \label{fig:lstm-wt-2}
    \end{figure}

    Back to the example of our language model again, we would prefer to add the gender of a new subject to the cell state, so as to replace the old one that we are forgetting.\\
    
    In order to finalize the updating of the old cell state $C_t$ into to the new cell state $C_t$, we have to implement the already made decisions from earlier steps. We do this by multiplying the old state by $f_t$ (forgetting the things we decided to forget in the previous step) and then adding  $i_t \ast \tilde{C_t}$, which are the new candidate values, scaled by how much we decided to update each state value. See figure \ref{fig:lstm-wt-3}.\\
    
    In the case of the language model, this is where we would actually drop the information about the old subject's gender and add the new information, as we decided in the previous steps. See figure 
    
    \begin{figure}[H]
    \centering
    \includegraphics[width=.9\linewidth]{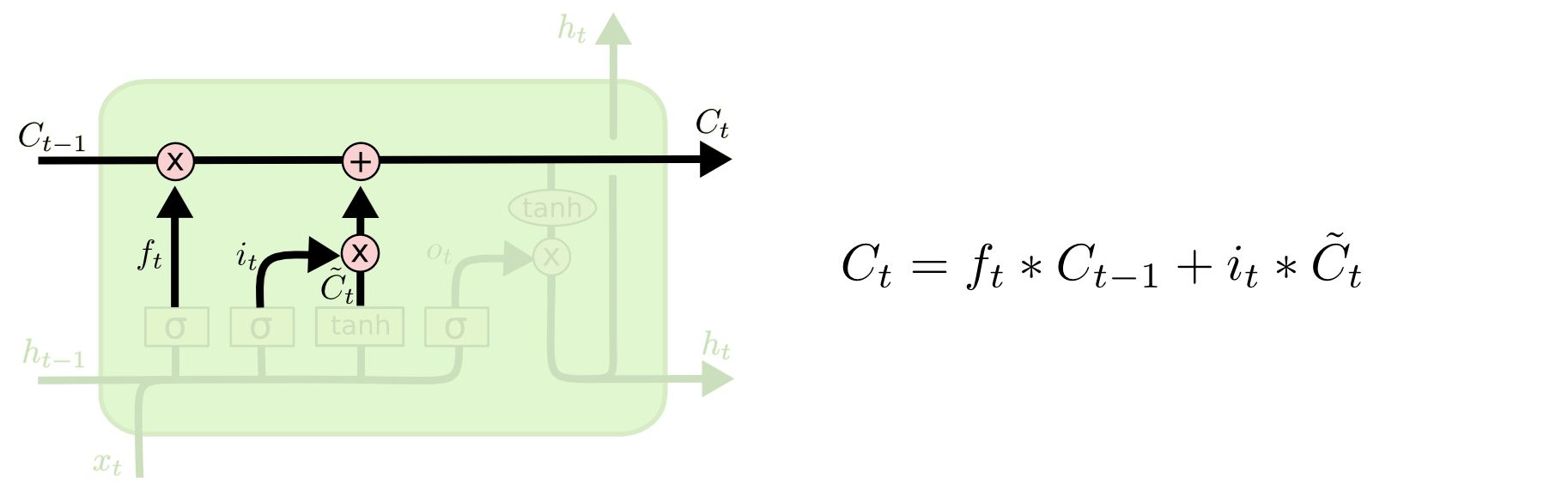}    \caption{Updating cell state}
    \medskip
    \small
    \label{fig:lstm-wt-3}
    \end{figure}

    \item Finally, we ought to decide which information is going to be outputted from the cell state. Firstly, a sigmoid layer, called \textit{output gate} decides which parts of the cell are yet to be outputted and then a tanh layer (it maps values between -1 and 1) multiplies the output of the sigmoid gate so that only the decided parts would exit the state. See figure \ref{fig:lstm-wt-4}.\\ In the language model example, since it just observed a subject, it might intend to output information relevant to a verb, in case the next input would be one. For example, it might output whether the subject is singular or plural so that we know what form a verb should be conjugated into. 
    
    \begin{figure}[H]
    \centering
    \includegraphics[width=.9\linewidth]{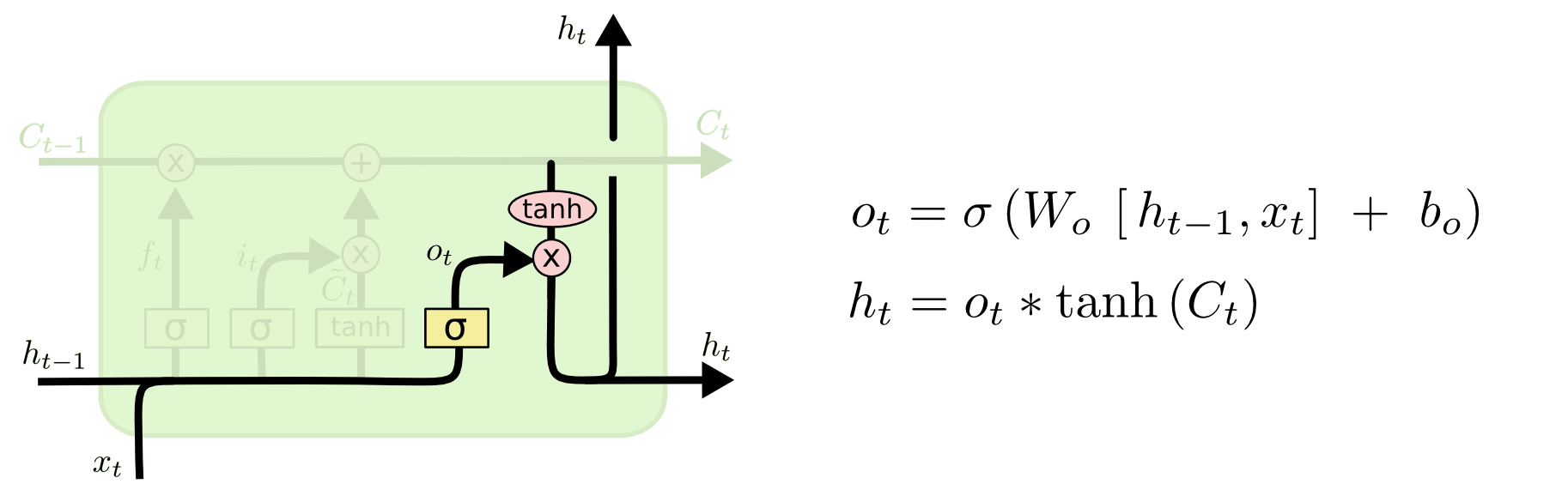}    \caption{Output Gate}
    \medskip
    \small
    \label{fig:lstm-wt-4}
    \end{figure}
    
\end{steps}

\ref{fig:lstm-ov} depicts the overall structure of LSTM and encompasses all the steps participating in cell state.

\begin{figure}[H]
\centering
\includegraphics[width=.7\linewidth]{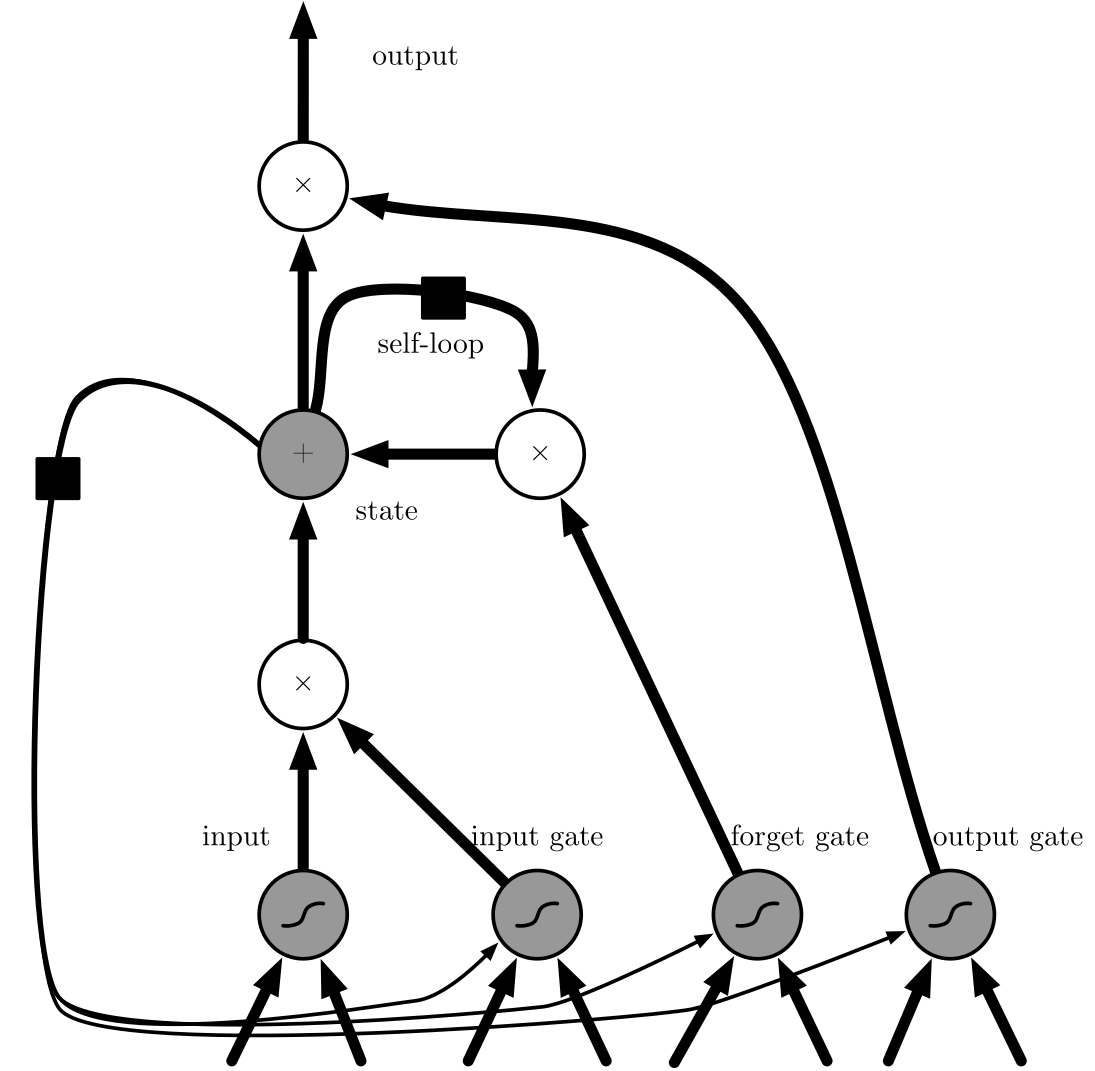}    \caption{Block diagram of the LSTM network cell. Cells are connected recurrently to each other, replacing the usual hidden units of ordinary recurrent networks. The state unit can also be used as an extra input to the gating units. The black square indicates a delay of a single time-step. Figure is from \cite{dl}}.
\label{fig:lstm-ov}
\end{figure}

\clearpage
\subsubsection{Variants on LSTMs} 
The walk-through we provided was a description for a generic LSTM. There are, however, other versions of LSTMs with slight alterations. For instance, One popular LSTM variant, introduced by Gers and Schmidhuber (2000), is adding “peephole connections.” This means that we let the gate layers look at the cell state. For more information on this, see \cite{lstm3}.

\section{Conclusion}
In this chapter, we described the general notion of deep learning along with its relation to neural networks. Furthermore, we pinpointed the advantages of deep learning models over traditional machine learning models from different perspectives, one of which was the ideas of representation of data and depth, and another one was reaching the fact that neural networks impose fewer prior beliefs on data in comparison to other models. Afterward, the structure of neural networks was introduced by pinpointing the way they are built upon perceptrons and sigmoid neurons. Then, we elaborated on the structure of RNNs after stressing their significant capability of memorizing earlier input as well as earlier outputs. We subsequently posed the challenge of long term dependency and proposed approaches to overcome it, such as leaky units and LSTMs. 
Finally, we described the structure of LSTMs by introducing the concepts of cell-state and gates and then explicate the interaction between the cell-state and the three gates contributing LSTMs (forget, input and output) through a step-by-step walk-through.\\

We realized that before feeding current input and previous state's output to the current state of the model, an input gate decides which parts of them to be discarded before entering the cell state. Afterward, an input gate looks into the new input and previous output (from which some parts are already discarded) and select the parts of them that are preferred to be stored in the cell state. Before storing new information, a vector of candidate values is created by a tanh layer. Thereafter, the cell state is updated by carrying out the tasks of the forget and output gates. Eventually, before the cell state output its values, an output gate regulates the leakage of information by choosing the parts that are yet to be outputted and serve as a new input for the next state.\\

This introduction to LSTMs would be insightful for understanding our practical part of this thesis, chapter \ref{subsec:implmnt}, in which we applied LSTM on stock prices in order to perform one-step ahead forecast, as well as theoretical part, Chapter \ref{sec:nnmf} in which the mathematical formalization of NNs and especially, RNNs and gated RNNs are explained in section \ref{sec:nnmf}.

\clearpage

\clearpage
\chapter{Implementation: Results with Analysis}
\clearpage
\label{subsec:implmnt}

In the previous chapter, after introducing deep learning, we elaborated on the structure of neural networks until we reached a particular kind of them, namely, LSTMs, and we will continue developing and formalizing this structure through a mathematical framework in the last chapter \ref{sec:nnmf}.

In this chapter, we provide results of implementing LSTM (long-short term memory) on two prices, namely, Goldman Sachs (GS) and General Electric (GE) to forecast one-step ahead of each stock's price. Two other stocks that are correlated with GS, namely, JPMorgan and Morgan Stanley were added as feature to GS, Also, auxiliary features were added to improve the model's accuracy. In the end, the ARMA (auto-reggressive integrated moving average) model is also applied so as to serve as a benchmark. 
The sections of this chapter are organized as follows: \\

\ref{subsec:data} \textit{Data}: In the first section, after describing the preprocessing methods we used to prepare data for our model, we will explain our feature selection procedure and then illustrate each feature's importance by XGBoost.

\ref{subsec:methodology} \textit{Methodology}:
In this section, we will point out the models we applied on the stock data and we will explain how the models process data and forecast. We then pose the challenges we confronted in each model.

\ref{subsec:optim} \textit{Parameter Optimization Algorithm}:
In this section, we will describe two parameter optimization algorithms that we used in LSTM (SGD, RMSProp), along with their advantages and disadvantages.

\ref{subsec:results} \textit{Results and Discussion}:
In this section, we put together figures, accuracy, and the diagrams of test vs prediction and then we evaluate our empirical results and draw deductions and provde analysis based on them. 

\ref{appx} \textit{Appendix}:
In this section, supplementary information is brought to provide further details of different implementation steps. Each part of this subsection is referenced throughout this section. 

\ref{conc} \textit{Conclusion}:
In the last section, we summarize the purpose of this section and prepare the reader for the rest of this thesis.

\section{Data} \label{subsec:data}

\subsection{Preprocessing}

One of the main assumptions of the model ARMA is that the data is stationary, a property which often is far from the characteristics of volatile stock prices. Therefore, before feeding prices into the model, we made attempts to make the time series of stocks stationary. The details of this are brought in \ref{appx:stationary}. \\

Before feeding features to LSTM model, a \textit{MinMaxScaler} method is utilized to normalize data between [0,1] via the following formula:
$$x\mbox{*}_t = (x_t - x_{min}) / (x_{max} - x_{min})$$

\subsection{Features} \label{subsec:features}
\subsubsection{Preparation}
One task in technical analysis involves gaining intuition about stocks from many sources, one of which might be technical indicators, e.g. mathematical calculation based on historic price, volume, etc. It would be enthralling that a machine learning model could gain some insight of its own from these indicators. To achieve this, a selected set of indicators of our target stock (GS) as well as correlated stocks (JPM and MS) are added to the features along with adjusted prices. Here is the list of features that we used:\\
\begin{itemize}
    \item \textit{ma7}: moving average with 7 window
    \item \textit{ma21}: moving average with 21 window
    \item \textit{12EMA}: exponential moving average with 12 window
    \item \textit{26EMA}: exponential moving average with 26 window
    \item \textit{MACD}: moving average convergence divergence = 12EMA - 26EMA
    \item \textit{20sd}: standard deviation with 20 window
    \item \textit{upper band}: ma21 + 20sd*2
    \item \textit{lower band}: ma21 - 20sd*2
\end{itemize}

In figure \ref{fig:ind}, diagram of the indicators is shown for all three stocks.

    \begin{figure}[H]
    \centering
    \includegraphics[width=.7\linewidth]{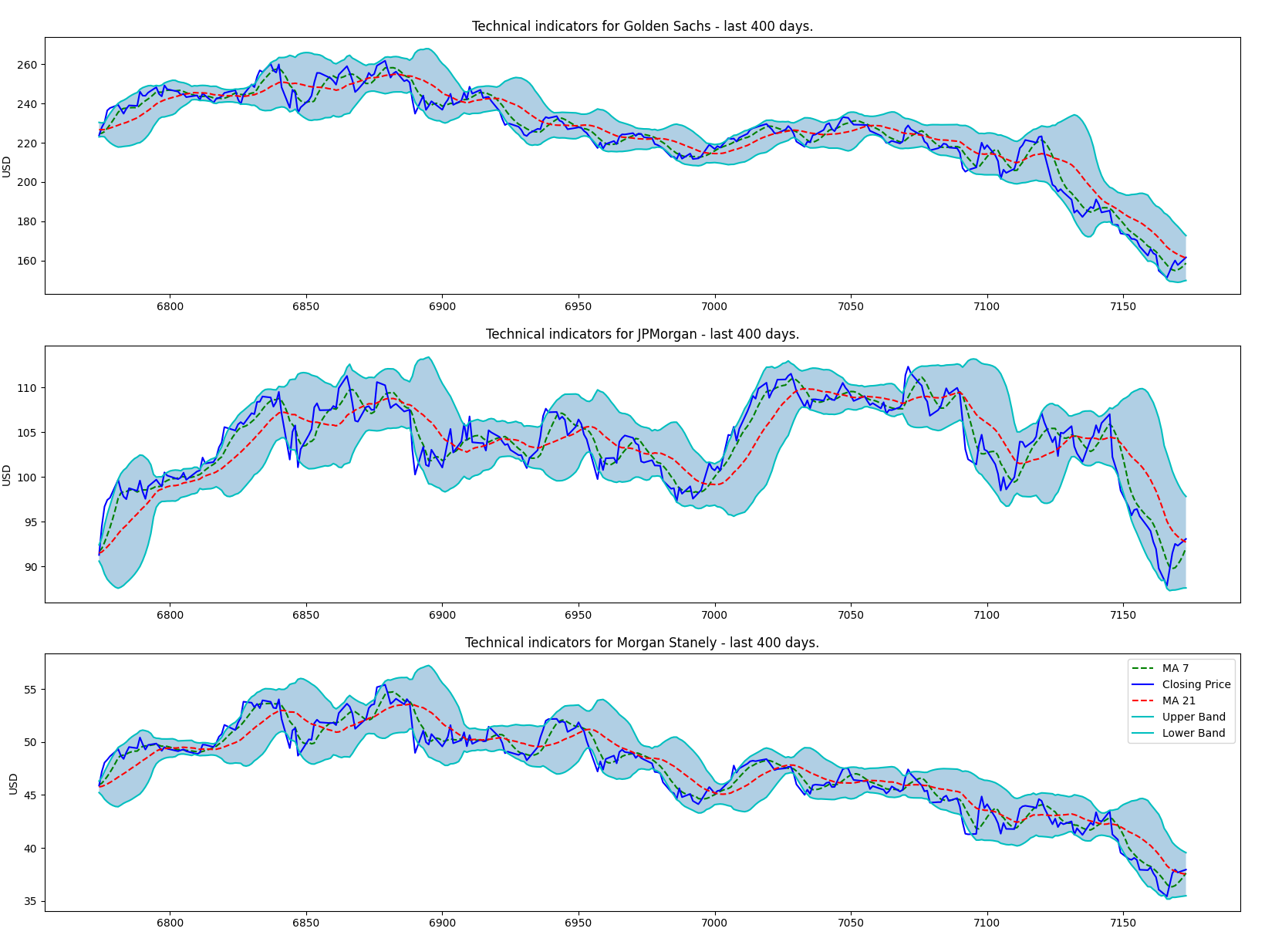}    
    \caption{Technical Indicators of price}
    \medskip
    \small
    \label{fig:ind}
    \end{figure}

In addition, Fourier transformation with different orders of components are fed to model so as to manifest long and short trends of time series to the model. See figure \ref{fig:fourt}.

    \begin{figure}[H]
    \centering
    \includegraphics[width=.7\linewidth]{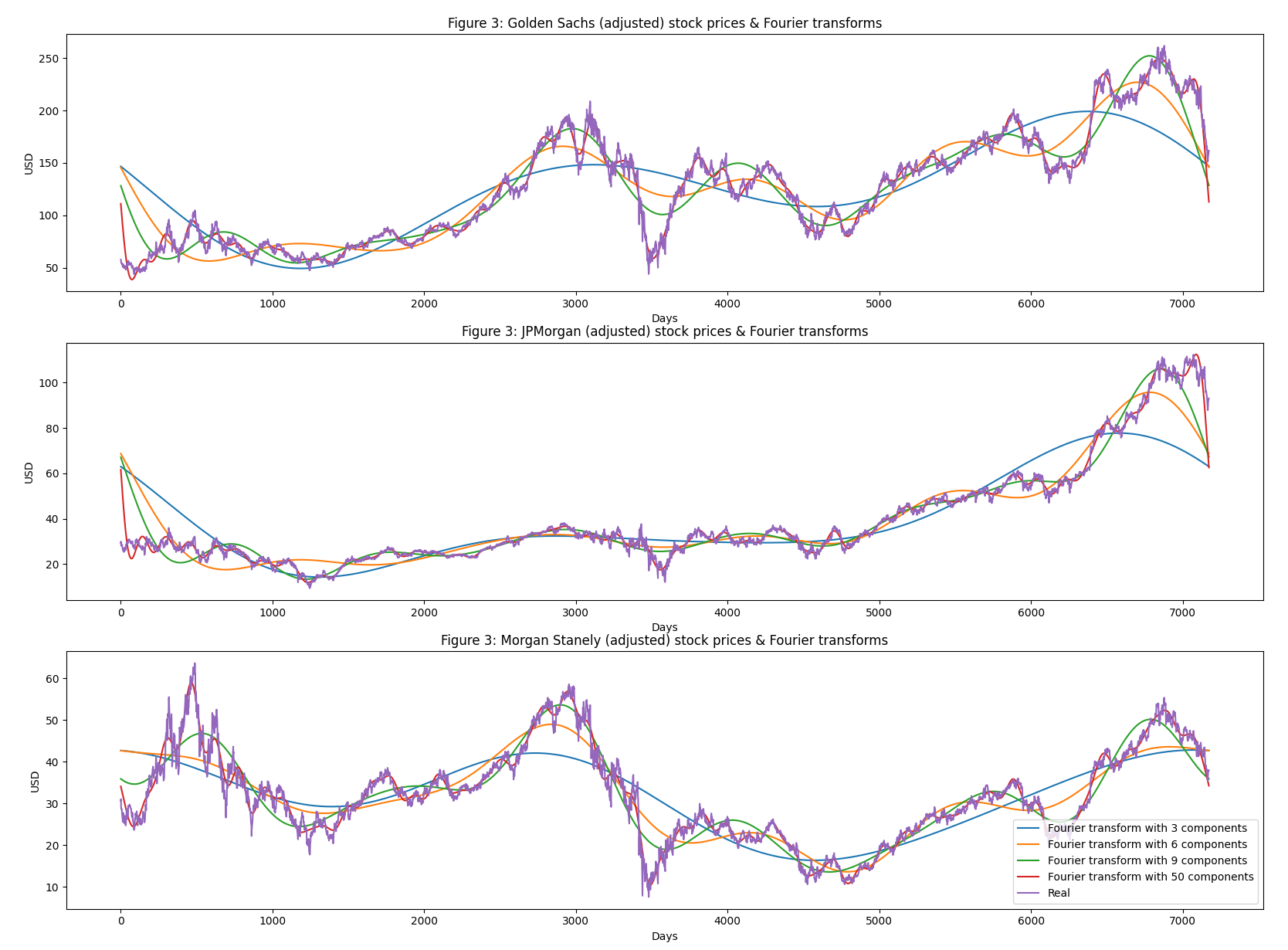}    
    \caption{Fourier transformations with different components as approximation of price}
    \medskip
    \small
    \label{fig:fourt}
    \end{figure}

\subsubsection{Feature Importance by XGBoost}
Measuring each feature's effectiveness in describing the data aid us in discarding obsolete ones and distinct the most effective ones. For this purpose, XGBoost is employed to determine each feature's importance with respect to the target feature, GS. XGBoost is an implementation of gradient boosted decision trees designed for enhancing speed and performance of more traditional decision trees. For more technical information on how feature importance is calculated in boosted decision trees, see [\cite{esl}, Chapter 10, Section 13].
Performance of XGBoost is shown in Figure \ref{fig:xgboost-loss} in terms of train error and validation error (test error). Also, Figure \ref{fig:xgboost-result} illustrates result of applying XGBoost such that each vertical line represents the importance of its corresponding feature.

 \begin{figure}[H]
        \centering
        \includegraphics[width=.6\linewidth]{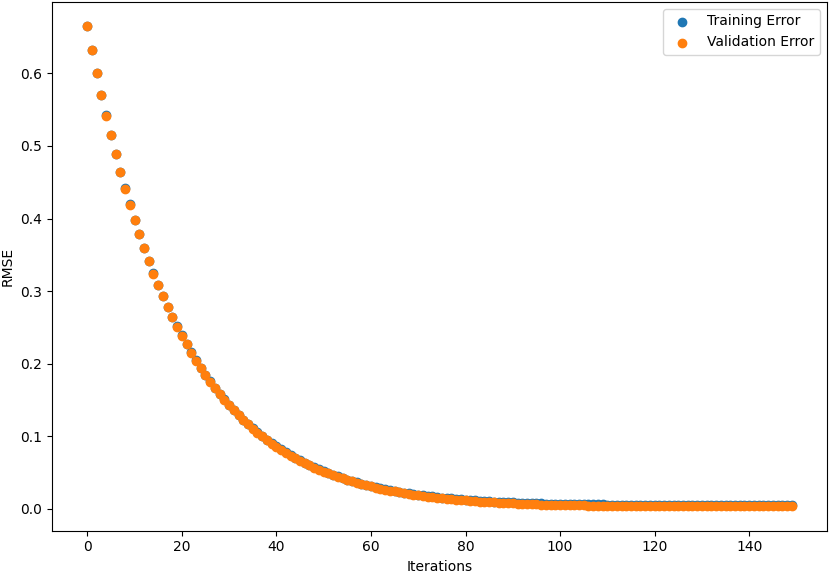}    
        \caption{Train and Validation Error of XGBoost model for feature importance}
        \medskip
        \small
        \label{fig:xgboost-loss}
    \end{figure}
 
 \begin{figure}[H]
        \centering
        \includegraphics[width=.8\linewidth,height=7cm]{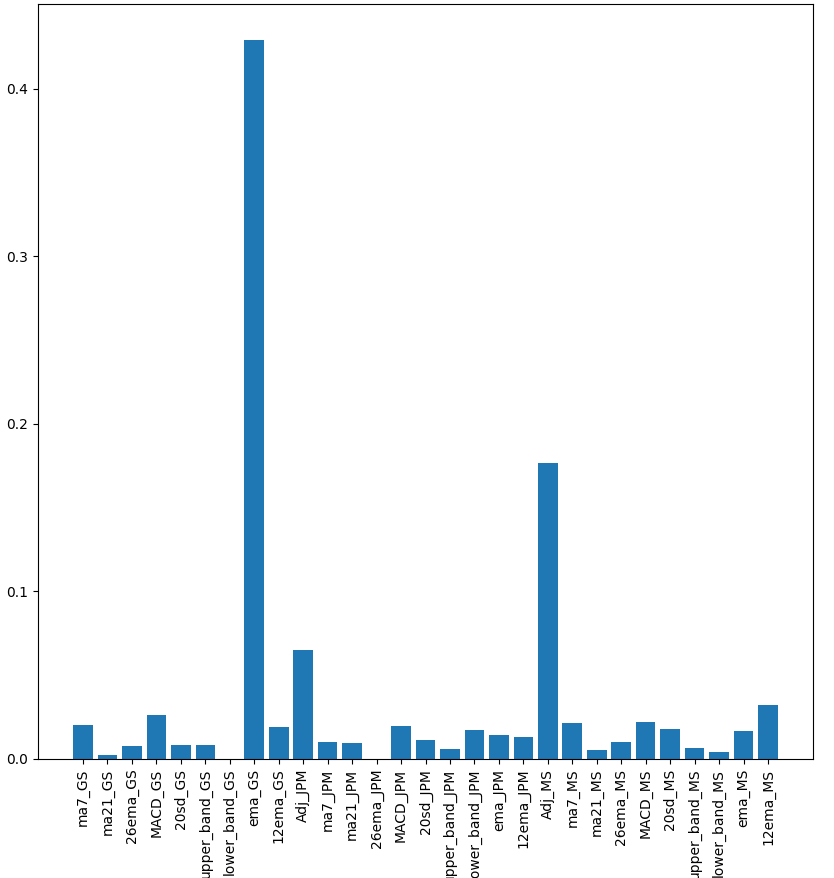}    
        \caption{Feature importance of technical indicators}
        \medskip
        \small
        \label{fig:xgboost-result}
    \end{figure}

From the figures, we can infer that among features, GS's exponential moving average is the most essential one and then, JPM and MS's prices have the most importance. This suggests that the exponential moving average (EMA) might have useful information. As a matter of fact, one of the ways to denoise data is to use EMA. This is elucidated in \ref{appx:EMA}.

\section{Methodology} \label{subsec:methodology}

\subsection{Recurrent Neural Networks (RNNs)}
The overall structure of RNNs is described in \ref{subsec:ann-rnn} and the precise mathematical formalization of them is explained in \ref{sec:nnmf-rnn}. Before elaborating on the implementation of LSTM, we will provide a terminology so that the next statements would be clear.\\

An \textit{epoch} elapses when an entire dataset is passed forward and backward through the neural network exactly one time. If the entire dataset cannot be passed into the algorithm at once, it must be divided into \textit{mini-batch}es. \textit{batch-size} is the total number of training samples present in a single \textit{min-batch}. An \textit{iteration} is a single gradient update (update of the model's weights) during training.  The number of iterations is equivalent to the number of batches needed to complete one epoch.
\textit{Samples or units} are the number of data, or say, how many rows are there in the dataframe.

The mentioned terms are valid for any neural networks. On the other hand, some other terms that are exclusive for recurrent neural networks should be explained.

Since we used the \code{Keras} library, we will explain the details of input shape throughout its framework. The input data of Keras' LSTM layer receives a 3D array as input which comprises three components: \textit{batch-size, time-step}, and \textit{input-dimension} respectively. Batch-size is already defined and \textit{time-step} will be defined shortly. \textit{input-dimensions} is equivalent to number of features.

See figure \ref{fig:lstminput} for visualization of input's shape.

  \begin{figure}[H]
        \centering
        \includegraphics[width=0.5\linewidth]{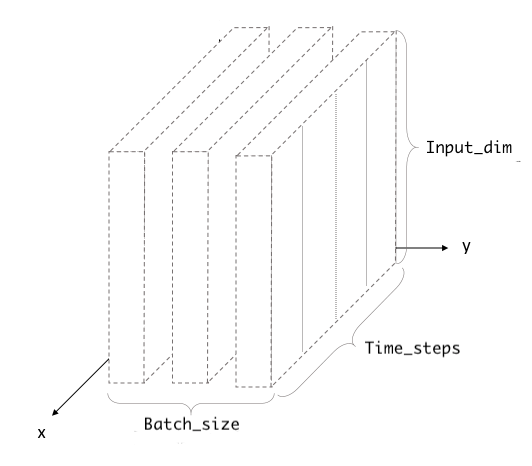}
        \caption{Input shape in LSTM network}
        \medskip
        \small
        \label{fig:lstminput}
    \end{figure}

One should meticulously reshape the input into a specific form of array, based on amount of \textit{time-step}. To illustrate this in the simplest form, suppose that the input is single sequence of $[10, 20, 30, 40, 50, 60, 70, 80, 90]$. It can be divided into multiple input/output patterns called samples, where three \textit{time-step}s are used as input and one \textit{time-step} is used as output for the one-step forecasting that is being learned. The data would appear as follows:

\begin{center}
 \begin{tabular}{||c c||} 
 \hline
 $X,$ & $y$ \\ [0.5ex] 
 \hline\hline
 10, 20, 30 & 40 \\ 

 20, 30, 40 & 50 \\

 30, 40, 50 & 50 \\
 ... & \\
 \hline
\end{tabular}
\end{center}

In the aforementioned example, \textit{time-step} is 3 and the type of of forecasting was one-step ahead. It is noteworthy to state that multi-step forecasting is also feasible by LSTM in which multiple days are predicted. However, our approach was one-step ahead forecasting. As shown in the example, LSTM puts a window of size \textit{time-step} on each data on \textit{mini-batch} and it shifts until it reaches the last value inside the \textit{mini-batch}. Each \textit{time-step} is a single occurrence of the cell in the unfolded graph of recurrent network. Number of \textit{time-step}s is denoted by $\tau$ as in \ref{subsec:ann-rnn}. In other words, \textit{time-step} is the number of units (days in our case) back in time that we expect the model to consider within the window while predicting the next day right after the window. For instance, if \textit{time-step} were one, the model would become equivalent to that of a feedforward since it merely uses previous output (today) to predict the next one (tomorrow). In effect, changing \textit{time-step}s would affect the memory of the model Consequently, each \textit{mini-batch} will be processed (and some of its values will persist through future steps) according to the stated steps.\\
For a more detailed explanation and example of RNN input shape, see \ref{appx:keras-rnn}

Once all features stated in \ref{subsec:features} are stacked up into a single dataframe, chunks of time series data will be created that are divisible by \textit{batch-size} and split them to train and test, we construct a 3D array in the form that we explained so as to feed it into the model. We used 80 percent of data for train and 20 percent for test.

 In figure \ref{fig:stocks}, the plot of all three stocks are illustrated along with a vertical split line that outlines the boundary between train and test data.
 \begin{center}
     
  \begin{figure}[H]
        \centering
        \includegraphics[width=0.9\linewidth]{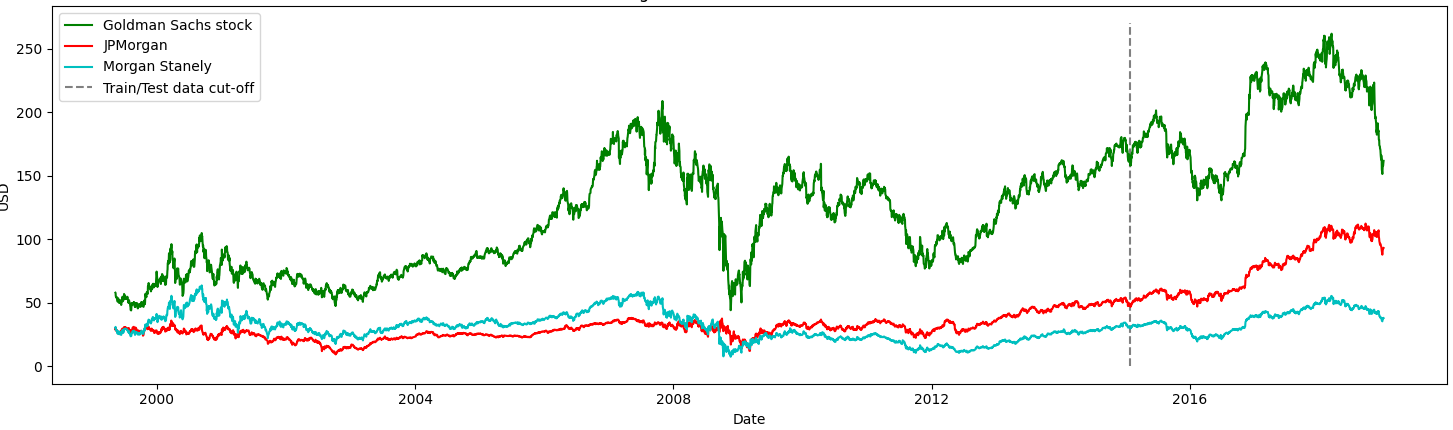}
        \caption{Daily adjusted prices of 3 stocks}
        \medskip
        \small
        \label{fig:stocks}
    \end{figure}
 \end{center}
 
The model attempts to minimize loss function, which is mean squared error (MSE) in our case. In each epoch, both training error and validation error are measured. The desired behavior would be a rather steady decrease in validation error even though it may shift upwards as some epochs. 

The discrepancy between prediction and test data will be also measured by R2 score loss after the fitting and prediction were done. 

A technique called \textit{Early Stopping} is also used so that an excessive number of epochs that would not lead to a reduction in validation error will not occur. More precisely, the minimum of previously obtained validation errors is computed at each epoch. Then, provided that this minimum wouldn't decrease after some certain amount of epochs, the model will stop automatically.
Minimizing the loss function is done by two optimization algorithms: Stochastic Gradient Descent and RMSProp, both of which are explained in \ref{subsec:paramopt}.

The main concern would be selecting hyperparamters of the model, i.e. parameters that is not tuned by the loss function and they are at least partly decided by the experimenter and determining them may also rely on prior knowledge and subjective experience.
 
 The hyperparameters that we needed to determine in our case were as following:\\
 Number of layers, number of neurons in each layer, learning rate, number of \textit{epoch, batch-size}, and \textit{time-step}. Moreover, we put activation functions in  the last and next to last layer. 
In table \ref{tab.3.1} our selection of hyperparameters is shown.

\begin{table}[H]
	\centering
	\begin{center}
		\caption{LSTM Hyperparameters}
		\renewcommand{\arraystretch}{2}
\begin{tabular}{|c|c|c|c|c|c|c|}
\hline
Data & \makecell{No. of \\Layers} & \makecell{No. of Neurons\\(in each layer)} & Learning Rate & Epoch & Batch-size & Time-step\\
\hline\hline
Closing Price & 4 & 100,60,20,1 & 0.0001 & 300 & 20 & 60
\\ 
\hline
Return of Price & 4 & 100,60,20,1 & 0.0001 & 300 & 20 & 90
\\
\hline
\end{tabular}
	
		\label{tab.3.1}
	\end{center}
\end{table}

In the last layer (output layer) and next to last layer, we chose \textit{sigmoid} and \textit{relu} as activation function, respectively. The \textit{stateful} option was set to 'True' in \textit{Keras}'s LSTM layer. For more details on how data is processed in \textit{Keras}'s LSTM, see \ref{appx:keras-rnn}.

\subsection{Auto Regressive Moving Average (ARMA)}
The ARMA model is obtained by summing two models, namely, auto-regressive and moving average, both of which are linear models. 

\subsubsection{Linear Models in Time Series}
A time series $r_t$ is said to be linear if it has the following form:
\begin{equation} \label{eq:linearts}
r_t = \mu + \sum_{i=0}^{\infty} \psi_i a_{t-i},
\end{equation}

where $\mu$ is the mean of $r_t, \psi_0=1$ and $\{a_t\}$ is a sequence of i.i.d random variables with mean zero and a well-defined distribution (i.e., $\{a_t\}$ is a white noise series). The mean and variance of \ref{eq:linearts} can be obtained as
$$E(r_t) = \mu, \: Var(r_t) = {\sigma_a}^2,$$

${\sigma_a}^2$ is the variance of $a_t$. Because $Var(r_t) < \infty, \{ {\psi_i}^2 \}$ must be a convergent sequence, that is, ${\psi_i}^2 \rightarrow 0$ as $i \rightarrow \infty$

\subsubsection{AR(p) Model} 
An AR(p) model has the following form:
$$r_t = \phi_0 + \phi_1 r_{t-1} + \phi_1 r_{t-1} + \cdots + \phi_p r_{t-p}+ a_t,$$

where $p$ is a non-negative integer and $\{a_t\}$ is assumed to be a white noise series with mean zero and variance
${\sigma_a}^2$. This model is in the same form as the well-known simple linear regression model in which $r_t$ is the dependent variable and $r_{t−1}$ is the explanatory variable.

\subsubsection{MA(q) model} 
There are several ways to introduce MA models. One approach is to treat the model as a simple extension of a white noise series. Another approach is to treat the model as an infinite-order AR model with some parameter constraints. \cite{fts} proves that the two approaches are equivalent. We will adopt the first approach. 
An MA(q) model has the following form:
$$r_t = a_t -\sum_{i=1}^{q} \theta_i a_{t-i},$$

where $\{a_t\}$ is a white noise series.

\subsubsection{ARMA(p,q)} 
A general ARMA(p,q) model has the following form:

$$r_t = \phi_0 + \sum_{i=1}^{p} \phi_i r_{t-i} + a_t - \sum_{i=1}^{q} \theta_i a_{t-i}$$

where $\{a_t\}$ is a white noise and $p$ and $q$ are nonnegative integers. Therefore, ARMA has both AR and MA as its components.

There are two matters that should be considered when using ARMA model, one of which is that ARMA assumes the data to be stationary and the other one is determining $p$ and $q$. Both matters are discussed in \ref{appx:stationary} and \ref{appx:stationary} respectively. 

Exponential smoothing (exponential moving average) (see \ref{appx:EMA}) and ARIMA models are the two most widely used approaches to time series forecasting. While the former is based on a description of the trend and seasonality in the data, ARIMA models aim to describe the autocorrelations in the data.

 \section{Parameter Optimization Algorithm} \label{subsec:optim} \label{subsec:paramopt}
  \subsection{Stochastic Gradient Descent} 
  Gradient Descent is a numerical optimization algorithm capable of finding solutions to a wide variety of problems. The gist of this algorithm is to tweak parameters iteratively so as to minimize a cost function (loss function, error function). To shed light on how it works, imagine that a person is stuck in a mountain within a dense fog but he is able to sense the slope of the ground below feet. Under this circumstance, a possible strategy is to follow the direction of the steepest slope. This strategy is the one used in GD, since it also measures the local gradient of the cost function with respect to the parameter vector $\theta$, and it navigates the direction of the descending gradient. The stop condition is when the gradient reaches zero. Basically, reaching the minima of function relies on two main decisions: the direction of movement and magnitude of movement at each step. GD aid in making these decisions effectively. The following formula provides the gist of this algorithm:
  
 \begin{figure}[H]
        \centering
        \includegraphics[width=.5\linewidth]{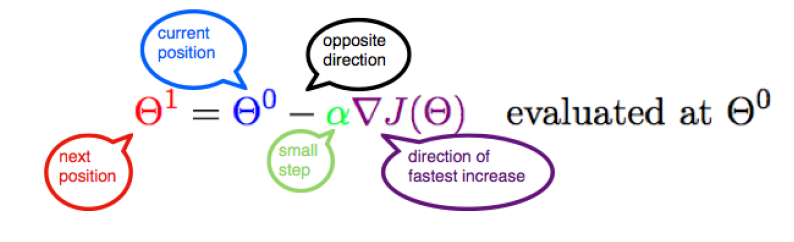}    
        \caption{A single iteration of Gradient Descent}
        \medskip
        \small
        \label{fig:gd}
    \end{figure}

Therefore, by calculating partial derivatives of each parameter, we are able to iterate through our datapoints using the value of newly updated parameters at each step. The new calculated gradient indicates the slope of our cost function at our current position and the direction we should move to update our parameters. The size of our update is controlled by the learning rate $\alpha$.\\

In algorithm \ref{alg:backppg}, Gradient Descent is applied to a general neural network. The difference between Gradient Descent (GD) and Stochastic Gradient Descent is that in the former all observations will contribute in computing $\nabla_J$ in \ref{fig:gd}. In the latter, however, a randomly selected subset (minibatch) of observations is used. SGD can be regarded as a stochastic approximation of gradient descent optimization because it replaces the actual gradient (calculated from the entire data set) by an estimate thereof (calculated from a randomly selected subset of the data). In the following figures, paths taken by GD and SGD are compared:

\begin{figure}[H]
    \centering
    \begin{minipage}{0.45\textwidth}
        \centering
        \includegraphics[width=.5\linewidth]{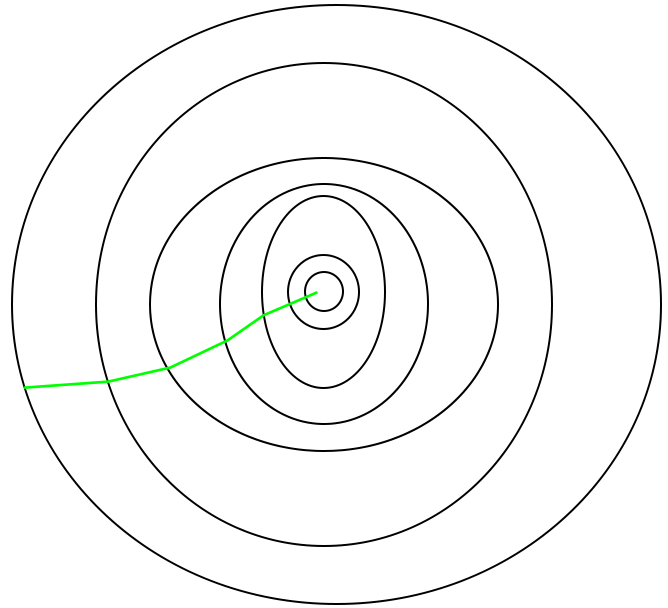}    
        \caption{Path taken by Gradient Descent}
    \end{minipage}\hfill
    \begin{minipage}{0.45\textwidth}
        \centering
        \includegraphics[width=.5\linewidth]{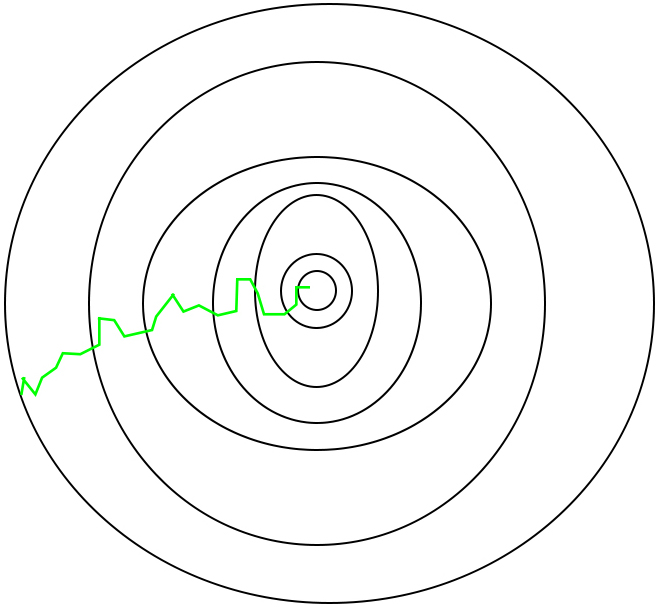} 
        \caption{Path taken by Stochastic Gradient Descent}
    \end{minipage}
\end{figure}

Figures imply that as SGD is generally noisier than typical Gradient Descent, it naturally took a higher number of iterations to reach the minima, because of its randomness in its descent. Even though it requires a higher number of iterations to reach the minima than typical Gradient Descent, it is still computationally much less expensive than typical Gradient Descent. So, in most scenarios, SGD is preferred over Batch Gradient Descent for optimizing a learning algorithm.

\subsection{RMSProp Optimizer}
 
 In the standard GD algorithm, it takes larger steps in one direction and smaller steps in another direction which slows down the algorithm. For instance, it takes larger steps in the $y$ axis in comparison to $x$ axis. Momentum resolves this issue by restricting oscillation in one direction so that algorithm converges faster. It provides the freedom to use higher learning rates with less chance of confronting the overshooting challenge. To fully appreciate how it is done, it is noteworthy to pinpoint that it employs the exponential moving average in its architecture. Therefore, we refer to EMA in \ref{appx:EMA} as the two concepts are closely tied.\\ 
 The GD with momentum updates parameters as follows:
 
 $$V_t = \beta V_{t-1} + (1-\beta) \nabla J_\theta (x,y;\theta)$$
 $$\theta = \theta - \alpha V_t,$$
 
 where $J$ is loss function and $\alpha$ is the learning rate.

In the following figures, paths taken by SGD with momentum and SGD without momentum are compared:

\begin{figure}[H]
    \centering
    \begin{minipage}{0.45\textwidth}
        \centering
        \includegraphics[width=.7\linewidth]{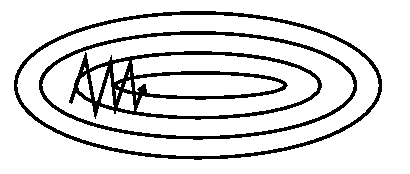}    
        \caption{Path taken by SGD without Momentum}
    \end{minipage}\hfill
    \begin{minipage}{0.45\textwidth}
        \centering
        \includegraphics[width=.7\linewidth]{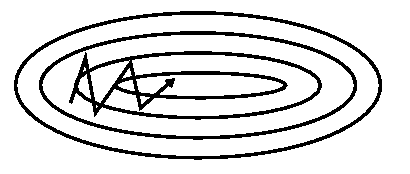}    
        \caption{Path taken by SGD with momentum}
    \end{minipage}
\end{figure}

Figures imply that momentum leads to accelerating gradients in the right direction.

Momentum GD almost always performs better than GD; yet we will utilize another optimizer, RMSProp, which is quite akin to momentum GD. The RMSProp also restricts the oscillations in the vertical direction in a similar manner but different in the way it calculates the gradient. It updates the parameters as follows:

 $$V_t = \beta V_{t-1} + (1-\beta) \nabla J_\theta (x,y;\theta)$$
 $$\theta = \theta - \alpha \frac{\nabla J(x,y;\theta)}{\sqrt{V_t} + \epsilon},$$
 
 We trained our model on both SGD and RSProp optimizers. In all scenarios, RMSProp was superior.




\section{Results and Discussion} \label{subsec:results}
\subsection{LSTM}
\subsubsection{Forecasting Lag Conundrum in RNN} 
 A certain predicament arises in sequence prediction when using RNNs. We begin explaining it by displaying the result of one of our earlier attempts to forecast Goldman Shack stock. The model and procedures we used in this earlier attempt are quite identical to those of our latest models. We also used the features stated in \ref{subsec:features} and the target was the adjusted prices.

 In figure \ref{fig:close-plot}, train error and test error from MSE loss is shown and in figure \ref{fig:close-trvd}, both prediction and test data are illustrated.
 
    \begin{figure}[H]
        \centering
        \includegraphics[width=.5\linewidth]{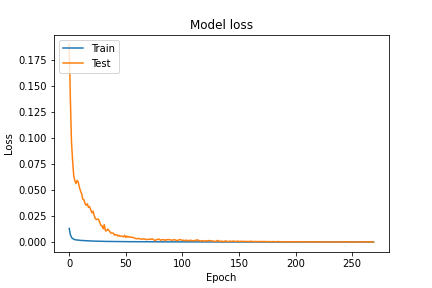}
        \caption{Train and Test Error of LSTM on adjusted prices}
        \medskip
        \small
        \label{fig:close-plot}
    \end{figure}

   \begin{figure}[H]
        \centering
        \includegraphics[width=.5\linewidth]{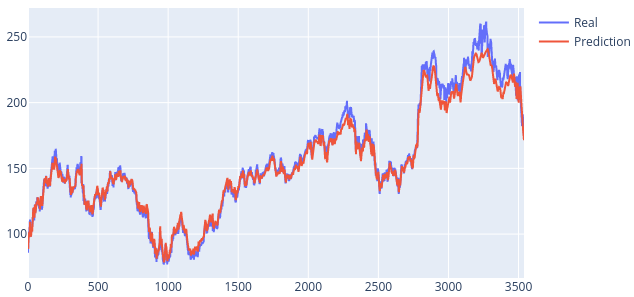}
        \caption{Test vs Prediction}
        \medskip
        \small
        \label{fig:close-trvd}
    \end{figure}

Judging from errors and the diagram of prediction vs test data, one may hastily infer that an almost impeccable forecasting is accomplished. Nevertheless, if we scrutinize our prediction by enlarging the graph's scale and zoom into days of the prediction, a phenomenon becomes salient. For instance, diagram of days 2100-2500 is illustrated in figure \ref{fig:close-zoom1}
 
    \begin{figure}[H]
        \centering
        \includegraphics[width=.5\linewidth]{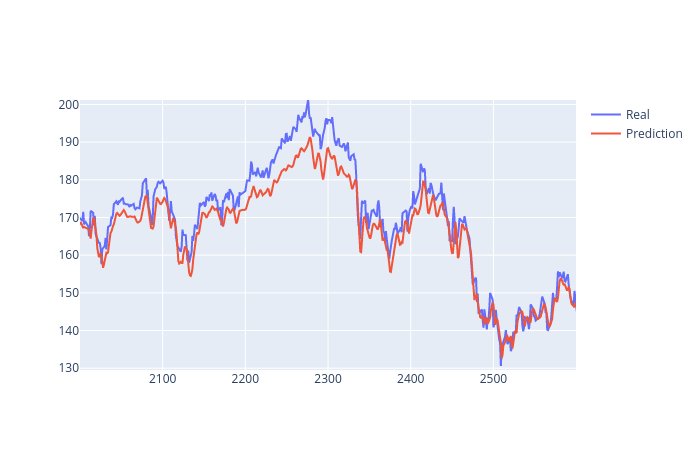}
        \caption{Test vs Prediction}
        \medskip
        \small
        \label{fig:close-zoom1}
    \end{figure}
    
In figure \ref{fig:close-zoom2}, we zoom further and observe days 2300-2500.

\begin{figure}[H]
    \centering
    \includegraphics[width=.5\linewidth]{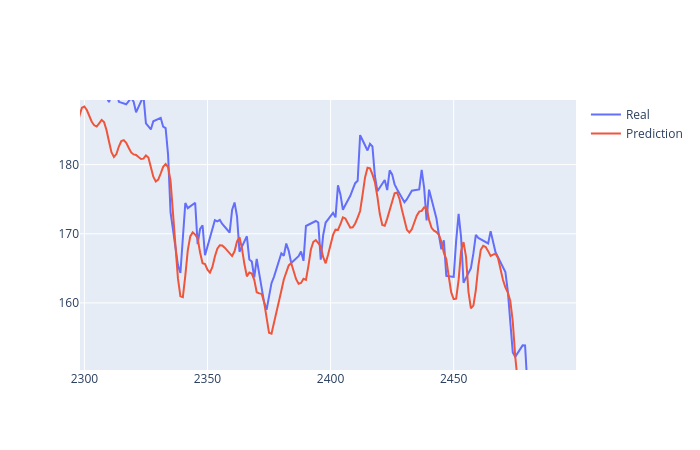}
    \caption{Test vs Prediction}
    \medskip
    \small
    \label{fig:close-zoom2}
\end{figure}
 
The mean squared error of test data and prediction is 0.00067, which also would seem excellent.

As soon as we narrow our scope to shorter days, a lag in forecasting becomes salient. We call this phenomenon "forecasting lag". One may wonder whether this lag in prediction matter and how it can be avoided. That's the conundrum we are facing. Note that if a person for example claims to predict the tomorrow's price of stock accurately and actually does it by waiting for tomorrow to come and after knowing the price, report it as his prediction, this prediction is futile. So the lag is indeed a problem. It exclusive for RNNs and it makes forecasting more challenging as this underperformance is not readily detectable from the standard loss measures of train and test errors. This challenge exists because of RNNs' recurrent structure (and the persistent of data which stems from model's memory) and also the essence of stock data, in the sense that although the model surmounts in following the dynamics of data, it might forecast by waiting for the input of future sequences to be fed to the model so as to set them as the forecast value of earlier days. By means of this, memorizing the future inputs and assigning them as output (prediction) for earlier inputs (a day or two days earlier for instance) creates a devious forecast result. Since stock data has high frequency and oscillates rather slow, the lag is not readily detectable.

There are some ways to address this particular pitfall:\\
\begin{enumerate}
    \item It is preferable to find a loss that is less subject to this issue. Both MSE and R-squared losses fail in detecting this problem. One may transform the forecasting from a regression type of problem employ to a classification one, perhaps by setting different classes based on the movement of stock. Once transformed, classification measures such as ROC and AOC curve become available.
    \item One should be cautious about the features. If a feature uses index later than the current and previous inputs at each state, then it reveals future values. An averaging that uses future values of sequence for computing is an example of this.
    \item Adjusting the parameters and hyperparameters may use a model devoid of the problem. For example, changing optimizer or select different hyperparameters might coerce the model not to memorize earlier inputs or at least use them only to the extent that the problem doesn't arise.
    \item Altering the representation of the target feature (price in our case) might solve the issue. As already discussed in \ref{subsec:rep}, the representation of input has a profound effect on the performance of the model. In many stock forecasting works, designate closing price or adjusted price of stock as the target feature. This might hinder the model's performance and the lag arise as price of each day lies close to that of previous and next day. Thus, instead of closing prices, one can use \textit{price return} as the target feature.
\end{enumerate}

\subsubsection{Price Return}
\begin{quote}
"Campbell, Lo, and MacKinlay (1997) give two main reasons for using returns. First, for average investors, return of an asset is a complete and scale-free summary of the investment opportunity. Second, return series are easier to handle than price series because the former have more attractive statistical properties." \cite{fts}
\end{quote}

We have used \textit{One-Period Sample Return} as feature which is defined as follows.

$$R_t = \frac{P(t) - P(t-1)}{P(t-1)} = \frac{P(t)}{P(t-1)} - 1$$
or
$$P_t = P_{t-1} (1 + R_t),$$

where $P(t)$ is the price of an asset at time index $t$ and $R_t$ is the return. The second formula is more intuitive as it can be interpreted as adjusting the previous price based on today's return.

In contrast with the closing or adjusted prices (see figure \ref{fig:stocks}), price return does not manifest trend (mean) of data. In effect, the dynamic of price return resembles that of a white noise with no sign of any pattern. See figures \ref{fig:return1} and \ref{fig:return2}.

\begin{figure}[H]
    \centering
    \includegraphics[width=.9\linewidth]{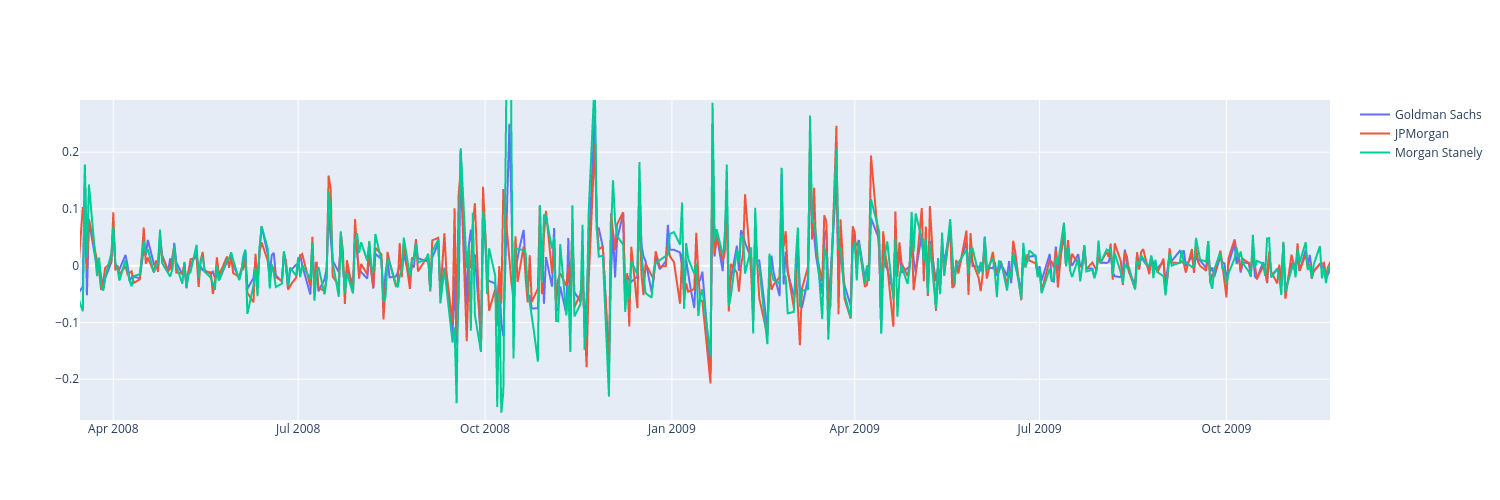}
    \caption{Price returns of three stocks in 2008-2009}
    \medskip
    \small
    \label{fig:return1}
\end{figure}

 \begin{figure}[H]
    \centering
    \includegraphics[width=.9\linewidth]{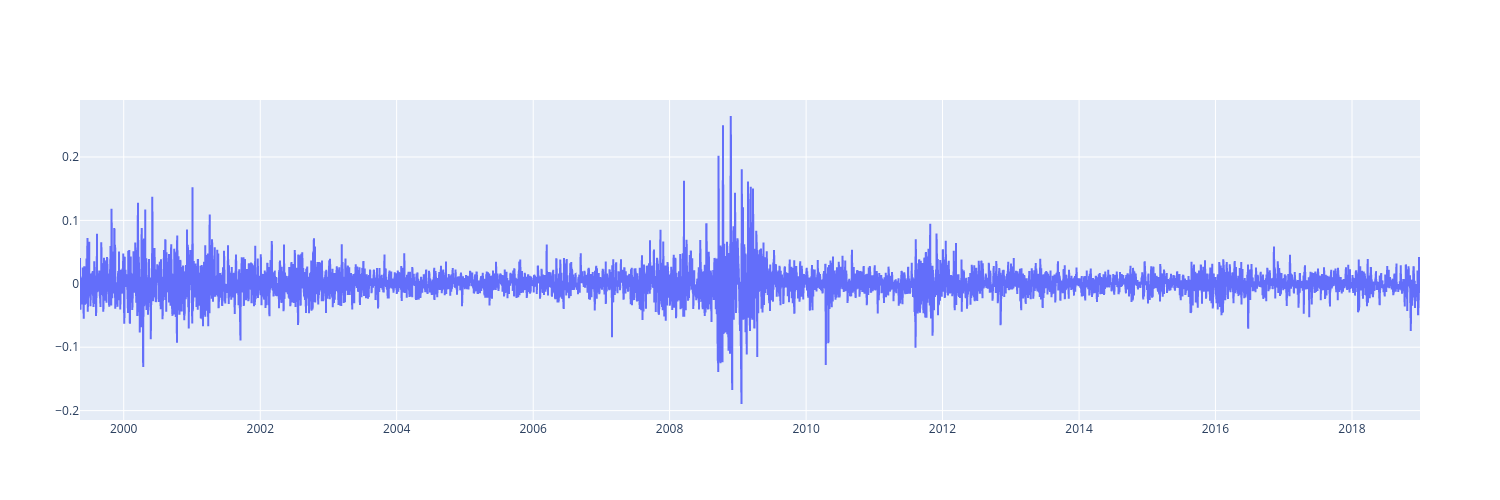}
    \caption{Price return of Goldman Sachs in 2000-2018}
    \medskip
    \small
    \label{fig:return2}
\end{figure}
 
Based on figures, price return seems a more prudent choice for target feature. Consequently, we trained our model on it.

Consider the General Electric (GE) closing prices shown in figure \ref{fig:GE-close} along with the ratio of train and test demarked by the vertical line.

  \begin{figure}[H]
        \centering
        \includegraphics[width=0.9\linewidth]{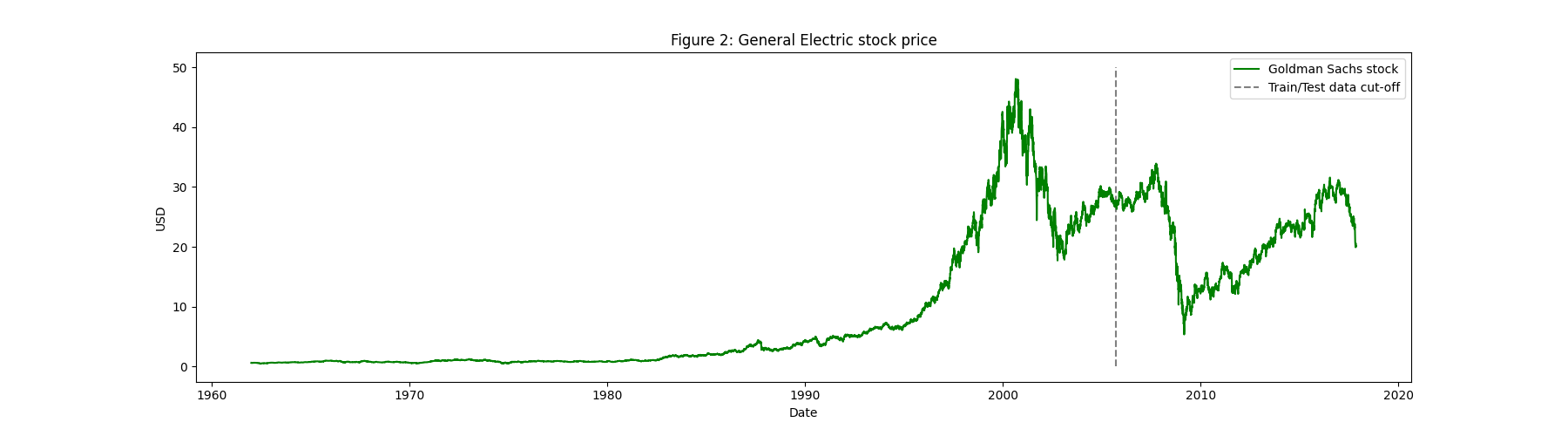}
        \caption{GE Closing Prices}
        \medskip
        \small
        \label{fig:GE-close}
    \end{figure}

At first we will report the result of applying LSTM on GE's closing price and then that on GE's return.
Applying LSTM on GE will yield the lagged forecasting again. In figure \ref{fig:GE-trainvar}, train and test errors are shown and also, in figure \ref{fig:GE-pred}, prediction is shown along with test data.

 \begin{figure}[H]
    \centering
    \includegraphics[width=0.5\linewidth]{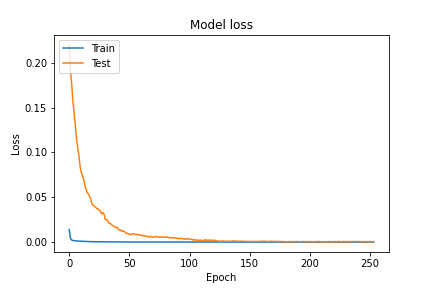}
    \caption{Model train and validation error}
    \medskip
    \small
    \label{fig:GE-trainvar}
\end{figure}

\begin{figure}[H]
    \centering
    \includegraphics[width=0.5\linewidth]{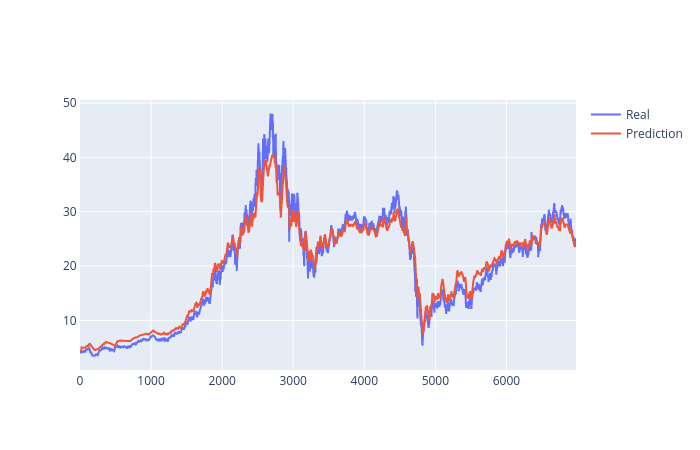}
    \caption{Test vs Prediction}
    \medskip
    \small
    \label{fig:GE-pred}
\end{figure}

By enlarging our graph's scale and zoom into days of the prediction, the lag appears again. See figures \ref{fig:close-zoom1} and \ref{fig:close-zoom2}

\begin{figure}[H]
    \centering
    \includegraphics[width=0.5\linewidth]{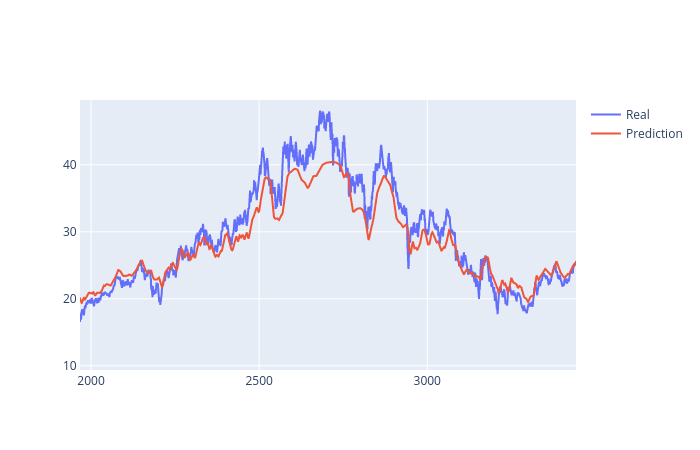}
    \caption{Test vs Prediction}
    \medskip
    \small
\end{figure}

\begin{figure}[H]
    \centering
    \includegraphics[width=0.5\linewidth]{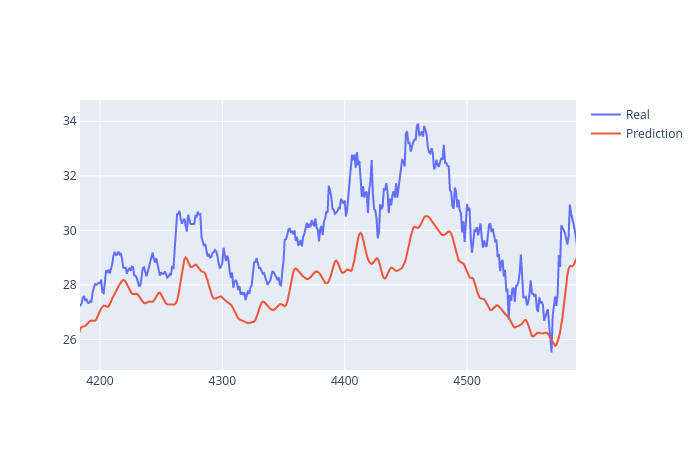}
    \caption{Test vs Prediction}
    \medskip
    \small
\end{figure}

The MSE and R-squared values of prediction are 0.0013767663656408873 and 0.9699375342029654 respectively.

Applying the exact same model on GE price return yields a result quite different than that on closing prices. In figure \ref{fig:GE-RET-predvstest}
 prediction is illustrated along with test data while in figure \ref{fig:GE-RET-trainval}, validation and train errors are illustrated.  
\begin{figure}[H]
    \centering
    \includegraphics[width=0.5\linewidth]{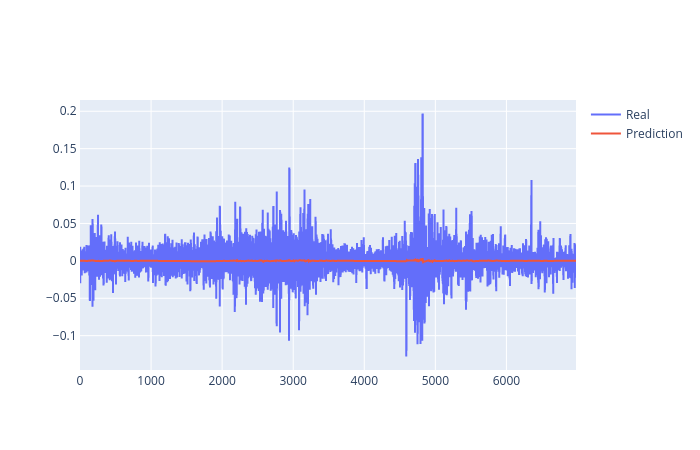}
    \caption{Test vs Prediction}
    \medskip
    \small
    \label{fig:GE-RET-predvstest}
\end{figure}

\begin{figure}[H]
    \centering
    \includegraphics[width=0.5\linewidth]{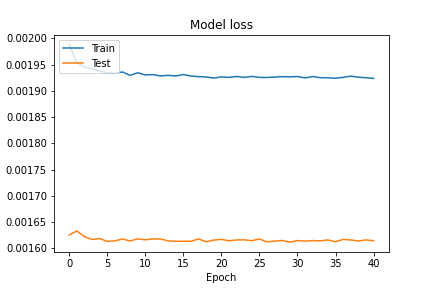}
    \caption{Test vs Prediction}
    \medskip
    \small
    \label{fig:GE-RET-trainval}
\end{figure}

For a more clear sight of prediction, look at the following figure.

\begin{figure}[H]
    \centering
    \includegraphics[width=0.5\linewidth]{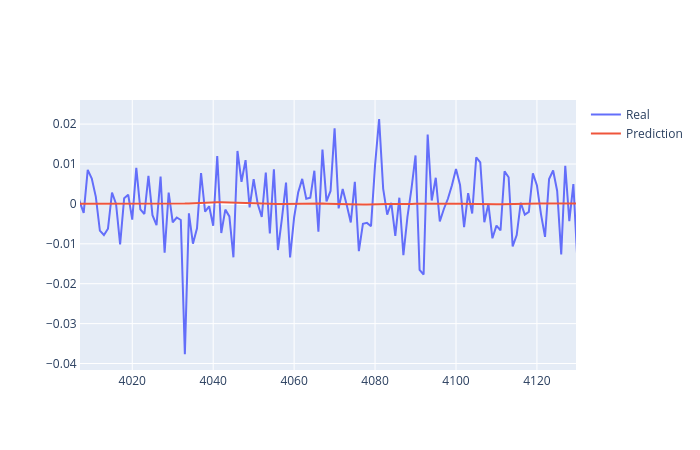}
    \caption{Test vs Prediction}
    \medskip
    \small
\end{figure}

MSE error of prediction and test is  0.002235389706528071.

\subsection{ARMA} 
Forecasting GS price with ARIMA with orders $(p = 5, q = 2)$ yields the following result:
 \begin{figure}[H]
        \centering
        \includegraphics[width=.5\linewidth]{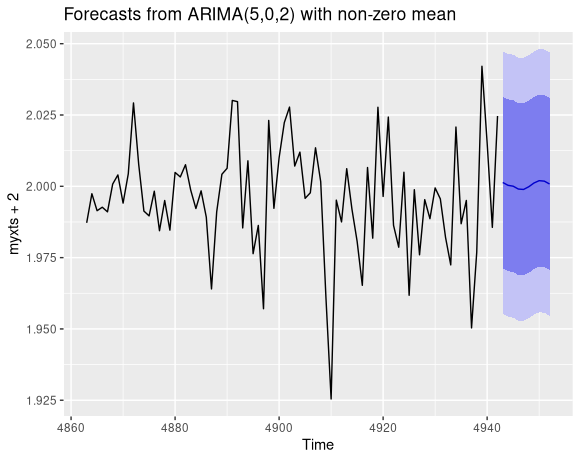}    
        \caption{ARIMA Forecasting with $h=10$}
        \medskip
        \small
    \end{figure}    
In figure \ref{fig:autoarima}, coefficients and error of fitting ARIMA is presented.
Residuals from the forecasting model along with its ACF are also illustrated in figure \ref{auto-arima-ret-res}. 

 \begin{figure}[H]
        \centering
        \includegraphics[width=.5\linewidth]{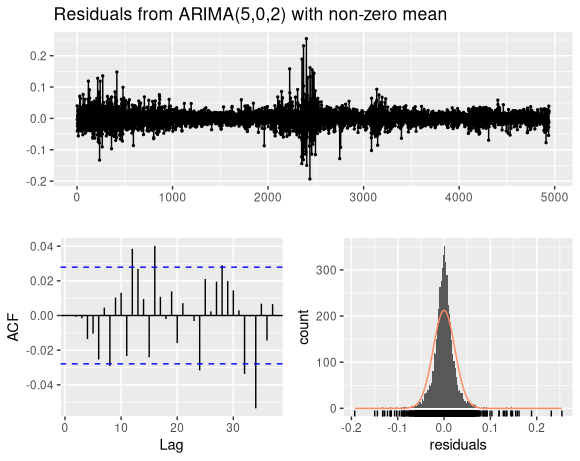}    
        \caption{Residuals}
        \medskip
        \small
        \label{auto-arima-ret-res}
    \end{figure}


\section{Appendix} \label{appx}


\subsection{ARMA: Creating stationary time series} \label{appx:stationary}

\subsubsection{Definition of Stationarity}
There are two kinds of stationarity: \textit{strictly stationary} and \textit{weekly stationary}. 
As the former is not in line with real-world datasets, we will only define the latter kind.

A time series is \textit{weakly stationary} if both the mean of $r_t$ and the covariance between $r_t$ and $r_{t-\ell}$ are time invariant, where $\ell$ is an arbitrary integer. More precisely, $\{r_t\}$ is weakly stationary if (a) $E(r_t) = \mu,$ which is a constant, and (b) $Cov(r_t,r_{t-\ell}) = \gamma_\ell$, which only depends on $\ell$. In practice, suppose that we
have observed $T$ data points $\{r_t | t=1,...,T\}$. The weak stationarity implies that the time plot of the data would show that the $T$ values fluctuate with constant variation around a fixed level.

Thus, time series with trends, or with
seasonality, are not stationary---the trend and seasonality will affect the value of the time series at different times. On the other hand, a white noise series is stationary---it does not matter when one observes it, it should look much the same at any point in time.
Some cases can be confusing---a time series with cyclic behavior (but with no trend or seasonality) is stationary.

To remove variations in mean (manifesting as seasonality and trend) and in variance (manifesting as shifts in height) to make the time series stationary, there is a general principle:
\begin{quote}
"The logarithm stabilizes the variance, while the seasonal differences remove the seasonality and trend." \cite{forepp}
\end{quote}

\subsubsection{Seasonal Plots} 
Before removing seasonality, we have to detect whether there it exists in data. For detecting seasonality, a practical approach is to observe seasonal plots. For this purpose, monthly data seasonal plots of monthly data are illustrated in linear and polar forms in figures \ref{fig:seas-month} and \ref{fig:seas-month-polar} respectively.

\begin{figure}[H]
    \centering
    \begin{minipage}{0.45\textwidth}
        \centering
        \includegraphics[height=6cm]{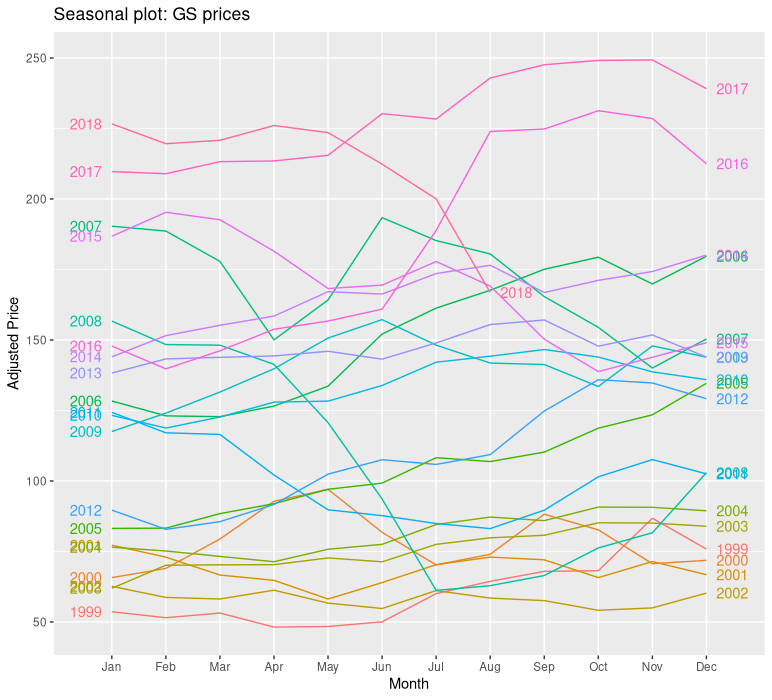} 
        \caption{Seasonal Plot of GS}
        \label{fig:seas-month}
    \end{minipage}\hfill
    \begin{minipage}{0.45\textwidth}
        \centering
        \includegraphics[height=6cm]{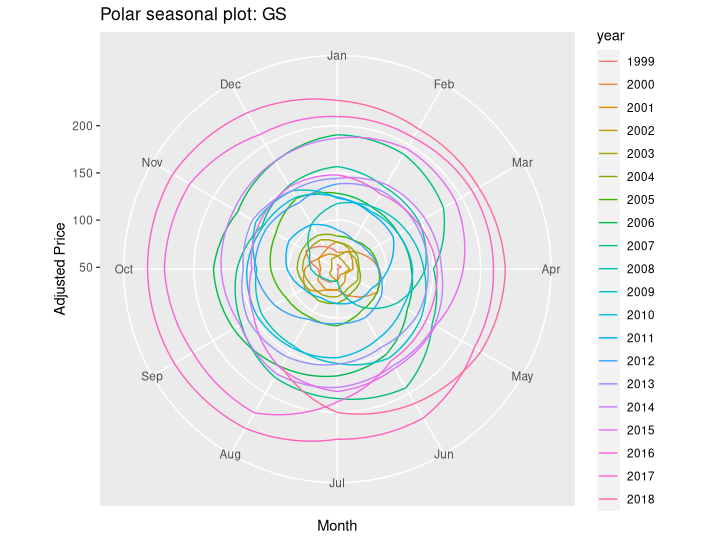} 
        \caption{Polar seasonal plot of GS}
        \label{fig:seas-month-polar}
    \end{minipage}
\end{figure}

Based on the seasonal plots, we can conclude the price is not inclined to a certain month or season, meaning that there is no apparent alteration that would occur constantly over regular periods. Therefore, it seems no seasonal differencing (difference of terms that are more than 1 lag distant) is required.

\subsubsection{Unit Root Tests}
Now, we need to determine whether usual differencing (difference of each two consequtive term) is required or not. For this purpose ,we will use \textit{Unit root tests}. We will use log of prices (so that price's volatility would be time-invariant) as well as price returns for the test.

To test whether the log price of $p_t$ follows a random walk or a random walk with drift (time trend), we employ the models

\begin{subequations}
\begin{align}
p_t & = \phi_1 p_{t-1} + e_t \\
p_t & = \phi_0 + \phi_1 p_{t-1} + e_t, 
\end{align}
\end{subequations}

where $e_t$ denotes the error term. We consider the null hypothesis $H_0; \phi_1 = 1$ (non-stationary) versus the alternative hypothesis: $H_a : \phi_1 < 1$ (stationary). This is the well-known unit-root testing problem.
A number of unit root tests are available, which are based on different assumptions and may lead to conflicting answers. In our analysis, we use both \textit{Dickey and Fuller} test (Dickey and Fuller (1979)) and 
\textit{Kwiatkowski-Phillips-Schmidt-Shin} (KPSS) test (Kwiatkowski, Phillips, Schmidt, \& Shin, 1992). 

\subsubsection{DF Test}
The Dickey-Fuller test is testing if $\phi = 0$ in the following model of the data:

$$p_t = \alpha + \beta t + \phi p_{t-1} + e_t$$

which can be written as

$$\Delta p_t = p_t - p_{t-1} = \alpha + \beta t + \gamma p_{t-1} + e_t$$

where $p_t$ is the price of stocks and $e_t$ is error. It is written this way so that we can perform a linear regression of $\delta p_t$ against $t$ and $p_{t-1}$ and test if $\gamma$ is different from 0. if $\gamma = 0$, it implies that the data is a random walk process. Otherwise, if $-1 < 1 + \gamma < 1$, then the data is stationary.

There is also the Augmented Dickey-Fuller test which allows for higher-order autoregressive processes by including  $\Delta p_{t-q}$ in the model. The main concern is still whether $\gamma = 0$ or not.

$$\Delta p_t = \alpha + \beta t + \gamma  p_{t-1} + \delta_1 \Delta p_{t-1} + \delta_2 \Delta p_{t-2} + \cdots + \Delta p_{t-q}$$

The null hypothesis for both tests is that the data is non-stationary. We want to \textbf{reject} the null hypothesis for this test, hence a p-value that is equal or less than 0.05.

Now, consider the log price of GS. Testing for a unit root is relevant if one wishes to verify empirically that the prices follows a random walk with drift. In figure \ref{fig:df}, the result of applying DF test on the data is indicated:

 \begin{figure}[H]
        \centering
        \includegraphics[width=.7\linewidth]{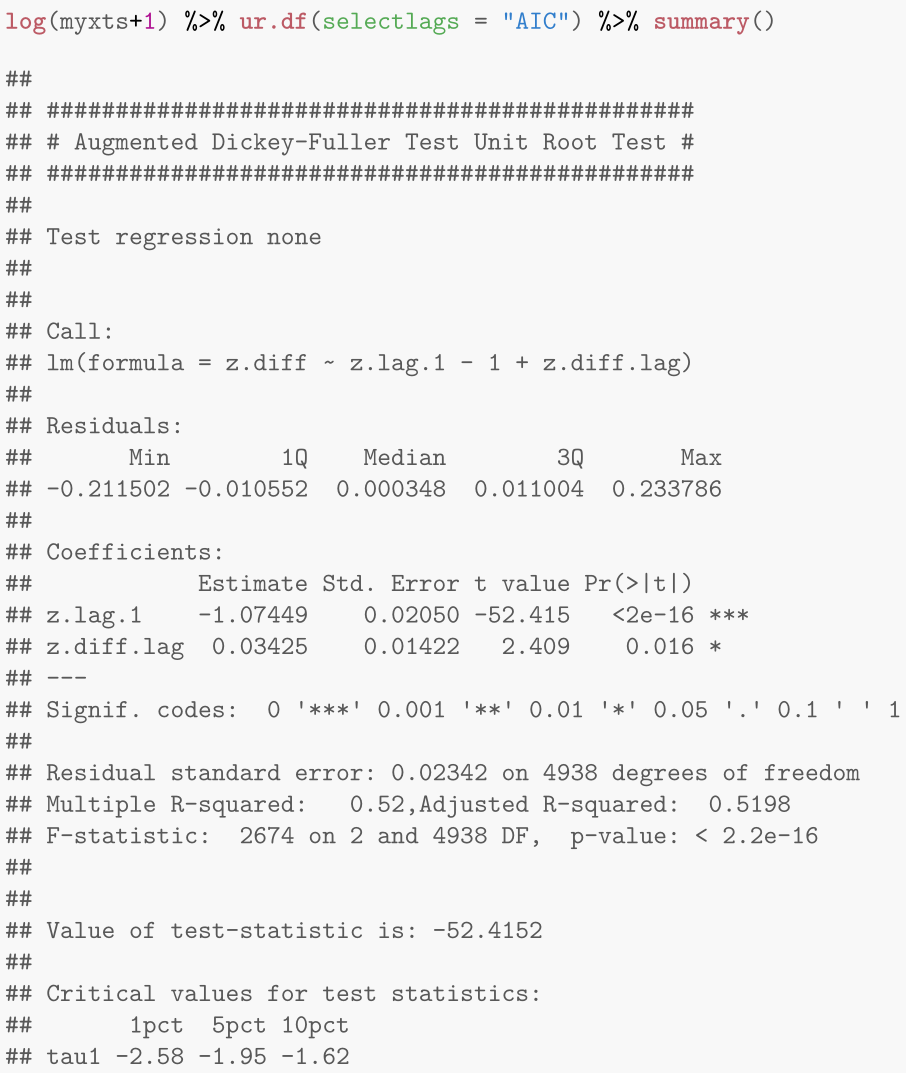}    
        \caption{Result of Dicky Fuller test on GS returns (implemented in \code{R})}
        \medskip
        \small
        \label{fig:df}
    \end{figure}

The coefficient part of the summary indicates that \code{z.lag.1} is different than 0 (so stationary).
Pay heed that the test statistic is \textbf{less} than the critical value for \code{tau3} at 5 percent. This implies that the null hypothesis is rejected at $\alpha = 0.05$, a standard level for significance testing.

\subsubsection{KPSS Test} 
 Unlike most unit root tests, \textit{Kwiatkowski et al.} provide straightforward test of the null hypothesis of trend stationarity against the alternative of a unit root. In fact, the null hypothesis is that the data are
stationary, and we look for evidence that refutes the null hypothesis. Consequently, small p-values (e.g., less than 0.05) suggest that differencing is
required. 

The model assumes that $p_t$ is sum of a deterministic time trend, a random walk and a stationary residual:

$$p_t = \beta t + (r_t + \alpha) + e_t,$$

where $r_t = r_{t-1} + u_t$ is a random walk (the initial value $r_0 = \alpha$ serves as an intercept) and $e_t$ are i.d.d with zero mean and variance ${\sigma^2}_u$.

Under this form, The null and the alternative hypotheses would be as following:

$H_0 : p_t$ is trend (or level stationary) OR ${\sigma^2}_u$

$H_1: p_t$ is a unit root process.

In figure \ref{fig:kpss}, the result of applying DF test on the data is indicated:

 \begin{figure}[H]
        \centering
        \includegraphics[width=.7\linewidth]{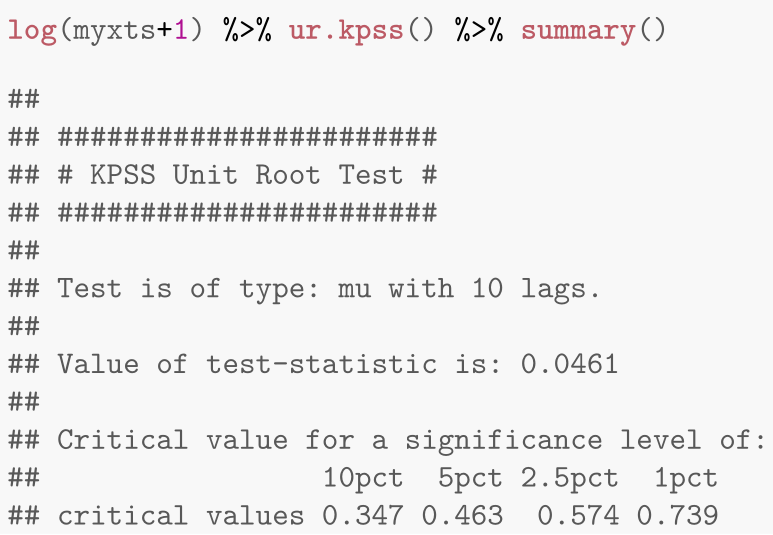}    
        \caption{Result of KPSS test on GS returns (implemented in \code{R})}
        \medskip
        \small
        \label{fig:kpss}
    \end{figure}   

The test statistic is tiny. So, we can conclude that data is stationary.

\subsection{ARMA: Choosing orders} \label{appx:order}

Determining the orders of ARMA can be a challenging task. One of the most common approaches is to observe ACF and PACF diagrams to choose the best value for $p$ (AR component) and $q$ (MA component).

\subsubsection{ACF and PACF} 
We can calculate the correlations for a time series observations with previous \textit{time-step}s, called lags. Since the correlation of the time series observations is calculated with values of the same series at previous times, this is called a serial correlation, or an \textit{autocorrelation}. A plot of the autocorrelation of a time series by lag is called the \textit{AutoCorrelation Function}, or the acronym ACF.
More precisely, the lag-$\ell$ autocorrelation is defined as
$$\rho_l = \frac{Cov(p_t, p_{t-\ell}}{\sqrt{Var(p_t)Var(p_{t-\ell})}} = \frac{Cov(p_t,p_{t-\ell})}{Var(p_t)} = \frac{\gamma_\ell}{\gamma_0},$$

where the property $Var(p_t) = Var(p_{t-\ell})$ for a weakly stationary series is used. From the definition, we have $\rho_0 = 1$, $\rho_{\ell} = \rho_{-\ell}$, and $-1 \leq \rho_{\ell} \leq 1$.

A partial autocorrelation is a summary of the relationship between an observation in a time series with observations at prior time-steps, such that the relationships of intervening observations removed. The partial autocorrelation at a certain lag is the correlation that results after removing the effect of any correlations due to the terms at shorter lags.

It is essential to take PACF into account beside ACF as well since ACF might reflect false correlation. To shed light on why this occurs, suppose two random variables, say $X$ and $Y$ might manifest positive correlation with each other, i.e. $Cor(X,Y) > 0$. However, this correlation might stem from the fact that both of them are dependent to an intermediary third variable, say $Z$. Therefore, if we compute the correlation of $X$ and $Y$ conditioned on $Z$, we may reach a value near zero, i.e. $Cor(X,Y | Z) \approx 0$. Similarly, PACF computes the correlation of different lags with considering the intermediary ones.\\

\begin{figure}[ht] 
  \begin{minipage}[b]{0.5\linewidth}
    \centering
    \includegraphics[height=4cm]{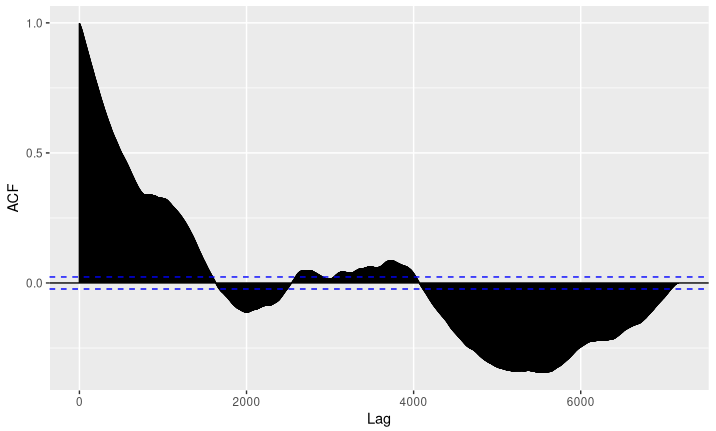}    
        \caption{ACF of GS adjusted}
    \vspace{4ex}
    \label{fig:ACF}
  \end{minipage}
  \begin{minipage}[b]{0.5\linewidth}
    \centering
    \includegraphics[height=4cm]{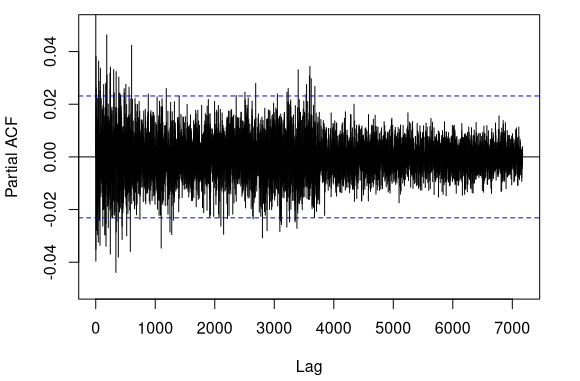}    
        \caption{PACF of GS adjusted}
    \vspace{4ex}
    \label{fig:PACF}
  \end{minipage} 
  \begin{minipage}[b]{0.5\linewidth}
    \centering
    \includegraphics[height=4cm]{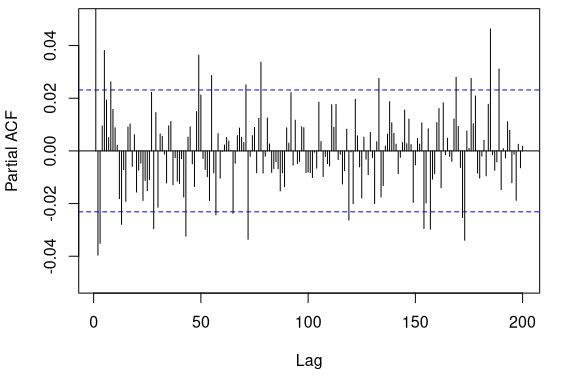}    
        \caption{PACF of GS adjusted (200 lags)}
    \vspace{4ex}
    \label{fig:PACF-200}
  \end{minipage}
  \begin{minipage}[b]{0.5\linewidth}
    \centering
    \includegraphics[height=4cm]{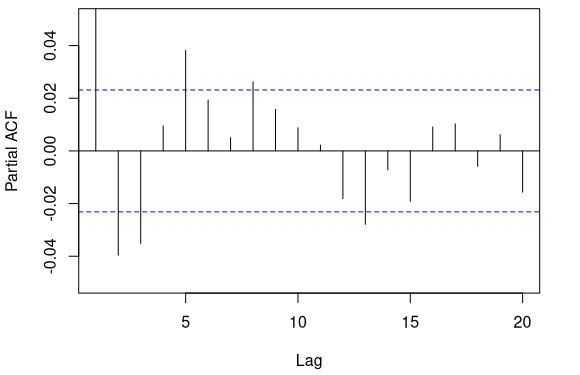}    
        \caption{PACF of GS adjusted (20 lags)}
    \vspace{4ex}
    \label{fig:PACF-20}
  \end{minipage} 
\end{figure}

In figures \ref{fig:ACF} and \ref{fig:PACF} ACF and PACF of GS stock's adjusted prices are illustrated and in figures \ref{fig:PACF-200} and \ref{fig:PACF-20}, PACF of first 200 lags and 20 lags are shown.

Let us look at the same diagrams for price return as well:

\begin{figure}[ht] 
  \label{ fig7} 
  \begin{minipage}[b]{0.5\linewidth}
    \centering
    \includegraphics[height=4cm]{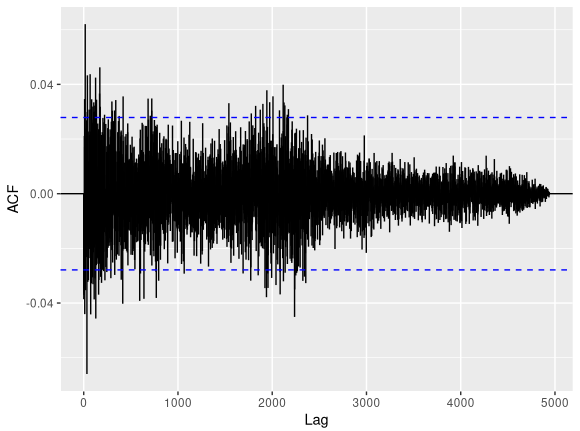}    
        \caption{ACF of GS price return}
    \vspace{4ex}
    \label{fig:ACF-RET}
  \end{minipage}
  \begin{minipage}[b]{0.5\linewidth}
    \centering
    \includegraphics[height=4cm]{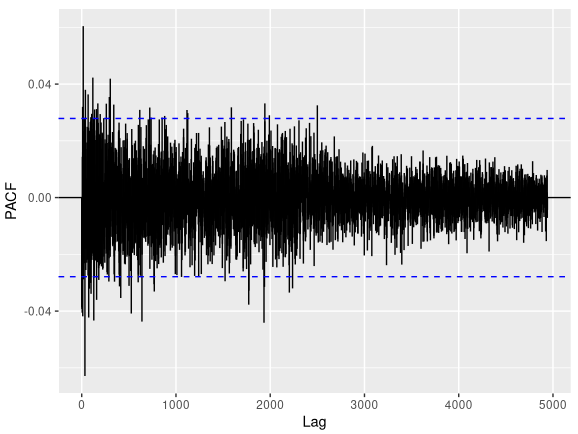}    
        \caption{PACF of GS price return}
    \vspace{4ex}
    \label{fig:PACF-RET}
  \end{minipage} 
  \begin{minipage}[b]{0.5\linewidth}
    \centering
    \includegraphics[height=4cm]{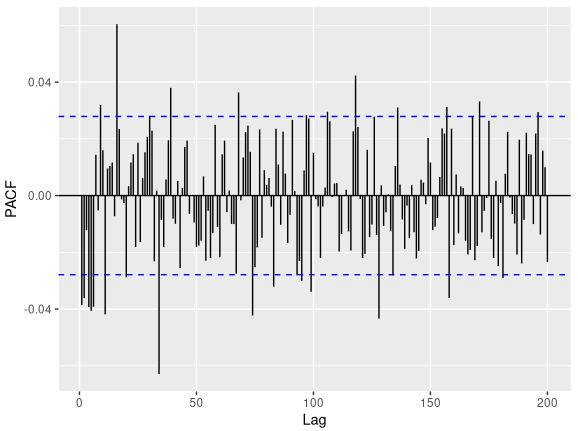}    
        \caption{PACF of GS price return (200 lags)}
    \vspace{4ex}
    \label{fig:PACF-RET-200}
  \end{minipage}
  \begin{minipage}[b]{0.5\linewidth}
    \centering
    \includegraphics[height=4cm]{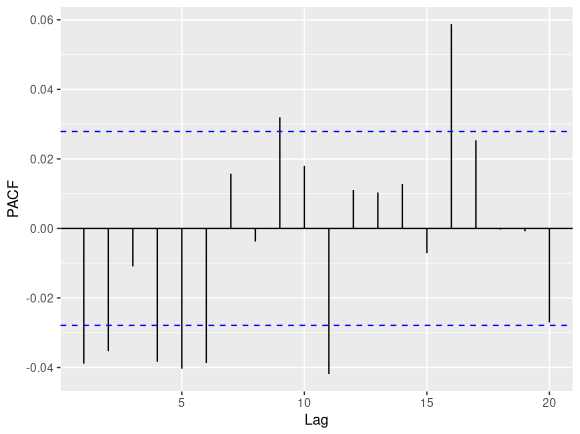}    
        \caption{PACF of GS price return (20 lags)}
    \vspace{4ex}
    \label{fig:PACF-RET-20}
  \end{minipage} 
\end{figure}

 Based on the ACF and PACF, we choose different orders. The package \code{auto.arima} in R provides an automatic detection of orders before forecasting. In figure \ref{fig:autoarima}, result of applying \code{auto.arima} on stock returns is shown. 

\begin{figure}[H]
        \centering
        \includegraphics[width=.7\linewidth]{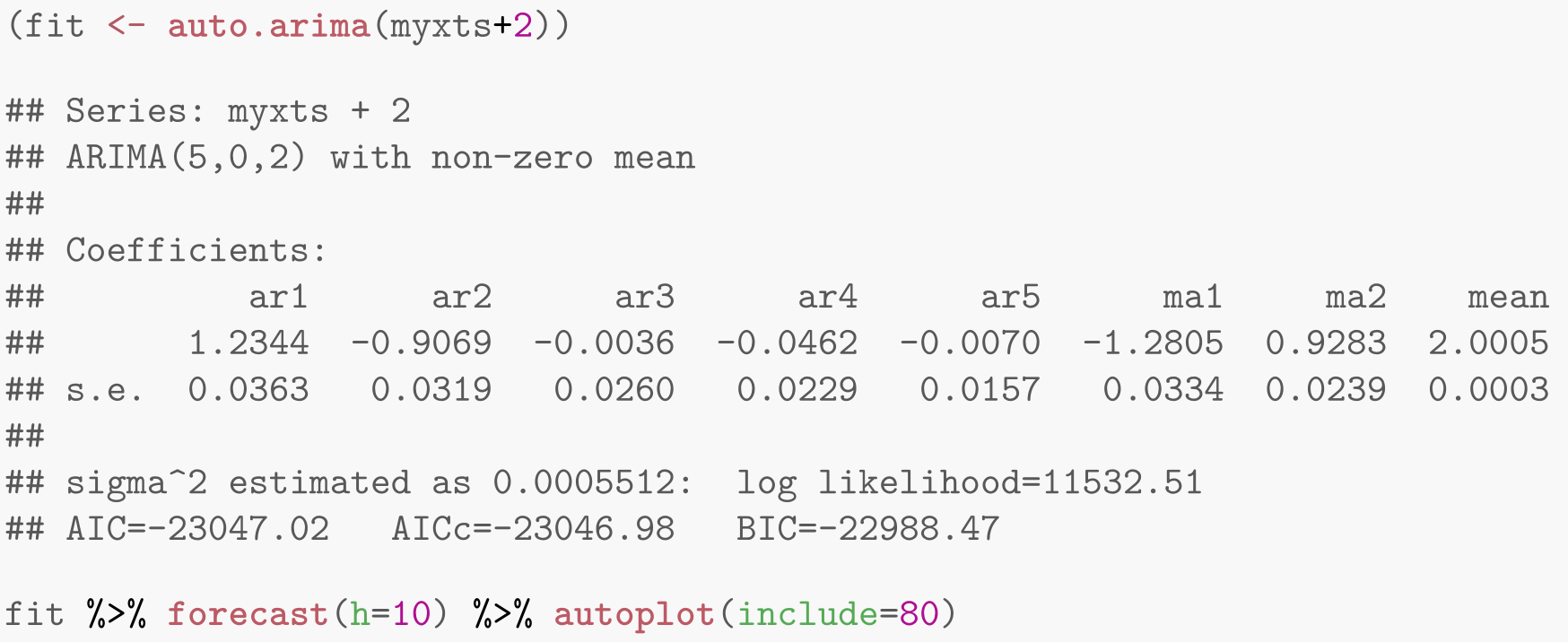}    
        \caption{Result of fitting \code{auto.arima} to GS price returns (implemented in \code{R})}
        \medskip
        \small
        \label{fig:autoarima}
    \end{figure}        
    
The result of fit is choosing 5 and 2 as values of $p$ and $q$ respectively. Our own selected orders did not perform better than the proposed order from the automatic procedure because the AIC and BIC errors of any other chosen order were more.

\subsection{Exponential Moving Average} \label{appx:EMA}
 Consider $S$, a sequence of numbers that is noisy. Figure \ref{fig:guassian}, a cosine function is shown to which the Guassian noise is added. 

 \begin{figure}[H]
        \centering
        \includegraphics[width=.5\linewidth]{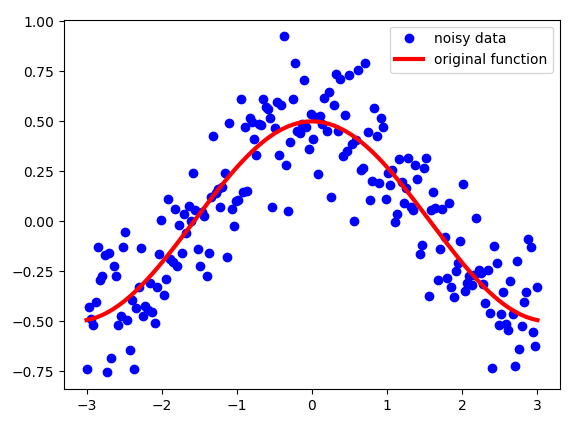}    
        \caption{A cosine function along with Guassian noise}
        \medskip
        \small
        \label{fig:guassian}
    \end{figure}
    
Instead of using data, if we use the exponential moving average (EMA), it would look like figure \ref{fig:EMA}.

     \begin{figure}[H]
        \centering
        \includegraphics[width=.5\linewidth]{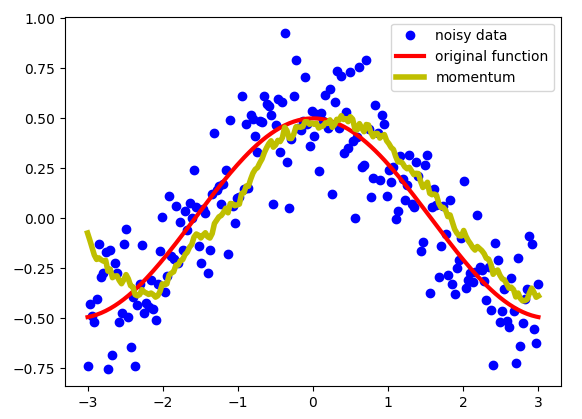}    
        \caption{Momentum is a smoother than noises and can be a fine representation}
        \medskip
        \small
        \label{fig:EMA}
    \end{figure}
    
The acquired plot is a smoother function that approaches the actual cosine function. In effect, EMA defines a new sequence $V$ by the following formula:

\begin{equation} \label{eq:EMA}
 V_t = \beta V_{t-1} + (1-\beta)S_t,   
\end{equation}

where $\beta \in [0,1].$

$\beta$ is a hyperparameter that is also used in momentum SGD. Figure \ref{fig:EMAbeta} illustrates the effect of using different amounts of $\beta$.

   \begin{figure}[H]
        \centering
        \includegraphics[width=.5\linewidth]{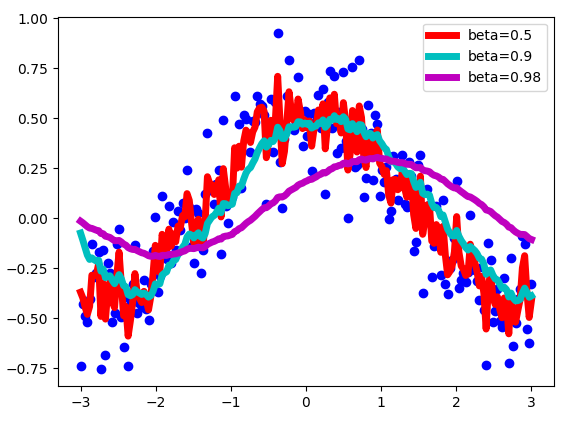}    
        \caption{The effect of $\beta$ on momentum. Smaller $\beta$ yields more accuracy while bigger $\beta$ yields smoother function}
        \medskip
        \small
        \label{fig:EMAbeta}
    \end{figure}

One of the key property of EMA is that it assigns different weights to different index of data. The more recent a sample is, the more stress EMA attributes to it and earlier data might even be discarded at some point of time. To prove this, we expand the recurrent equation of \ref{eq:EMA} as follows:

\begin{equation} 
\begin{split}
V_t & = \beta V_{t-1} + (1-\beta) S_t \\
V_{t-1} & = \beta V_{t-2} + (1-\beta) S_{t-1}\\
V_{t-2} & = \beta V_{t-3} + (1-\beta) S_{t-2}
\end{split}
\end{equation}

If we substitute each line in \ref{eq:EMA} we will obtain the following equation:

\begin{equation} 
\begin{split}
V_t & = \beta \big(\beta (\beta V_{t-3} + (1-\beta) S_{t-2}) + (1 - \beta) S_{t-1}\big) + (1-\beta) S_t \\
& = \beta^k V_{t-k} + (1-\beta) \sum_{i=0}^{k} \beta^i S_{t-i},
\end{split}
\end{equation}

where $k$ is the lag (window) of time that we consider. Since $\beta < 1$, as it increases, contribution of previous observations of sequence will be diminished.


\subsection{Process of data in LSTM (Keras)} \label{appx:keras-rnn}
To elucidate the execution process, consider the following example. Suppose that the \textit{batch-size} is 6, RNN size is 7, number of \textit{time-step}s that are included in one input line is 5 and the number of features in each \textit{time-step} is 3. If this is the case, the input tensor(matrix) shape for one batch would appear as follows:
\\
\\
Tensor shape of one batch $= (6,5,3)$
\\
\\
In figure \ref{fig:batch-data} , the data inside a batch is illustrated.

\begin{figure}[H]
    \centering
    \includegraphics[width=.5\linewidth]{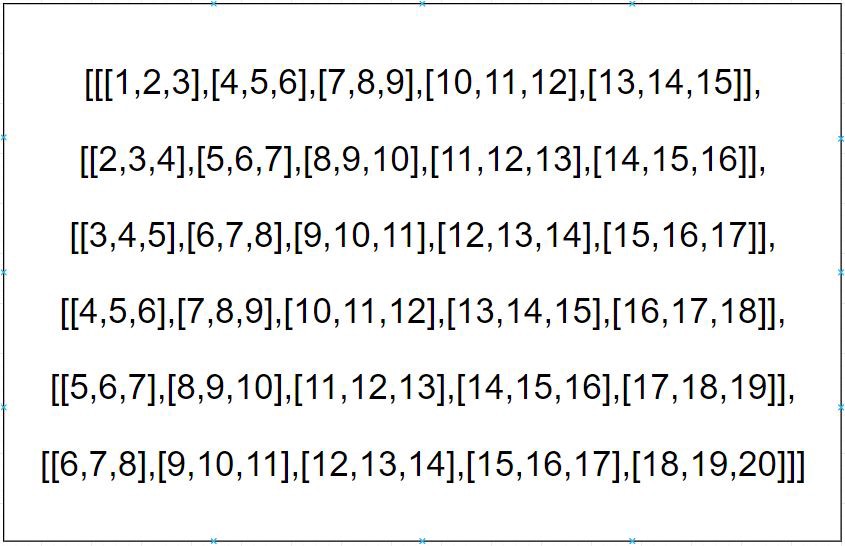}    
    \caption{Data representation inside a batch of data}
    \medskip
    \small
    \label{fig:batch-data}
\end{figure}
    
After the RNN starts processing the data by unrolling its structure, the outputs are produced in the manner shown in figure \ref{fig:RNN-unrolled}.
    
\begin{figure}[H]
    \centering
    \includegraphics[width=.98\linewidth,height=10cm]{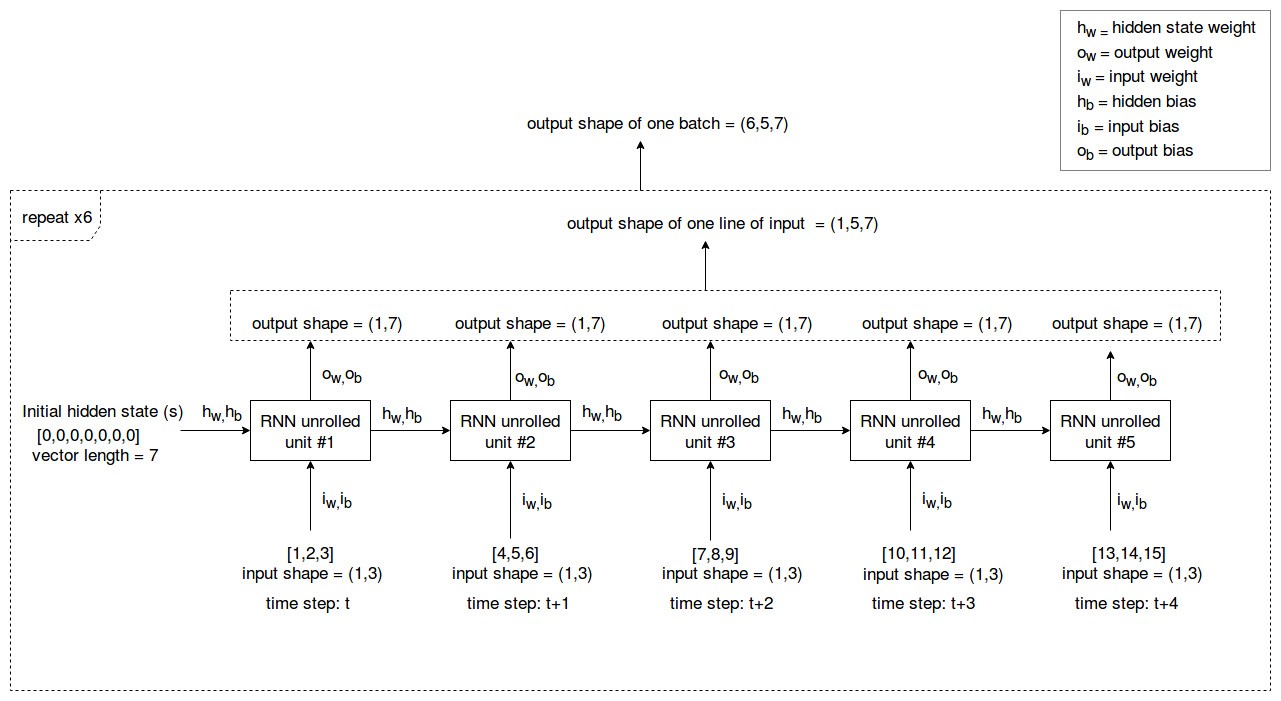}    
    \caption{Overall structure of an unrolled RNN}
    \medskip
    \small
    \label{fig:RNN-unrolled}
\end{figure}

In what follows, the processing of a batch and then that of a single line of input is explained.

\subsubsection{Processing a batch}
When batch of data is fed into the RNN cell, this cell starts the processing from the 1st line of input. Likewise, the RNN cell will sequentially process all the input lines in the batch of data that was fed and give one output at the end which includes all the outputs of all the input lines.

\subsubsection{Processing a single line of input}
Since we defined “number of steps” as 5, the RNN cell has been unrolled 5 times. The execution process is as follows:
\begin{steps}

    \item First, the initial hidden state $(S)$, which is typically a vector of zeros and the hidden state weight $(h)$ is multiplied and then the hidden state bias is added to the result. In the meantime, the input at the \textit{time-step} $t ([1,2,3])$ and the input weight $i$ is multiplied and the input bias is added to that result. We can obtain the hidden state at \textit{time-step} t by sending the addition of the above two results through an activation function, typically tanh $(f)$. 
    $$S_t = f\Big((i_w [1,2,3] + i_b) + (h_w S_{initial} + h_b)\Big)$$
    
    \item To obtain the output at \textit{time-step} t, the hidden state $(S)$ at \textit{time-step} $t$ is multiplied by the output weight $(O)$ at \textit{time-step} $t$ and then the output bias is added to the result.
    
    $${Output}_t = (S_t O_w) + O_b$$
    
    \item When calculating the hidden state at time step $t+1$, the hidden state $(S)$ at \textit{time-step} $t$ is multiplied by the hidden state weight $(h)$ and the hidden state bias is added to the result. Then as mentioned before the input at \textit{time-step} $t+1 ([4,5,6])$ will get multiplied by the input weight $(i)$ and the input bias will be added to the result. These two results will then be sent through an activation function, typically tanh $(f)$.
    
    $$S_{t+1} = f\big((i_w [4,5,6] + i_b) + (h_w S_t + h_b)\big)$$
    
    \item To obtain the output at \textit{time-step} $t+1$, the hidden state $(S)$ at \textit{time-step} $t+1$ is multiplied by the output weight $(O)$ at \textit{time-step} $t+1$ and then output bias is added to the result. While producing the output of \textit{times-step} $t+1$, it uses not only the input data of \textit{time-step} $t+1$, but also uses information of data in \textit{time-step} $t$ via the hidden state at \textit{time-step} $t+1$.
    
    \item Finally, the explained process is repeated for all of the \textit{time-step}s.
\end{steps}

After processing all time steps in one line of input in the batch, we will have 5 outputs of shape (1,7). So when all these outputs are concatenated together. the shape becomes (1,5,7). When all the input lines of the batch are done processing we get 6 outputs of size (1,5,7). Thus, the final output of the whole batch would be (6,5,7).

\section{Conclusion} \label{conc}

From reported figures and accuracy as well as the forecasting lag problem, we can conclude that provided that the target feature is price return, the model will be immune to falling into the forecasting lag trap. On the other hand, designating price return as feature is at the expense of the model losing sight of the trends and magnitude of change over time. Consequently, it seems that forecasting stocks' prices by their historical data as the sole feature by LSTM seems a quite cumbersome and perhaps impossible task due to the mentioned issues. This is in contradiction with the suggestive theory that LSTM seems like the most suitable model for predicting prices and sequence data in general.

\clearpage
\chapter{Neural Networks in Mathematical Framework} \label{sec:nnmf}
\clearpage
%

In the previous chapter, we implemented neural networks on our problem. Although neural networks are a great candidate for a wide variety of tasks, at this point we do not have a full rigorous understanding of why Deep Neural Networks work so well, and how exactly to construct neural networks that perform well for a specific problem. ANNs are partly considered a blackbox without being formalized in a proper mathematical framework. \cite{dnnmf} takes the first steps towards forming this rigorous understanding. In this chapter, the same framework and notation that used in \cite{dnnmf} will be used here as well.\\

At first, the mathematical preliminaries for understanding the theory behind NNs are pvoided. Then, the desired framework is developed for a generic NN and gradient descent algorithm is expressed within that framework. At last, this framework is extended to specific architectures of NNs, from which we chose RNNs to explain and express gradient descent within their structure. \cite{dnnmf} proves theorems of this part for classification case and cite some of the theorems to those relevant to feedforward networks. On the other hand, we proved the theorems for the regression case (which includes forecasting), and prove theorems of RNNs as well as expressing gradient descent algorithm independently from feedforward networks. At the end, narrow our scope to specific RNNs, such as Vanilla RNN and LSTM.

\section{Mathematical Preliminaries} \label{sec:math tools}
Most approaches to describing DNNs rely upon decomposing the parameters and inputs into scalars, as opposed to referencing their underlying vector spaces, which adds a level of awkwardness into their analysis. On the other hand, the framework that \cite{dnnmf} develops strictly operates over these vector spaces, affording a more natural mathematical description of DNNs once the objects that we use are well defined and understood. \\

To set foot in the desired mathematical framework of neural networks, we introduce prerequisite mathematical concepts and notation for handling generic vector-valued maps. Although some of the posed concepts in this section are quite basic, it is necessary to solidify the symbols and language that we will use throughout this chapter so that the notation would be crystal clear and without any ambiguity. \cite{dnnmf} \\

\subsection{Linear Maps, Bilinear Maps, and Adjoints}

\newtheorem{definition}{Definition}[section]
\theoremstyle{definition}
\begin{definition}[bilinear map]

A map $\beta: V \times W \to Z$ is bilinear if $V,W,Z$ are vector spaces and for each fixed $v \in V$ the map $\beta(v,.):W \to Z$ is linear, while for each fixed $w \in W$ the map $\beta(.,w):V \to Z$ is linear. Examples are
\begin{itemize}
    \item Ordinary real multiplication $(x,y) \mapsto xy$ is a bilinear map $\mathbb{R} \times \mathbb{R} \to \mathbb{R}$.
    \item The dot product is a bilinear map $\mathbb{R}^n \times \mathbb{R}^n \to \mathbb{R}$ (Hfowever, it's not a linear map while considering a multivariate function.)
    \item The matrix product is a bilinear map $M(m \times k) \times M(k \times n) \to M(m \times n)$
\end{itemize}
\end{definition}

Let us start by considering three finite-dimensional and real inner product spaces
$E_1$, $E_2$ , and $E_3$, with the inner product denoted $\langle\,,\rangle$ on each space. We will denote the space of linear maps from $E_1$ to $E_2$ by $L(E_1; E_2)$, and the space of bilinear maps from $E_1 \times E_2$ to $E_3$ by $L(E_1, E_2; E_3)$. For any bilinear map $B \in L(E_1, E_2; E_3)$ and any vector $e_1 \in E_1$, we can define a linear map $(e_1 \lrcorner B) \cdot L(E_2; E_3)$ as 

\begin{center}
$(e_1 \lrcorner B) \cdot e_2 = B(e_1, e_2)$
\end{center}

for all $e_2 \in E_2$. Similarly, for any $e_2 \in E_2$ , we can define a linear map $(B \cdot e_2)\in L(E_1; E_3)$ as

\begin{center}
$(B \llcorner e_2) \cdot e_1 = B(e_1, e_2)$
\end{center}

for all $e_1 \in E_1$. We refer to the symbols $\lrcorner$ and $\llcorner$ as the \emph{left-hook} and \emph{right-hook}, respectively.
Intuitively, each hook holds the corresponding variable (first or second) constant while applying the linear map to the other variable.
\\
\\
We will use the standard definition of the adjoint $L\mbox{*}$ of a linear map $L \in L(E_1; E_2):L\mbox{*}$ is defined as the linear map satisfying 

\begin{center}
$\langle L\mbox{*} \cdot e_2,e_1\rangle = \langle e_2,  L \cdot e_1\rangle$
\end{center}

for all $e_1 \in E_1$ and $e_2 \in E_2$. Notice that $L\mbox{*}(E_2; E_1)$ ---it is a linear map exchanging the domain and codomain of L. The adjoint operator satisfies the direction reversing property: 

\begin{center}
$(L_2 \cdot L_1)\mbox{*} = L_1\mbox{*} \cdot L_2\mbox{*}$
\end{center}

for all $L_1 \in L(E_1; E_2)$ and $L_2 \in L(E_2; E_3)$. A linear map $L \in L(E_1; E_1)$ is said to be \emph{self-adjoint} if $L\mbox{*} = L$.

As composition and identity functions in space of linear operators are corresponded to multiplication matrix and identity matrix in that of matrix spaces, adjoint is corresponded to transpose . To clarify this correspondence, suppose A is a matrix. Then

$$\langle \mathbf{A}e_1, e_2 \rangle = (\mathbf{A}e_1)^\intercal e_2 = e_1^\intercal (\mathbf{A}^\intercal e_2) = \langle e_1, \mathbf{A}^\intercal e_2 \rangle$$

\subsection{Differentiability in multivariate normed spaces}

\begin{definition}{Linear Transformation} \label{def1}
an $m \times n$ matrix $A$ with entries $a_{ij}$ defines a linear transformation $T_A: \mathbb{R}^n \to \mathbb{R}^m$ that sends $n$-space to $m$-space according to the formula 

$$T_A(v) = \sum_{i=1}^{m}\sum_{j=1}^{n}a_{ij}v_je_i$$

where $v = \sum v_je_j \in \mathbb{R}^n$ and $e_1,...,e_n$ is the standard basis of $\mathbb{R}^n$. (Equally, $e_1,...,e_m$ is the standard basis of $\mathbb{R}^m$
\end{definition}

A function of real variable $y = f(x)$ has derivative $f'(x)$ at $x$ when

\begin{equation} \label{eq:1}
\lim_{h \to 0} {\frac{f(x+h) - f(x)}{h}} = f'(x)
\end{equation}

If, however, $x$ is a vector variable, \eqref{eq:1} makes no sense. For what does it mean to divide by the vector increment h? Equivalent to \eqref{eq:1} is the condition 

$$f(x+h) = f(x) + f'(x)h + R(h) \Rightarrow \lim_{h \to 0} {\frac{R(h)}{|h|}} = 0$$

which is easy to recast in vector terms.

\begin{definition}{}
Let $f:U \to \mathbb{R}^m$ be given where $U$ is an open subset of $\mathbb{R}^n$. The function $f$ is differentiable at $u \in U$ with derivative $(Df)_u = T$ if $T:\mathbb{R}^n \to \mathbb{R}^m$ is a linear transformation and 

\begin{equation} \label{eq:2}
f(u + v) = f(u) + T(v) + R(v) \Rightarrow \lim_{|v| \to 0} {\frac{R(v)}{|v|}} = 0
\end{equation}

We say that the Taylor remainder $R$ is sublinear because it tends to $0$ faster than $|v|$.
\end{definition}

Here is how to visualize $Df$. Take $m=n=2$. The mapping $f:U \to \mathbb{R}^2$ distorts shapes nonlinearly; its derivative describes the linear part of the distortion. Circles are sent by $f$ to wobbly ovals, but they become ellipses under $(Df)_p$ . Lines are sent by f to curves, but they become straight lines under $(Df)_p$ . See figure \ref{fig:der1}.

\begin{figure}[H]
    \centering
    \includegraphics[width=.5\linewidth]{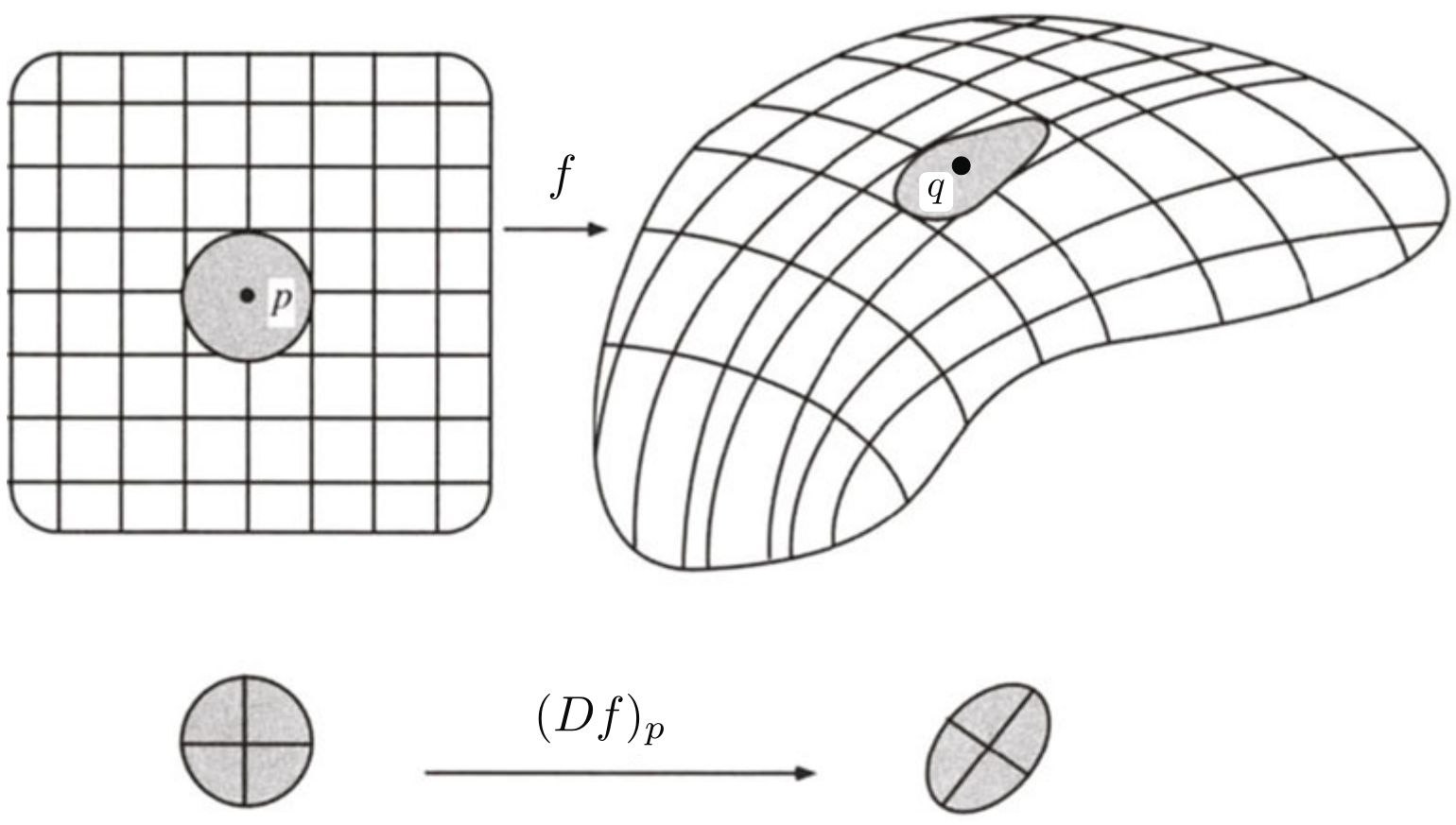}
    \caption{$(Df)_p$ is the linear part of $f$ at $p$}
    \label{fig:der1}
\end{figure}

This way of looking at differentiability is conceptually simple. Near p, f is the sum of three terms: A constant term $q = fp$, a linear term $(Df)_p v$, and a sublinear remainder term $R(v)$. Bear in mind what kind of an object the derivative is. It is not a number. It is not a vector. No, if it exists then $(Df)_p$ is a linear transformation (\ref{def1}) from the domain space to the target space. 
\\

\newtheorem{theorem}{Theorem}[section]
\begin{theorem}
If $f$ is differentiable at $p$ then it unambiguously determines $(Df)_p$ according to the limit formula, valid for all $u \in \mathbb{R}^n$,

$$(Df)_p(u) = \lim_{t \to 0} {\frac{f(p + tu) - f(p)}{t}}.$$
\end{theorem}

\begin{proof}
Let $T$ be a linear transformation that satisfies \eqref{eq:2}. Fix any $u \in \mathbb{R}^n$ and take $v = tu$. Then 

\begin{equation} 
\label{eq:3}
\frac{f(p+tu) - f(p)}{t} = \frac{T(tu) + R(tu)}{t} = T(u) + \frac{R(tu)}{t|u|}|u|.
\end{equation}

The last term converges to zero as $t \rightarrow 0$, which verifies \eqref{eq:3} Limits, when they exist, are unambiguous and therefore if $T'$ is a second linear transformation that satisfies \eqref{eq:2} then $T(u) = T'(u)$ so $T=T'$.
\end{proof}

Here is another definition of derivative which is trivial that it is equivalent to the previous one.

\begin{definition}{}
Let $E,F$ be normed vector spaces, $U$ be an open subset of $E$ and let $f: U \subset E \to F$ a given mapping. Let $u_0 \in U$. We say that $f$ is differentiable at the point $u_0$ provided that there is a bounded linear map $Df(u_0): E \to F$ such that for every $\epsilon > 0$, there is a $\delta > 0$ such that whenever $0 < \norm{u-u_0} < \delta$, we have

\begin{equation} \label{eq:def3}
\frac{\norm{f(u) - f(u_0) - Df(u_0) \cdot (u-u_0)}}{\norm{u-u_0}} < \epsilon
\end{equation}

where $\norm{\cdot}$ represents the norm on the appropriate space and where the evaluation of $Df(u_0)$ on $e \in E$ is denoted $Df(u_0) \cdot e$.
\end{definition}

This definition can also be written as

$$\lim_{u \to u_0} {\frac{f(u) - f(u_0) - Df(u_0) \cdot (u-u_0)}{\norm{u-u_0}}} = 0$$

In [\cite{nnm1}, Chapter 2, Section 3] proves that this derivative is unique in case of real numbers which is not a loss of generality.
\\
\\
In the aforementioned definitions, we exclude $u=u_0$ in taking the limit, since we are dividing by $\norm{u-u_0}$, and take the limit through those $x \in U$. More explicitly, it may be again rewritten by saying that for every $\epsilon > 0$ there is a $\delta > 0$ such that $u \in U$ and $\norm{u - u_0} < \delta$ implies

$$\norm{f(u) - f(u_0) - Df(u_0) \cdot (u - u_0)} \leq \epsilon \norm{u-u_0}$$

In this formulation we can allow $u = u_0$ since both sides reduce to zero. 
\\
Intuitively, $u \mapsto f(u_0) + Df(u_0)(u-u_0)$ is supposed to be the best affine approximation to $f$ near the point $u_0$. See figure \ref{fig:der2}. The figure indicates the equations of the tangent planes to the graph of $f$. 
\\
From the figure, we expect that there can be only one best linear approximation which is actually the case and we already stressed that the definition is unique. If we compare the definitions of $Df(x)$ and $\frac{df}{dx} = f'(x)$, we see that $Df(x)(h) = f'(x) \cdot h$ (the product of the numbers $f'(x)$ and $h \in \mathbb{R}$. Thus the linear map $Df(x)$ is just multiplication by $\frac{df}{dx}$

\begin{figure}[H]
    \centering
    \includegraphics[width=.5\linewidth]{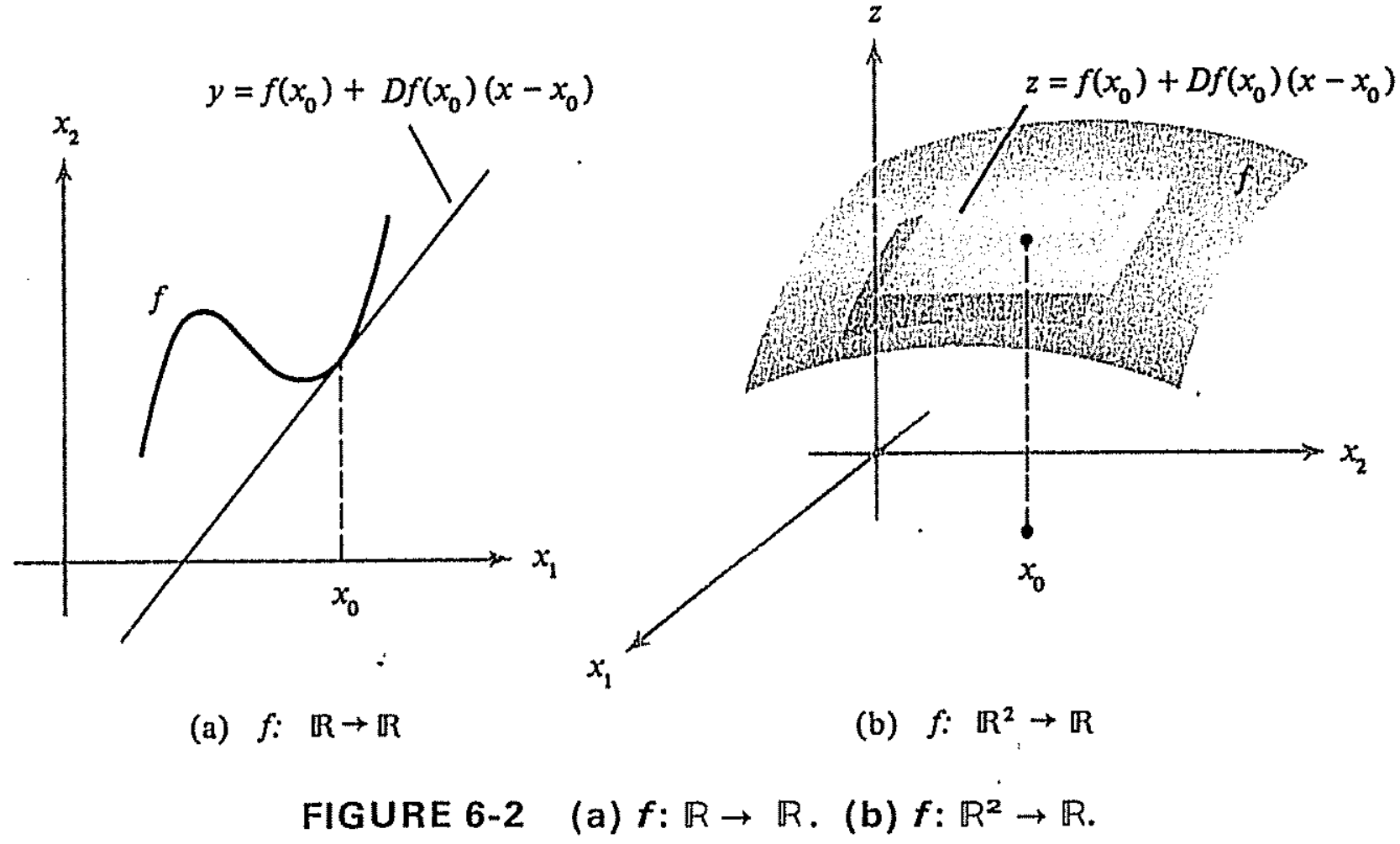}
    \caption{}
    \label{fig:der2}
\end{figure}

So far we introduced two definitions for differentiability and the derivative. Let us discuss the directional derivative:

\begin{definition}
Let $f$ be a defined in a neighborhood of $x_0 \in \mathbb{R}^n$ and let $e \in \mathbb{R}^n$ be a unit vector. Then
$$\frac{d}{dt}f(x_0+te) \Big |_{t=0} = \lim_{t \to 0} {\frac{f(x_0 + te) - f(x_0)}{t}}$$
is called the directional derivative of $f$ at $x_0$ in the direction $e$.
\end{definition}

From this definition, the directional derivative is just the rate of change of $f$ in the direction $e$; see figure \ref{fig:der3}.

\begin{figure}[H]
    \centering
    \includegraphics[width=.5\linewidth]{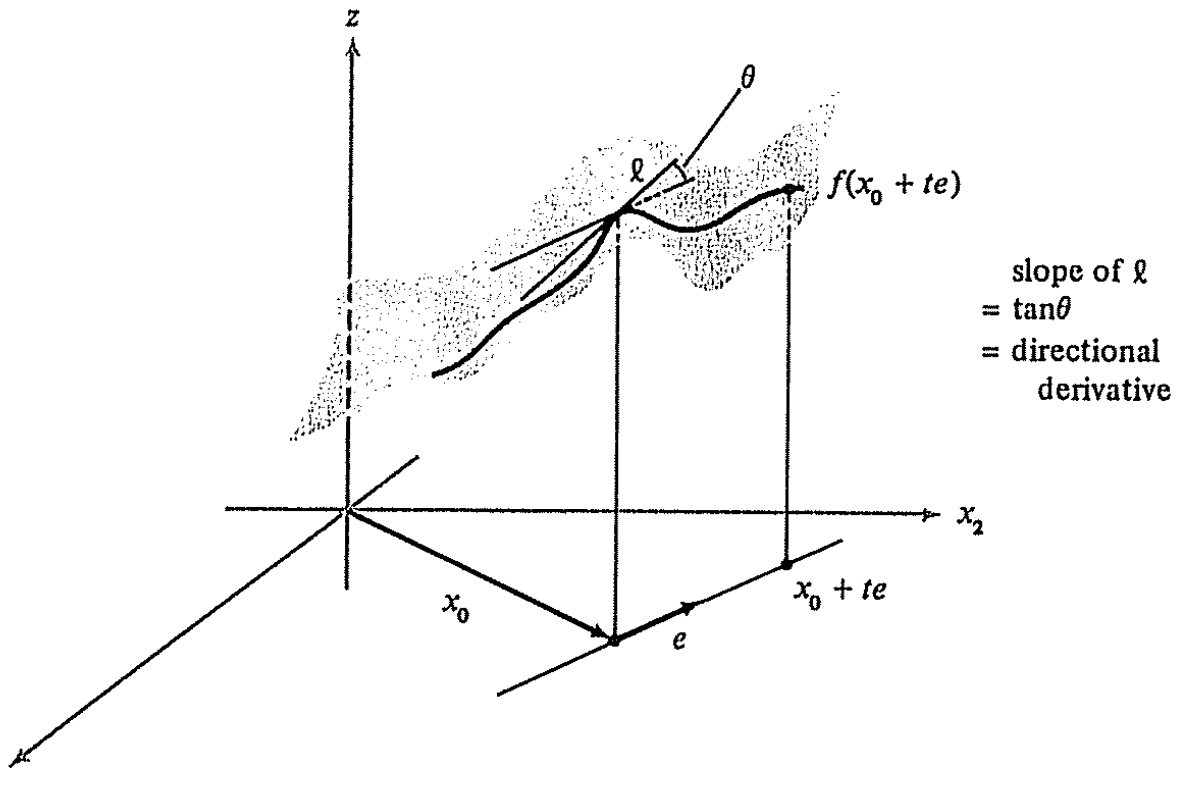}
    \caption{Slope of $\ell = tan\theta =$  directional derivative}
    \label{fig:der3}
\end{figure}

We claim that the directional derivative in the definition of $e$ equals $Df(x_0) \cdot e$. To see this just look at the definition of $Df(x_0)$ \eqref{eq:def3} with $x = x_0 + te$; we get

$$ \norm{\frac{f(x_0 + te) - f(x_0)}{t} - Df(x_) \cdot e} \leq \epsilon \norm{e}$$ for any $\epsilon > 0$

if $\norm{t}$ is sufficiently small. This proves that if f is differentiable at $x_0$, then the directional derivatives also exist and are give by 

$$\lim_{t \to 0} {\frac{f(x_0) + te - f(x_0)}{t} = Df(x_0) \cdot e}.$$

In particular, observe that $\frac{\partial f}{\partial x_i}$ is the derivative of $f$ in the direction of the ith coordinate axis (with $e = e_i = (0,0,...,0,1,0,...,0)$.)
\\

In order to differentiate a function $f$ of several variables, one can either use the definition of derivative (for which an example will be demonstrated in \ref{sec3-2}) or use that of partial derivatives, meaning that one writes $f$ in component form $f(x_1,...,x_n) = (f_1(x_1,...,x_n),...,f_m(x_1,...,x_n))$ and compute the partial derivatives, $\frac{\partial f_j}{\partial x_i}$ for $j=1,...,m$ and $i=1,...,n$, where the symbol $\frac{\partial f_j}{\partial x_i}$ means that we compute the usual derivative of $f_j$ with respect to $x_i$ while keeping the other variables $x_1,...,x_{i-1},x_{i+1},...,x_n$ fixed. Explicitly, \ref{partdef}

\begin{definition} \label{partdef}
 $\frac{\partial f_j}{\partial x_i}$ is given by the following limit, when the latter exists:
$$\frac{\partial f_j}{\partial x_i}(x_1,...,x_n) = \lim_{h \to 0} {\{\frac{f_j(x_1,...,x_i+h,...,x_n) - f_j(x_1,...,x_n)}{h}\}}$$
\end{definition}

\begin{theorem}
Suppose $A \subset \mathbb{R}^n$ is an open set and $f:A \to \mathbb{R}^m$ is differentiable. Then the partial derivatives $\frac{\partial f_j}{\partial x_i}$ exist and the matrix of the linear map $Df(x)$ with respect to the standard bases in $\mathbb{R}^n$ and $\mathbb{R}^m$ is given by 

$$
\begin{pmatrix}
\frac{\partial f_1}{\partial x_1} & \frac{\partial f_1}{\partial x_2} & \cdots & \frac{\partial f_1}{\partial x_n}\\
. & . & & .\\
. & . & & .\\
. & . & & .\\
\frac{\partial f_m}{\partial x_1} & \frac{\partial f_m}{\partial x_2} & \cdots & \frac{\partial f_m}{\partial x_n}\\
\frac{\partial f_j}{\partial x_i} & \frac{\partial f_j}{\partial x_i} & \cdots & \frac{\partial f_j}{\partial x_i} 
\end{pmatrix}
$$

where each partial derivative is evaluated at $x = (x_1,...,x_n).$ This matrix is called the Jacobian matrix of $f$.

\end{theorem}

\begin{proof}
See [\cite{nnm2}, Chapter 6, Section 6].
\end{proof}

One should take special note when $m=1$, in which case we have a real-valued function of $n$ variables. Then $Df$ has the matrix

$$(\frac{\partial f}{\partial x_1} ... \frac{\partial f}{\partial x_n})$$

and the derivative applied to a vector $e = (a_1,...,a_m)$ is 

$$Df(x) \cdot e = \sum_{i=1}^{n} \frac{\partial f}{\partial x_i} a_i .$$

It should be emphasized that $Df$ is a linear mapping at each $x \in A$ and the definition of $Df(x)$ is independent of the basis used. If we change the basis from standard basis to another one, the matrix elements will definitely change. However, the $Df$ transformation is denoted by Jacobian Matrix which is merely a representation and simplifies the computation of derivative. 

\subsection{Derivatives} \label{sec3-2}
\subsubsection{First Derivatives}
So far we explained the definition of differentiability and derivative in multivariate spaces. Now, we will set a uniform notation that serves our primary purpose of formalizing Neural Networks' structure. This notation is in line with \cite{dnnmf} and it is as follows:\\

Consider the function $f: E_1 \to E_2$, where $E_1$ and $E_2$ are inner product spaces. The first derivative map of $f$, denoted $Df$, is a map from $E_1$ to $L(E_1; E_2)$ which operates in the following manner for any $v \in E_1$:

\begin{equation} \label{eq:def2}
Df(x) \cdot v = \frac{d}{dt}f(x+tv) \Big |_{t=0}.    
\end{equation}

\subsubsection{Example 1}
Let $f: \mathbb{R}^2 \to \mathbb{R}$, $f(x,y) = xy$. Compute $Df(x)$ and $\frac{df}{dx}$. 

\begin{equation}
\begin{split}
Df(x,y) \cdot (v_1,v_2) &  = \frac{d}{dt}(f(x+tv_1, y+tv_2)) \Big |_{t=0} \\
 & = \frac{d}{dt}((x + tv_1)\cdot(y+tv_2)) \Big |_{t=0} \\
 & = (xv_2 + yv_1 + 2t{v_1}{v_2}) \Big |_{t=0} \\
 & = xv_2 + yv_1
\end{split}
\end{equation}

\subsubsection{Example 2}
Let $f: \mathbb{R}^2 \to \mathbb{R}^3$, $f(x,y) = (x^2,{x^3}y,{x^4}{y^2})$. Compute $Df(x)$.

$$
\begin{pmatrix} 
 \frac{\partial f_1}{\partial x} & \frac{\partial f_1}{\partial y} \\
 \frac{\partial f_2}{\partial x} & \frac{\partial f_2}{\partial y} \\
 \frac{\partial f_3}{\partial x} & \frac{\partial f_3}{\partial y}
\end{pmatrix}
 =
  \begin{pmatrix}
   2x & 0 \\
   3{x^2}y & x^3 \\
   4{x^3}{y^2} & {2{x^4}y}
   \end{pmatrix}
$$

where $f_1(x,y) = x^2, f_2(x,y) = {x^3}y, f_3(x,y) = {x^4}{y^2}$.

For each $x \in E_1$, the adjoint of the derivative $Df(x) \in L(E_1;E_2)$ is well-defined, and we will denote it $D\mbox{*}f(x)$ instead of $Df(x)\mbox{*}$ for the sake of convenience. Then, $D\mbox{*}f(x): E_1 \to L(E_2;E_1)$ denotes the map that takes each point $x \in E_1$ to $D\mbox{*}f(x) \in L(E_2;E_1)$.
\\
\\
Now, let us consider two maps $f_1:E_1 \to E_2$ and $f_2:E_2 \to E_3$ that are $C^1$, where $E_3$, where $E_3$ is another inner product space. The derivative of their composition, $D(f_2 \circ f_1)(x)$, is a linear map from $E_1$ to $E_3$ for any $x \in E_1$, and is calculated using the well-known chain rule, i.e.
\begin{equation}
\label{eq:cr}
D(f_2 \circ f_1)(x) = Df_2(f_1(x)) \cdot Df_1(x).
\end{equation}

\subsubsection{Second Derivatives}
We can safely assume that every map here is $C^2$. The second derivative map of $f$, denoted $D^2f$, is a map from $E_1$ to $L(E_1,E_1;E_2)$, which operates as $x \mapsto D^2f(x)$ for any $x \in E_1$. The bilinear map $D^2f(x) \in L(E_1,E_1;E_2)$ operates as 

\begin{equation}
D^2f(x) \cdot (v_1,v_2) = D(D(f(x) \cdot v_2) \cdot v_1 = \frac{d}{d_t}(Df(x+tv_1) \cdot v_2) \Big |_{t=0}
\end{equation}

for any $v_1,v_2 \in E_1$. The map $D^2f(x)$ is symmetric, i.e. $D^2f(x) \cdot (v_1,v_2) = D^2f(x) \cdot (v_2,v_1)$ for all $v_1,v_2 \in E_1$.


Two useful identities exist for vector-valued second derivatives—the higher order chain rule and the result of mixing $D$ with $D\mbox{*}$ —which we will describe in the next two lemmas.

\newtheorem{lemma}[theorem]{Lemma}

\begin{lemma}
For any $x,v_1,v_2 \in E_1$,

$$D^2(f_2 \circ f_1)(x) \cdot (v_1,v_2) = D^2f_2(f_1(x)) \cdot (Df_1(x) \cdot v_1,Df_1(x) \cdot v_2) + Df_2(f_1(x)) \cdot D^2f_1(x) \cdot (v_1,v_2),$$

where $f_1:E_1 \to E_2$ is $C^2$ and $f_2:E_2 \to E_3$ is $C^2$ for vectors spaces $E_1, E_2$ and $E_3$
\end{lemma}

\begin{proof}
We can prove this directly from the definition of the derivative.


\begin{subequations}
\begin{align}
D^2(f_2 \circ f_1)(x) \cdot (v_1,v_2) & = D(D(f_2 \circ f_1)(x) \cdot v_2) \cdot v_1 \label{subeq:11} \\
& = D(D(f_2(f_1(x)) \cdot Df_1(x) \cdot v_2) \cdot v_1 \label{subeq:12} \\
& = \frac{d}{d_t} (Df_2(f_1(x+tv_1)) \cdot Df_1(x+tv_1) \cdot v_2) \Big |_{t=0} \label{subeq:13} \\
& = \frac{d}{d_t} (Df_2(f_1(x+tv_1)) \cdot Df_1(x+tv_1) \cdot v_2) \Big |_{t=0} \label{subeq:14} \\
& \hspace{3.5mm} + D(f_2(f_1(x)) \cdot \frac{d}{d_t}(Df_1(x+tv_1) \cdot v_2) \Big |_{t=0} \\
& = D^2f_2(f_1(x)) \cdot \Big(\frac{d}{d_t}f_1(x+tv_1) \Big |_{t=0}, Df_1(x) \cdot v_2\Big) \\
& \hspace{3.5mm} + Df_2(f_1(x)) \cdot D^2f_1(x) \cdot (v_1,v_2) \label{subeq:15}\\
& = D^2f_2(f_1(x)) \cdot (Df_1(x) \cdot v_1, Df_1(x) \cdot v_2) \label{subeq:16}\\
& \hspace{3.5mm} + Df_2(f_1(x)) \cdot D^2f_1(x) \cdot (v_1,v_2), \label{subeq:17}
\end{align}
\end{subequations}

where \eqref{subeq:11} is from \eqref{eq:cr}, \eqref{subeq:12} is from the definition of derivative \eqref{eq:def2}, \eqref{subeq:13} and \eqref{subeq:14} is from the standard product rule, and \eqref{subeq:16} and \eqref{subeq:17} is from the standard chain rule along with the definition of second derivative.

\end{proof}

\begin{lemma} \label{lemma 3-2 2}
Consider three inner product spaces $E_1, E_2$ and $E_3$, and two functions $f:E_1 \to E_2$ and $g:E_2 \to E_3$. Then, for any $x,v \in E_1$ and $w \in E_3$, 

$$D(D\mbox{*}g(f(x)) \cdot w) \cdot v = \Big( (Df(x) \cdot v)\lrcorner D^2g(f(x)) \Big) \mbox{*} \cdot w.$$

\end{lemma}

\begin{proof}
Pair the derivative of the map $D\mbox{*}g(f(x)) \cdot w$ with any $y \in E_2$ in the inner product. Since product rule applies to norms as well, we have:

\begin{subequations}
\begin{align}
\langle y, D(D\mbox{*}g(f(x)) \cdot w) \cdot v \rangle & = D(\langle y, D\mbox{*}g(f(x)) \cdot w \rangle) \cdot v\\
& = D(\langle D(g(f(x)) \cdot y, w \rangle) \cdot v\\
& = \langle D^2g(f(x)) \cdot (Df(x) \cdot v,y),w \rangle\\
& = \langle\Big( (Df(x) \cdot v) \lrcorner D^2g(f(x)) \Big) \cdot y,w\rangle\\
& = \langle y, \Big( (Df(x)\cdot v) \lrcorner D^2g(f(x)) \Big) \mbox{*} \cdot w\rangle .
\end{align}
\end{subequations}

This holds for any $y \in E_2$.
\end{proof}

\subsubsection{Parameter-Dependent Maps}
Now, the derivative notation developed in the previous subsection will be developed to parameter-dependent maps: maps containing both a state variable and a parameter. We will heavily rely on parameter-dependent maps because we can regard the input of each layer of a feed-forward neural network as the current state of the network, which will be evolved according to the parameters at the current layer. To formalize this notion, suppose $f$ is a parameter-dependent map from $E_1 \times H_1$ to $E_2$, i.e. $f(x;\theta) \in E_2$ for any $x \in E_1$ and $\theta \in H_1$, where $H_1$ is also an inner product space. In this context, we will refer to $x \in E_1$ as the state for $f$, whereas $\theta \in H_1$ is the parameter. 

\subsubsection{First Derivatives}
We will use the notation presented in \ref{sec3-2} to denote the derivative of $f$ with respect to the state variable:\\
for all $v \in E_1$,

\begin{equation} \label{eq:def}
Df(x;\theta) \cdot v = \frac{d}{dt}f(x+tv;\theta) \Big |_{t=0}.    
\end{equation}

Also, $D^2f(x;\theta) \cdot (v_1,v_2) = D(Df(x;\theta) \cdot v_2) \cdot v_1$ as before. However, we will introduce new notation to denote the derivative of $f$ with respect to the parameters as follows:

\begin{equation} \label{eqpd:def}
\nabla f(x;\theta) \cdot u = \frac{d}{dt}f(x;\theta + tu) \Big |_{t=0}    
\end{equation} 

for any $u \in H_1$. Note that $\nabla f(x;\theta) \in L(H_1;E_2)$. We will require a chain rule for the composition of functions involving parameter-dependent maps, especially when not all of the functions involving parameter-dependent maps, especially when not all of the functions in the composition depend on the parameter, and this all appears in \ref{lemma 3-3 1}.

\begin{lemma} \label{lemma 3-3 1}
Suppose that $E_1, E_2, E_3,$ and $H_1$ are inner product spaces, and $g: E_2 \to E_3$ and $f: E_1 \times H_1 \to E_2$ are both $C^1$ functions. Then, the derivative of their composition with respect to the second argument of $f$, i.e. $\nabla (g \circ f)(x;\theta) \in L(H_1; E_3),$ is given by 

\begin{equation} \label{eq:3-3}
\nabla (g \circ f)(x;\theta) = Dg(f(x;\theta)) \cdot \nabla f(x;\theta),  
\end{equation}

for any $x \in E_1$ and $\theta \in H_1$
\end{lemma} 

\begin{proof}
This is just an extension of \eqref{eq:cr}.
\end{proof}

\subsubsection{Higher-Order Derivatives}
We can define the mixed partial derivative maps, $\nabla Df(x;\theta) \in L(H_1; E_2)$ as

$$\nabla Df(x;\theta) \cdot (u,e) = \frac{d}{d_t} (Df(x; \theta +tu) \cdot e) \Big |_{t=0}$$

and

$$D \nabla f(x;\theta) \cdot (e,u) = \frac{d}{d_t} (\nabla f(x+te; \theta ) \cdot u) \Big |_{t=0}$$

for any $e \in E_1, u \in H_1$. Note that if $f$ is $C^2,$ then

$$D\nabla f(x;\theta) \cdot (e,u) = \nabla Df(x;\theta) \cdot (u,e).$$

A useful identity similar to Lemma \ref{lemma 3-2 2} exists when mixing $\nabla\mbox{*}$ and $D$. 

\begin{lemma}
Consider three inner product spaces $E_1, E_2$ and $H_1,$ and a parameter-dependent map $g: E_1 \times H_1 \to E_2$. Then, for any $x,v \in E_1, w \in E_2,$ and $\theta \in H_1$,

$$D(\nabla\mbox{*}g(x;\theta) \cdot w) \cdot v = (\nabla Dg(x;\theta) \llcorner v)\mbox{*} \cdot w = (v \lrcorner D\nabla g(x;\theta))\mbox{*} \cdot w.$$
\end{lemma}

\begin{proof}
Prove similarly to Lemma \ref{lemma 3-2 2} by choosing $y \in H_1$ as a test vector.
\end{proof}

\subsection{Elementwise Functions} \label{subsec:elemwise funcs}

Layered neural networks conventionally contain a nonlinear activation function operating on individual coordinates---also known as \emph{elementwise nonlinearity}--- placed at the end of each layer. Without these, neural networks would be nothing more than over-parametrized linear models; it is therefore important to understand the properties of elementwise functions. To this end, consider an inner product space $E$ of dimension $n$, and let $\{e_k\}_{k=1}^n$ be an orthonormal basis of $E$. We define an elementwise function as a map $\Psi: E \to E$ of the form 

\begin{equation} \label{eq:psi}
    \Psi(v) = \sum_{k=1}^{n} \psi(\langle v,e_k \rangle) e_k,
\end{equation}

where $\psi : \mathbb{R} \to \mathbb{R}$---which we will refer to as the elementwise operation associated with $\Psi$---defines the operation of the elementwise function over the coordinates $\{\langle v,e_k \rangle\}_k$ of the vector $v \in E$ with respect to the chosen basis. If we use the convention that $\langle v,e_k \rangle \equiv v_k \in \mathbb{R}$, we can rewrite \eqref{eq:psi} as

$$\Psi(v) = \sum_{k=1}^{n} \psi(v_k) e_k,$$

but we will tend to avoid this as it becomes confusing when there are multiple subscripts. The operator $\Psi$ is basis-dependent, but $\{e_k\}_{k=1}^n$ can by any orthonormal basis of $E$. \\

We define the associated \emph{elementwise first derivative,} $\Psi' : E \to E$ as

\begin{equation} \label{eq:psi1stdev}
    \Psi'(v) = \sum_{k=1}^{n} \psi'(\langle v,e_k \rangle) e_k.
\end{equation}

Similarly, we define \emph{elementwise first derivative} $\Psi'' : E \to E$ as

\begin{equation} \label{eq:psi2nddev}
    \Psi''(v) = \sum_{k=1}^{n} \psi''(\langle v,e_k \rangle) e_k.
\end{equation}

We can also rewrite Eqs. \eqref{eq:psi1stdev} and \eqref{eq:psi2nddev} using $\langle v,e_k \rangle \equiv v_k$ as 

$$\Psi'(v) = \sum_{k=1}^{n} \psi'(v_k) e_k$$

and

$$\Psi''(v) = \sum_{k=1}^{n} \psi''(v_k) e_k.$$

\subsubsection{Hadamard Product}
To assist in the calculation of derivatives of elementwise functions, we will define a symmetric bilinear operator $\odot \in L(E,E;E)$ over the orthogonal basis $\{e_k\}_{k=1}^n$ as 

\begin{equation} \label{hadprod}
    e_k \odot e_k' \equiv \delta_{k,k'} e_k,
\end{equation}

where $\delta_{k,k'}$ is the Kronecker delta. This is the standard Hadamard product---also known as elementwise multiplication--- when $E = \mathbb{R}^n$ and $\{e_k\}_{k=1}^n$ is the standard basis of $\mathbb{R}^n$, which we can see by calculating $v \odot v'$ for any $v,v' \in \mathbb{R}^n :$

\begin{equation}
\begin{split}
    v \odot v' & = \Bigg( \sum_{k=1}^{n} v_k e_k \Bigg) \odot \Bigg( \sum_{k=1}^{n} v'_{k'} e_{k'} \Bigg) \\
    & = \sum_{k,k'=1}^{n} v_k v'_{k'} e_k \odot e_{k'} \\
    & = \sum_{k,k'=1}^{n} v_k v'_{k'} \delta_{k,k'} e_{k} \\
    & = \sum_{k=1}^{n} v_k v'_{k} e_{k}, 
\end{split}
\end{equation}

where we have used the convention that $\langle v,e_k \rangle \equiv v_k.$ However, when $E \neq \mathbb{R}^n$ or $\{e_k\}_{k=1}^n$ is not the standard basis, we can regard $\odot$ as a generalization of the Hadamard product. For all $y,v,v' \in E$, the Hadamard product satisfies the following properties:

\begin{equation}
\begin{split}
    v \odot v' & = v' \odot v \\
    (v \odot v') \odot y & = v \odot (v' \odot y) ,\\
    \langle y, v \odot v' \rangle & = \langle v \odot y, v' \rangle = \langle y \odot v', v \rangle.
\end{split}
\end{equation}

\subsubsection{Derivative of Elementwise Functions}
We can now compute the derivative of elementwise functions using the Hadamard product as described below.

\begin{theorem} \label{thm:elemfunc1stder}
Let $\Psi : E \to E$ be an elementwise function over an inner product
space $E$ as defined in \eqref{eq:psi}. Then, for any $v,z \in E,$

$$D \Psi(z) \cdot v = \Psi' (z) \odot v.$$

Furthermore, $D\Psi(z)$ is self-adjoint for all $z \in E,$ i.e. $D\mbox{*} \Psi(z) = D\Psi(z)$ for all $z \in E$.
\end{theorem}

\begin{proof}
Let $\psi$ be the elementwise operation associated with $\Psi$. Then, 

\begin{equation}
\begin{split}
    D\Psi(z) \cdot v & = \frac{d}{dt}\Psi(z+tv) \Big |_{t=0} \\
    & = \frac{d}{dt} \sum_{k=1}^{n} \psi(\langle z+tv, e_k\rangle) e_k \Big |_{t=0}\\
    & = \sum_{k=1}^{n} \psi'(\langle z, e_k\rangle) \langle v, e_k \rangle e_k \\
    & = \Psi'(z) \odot v,
\end{split}
\end{equation}

where the third equality follows from the chain rule and linearity of the derivative. Furthermore, for any $y \in E,$

$$\langle y, D\Psi(z) \cdot v \rangle = \langle y, \Psi'(z) \odot v \rangle = \langle \Psi'(z) \odot y, v \rangle = \langle D\Psi(z) \cdot y, v \rangle .$$

Since $\langle y, D\Psi(z) \cdot v \rangle = \langle D \Psi(z) \cdot y, v\rangle$ for any $v,y,z \in E, D\Psi(z)$ is self-adjoint.
\end{proof}

\begin{theorem} \label{thm:elemfunc2ndder}
Let $\Psi : E \to E$ be an elementwise function over an inner product
space $E$ as defined in \eqref{eq:psi}. Then, for any $v_1,v_2,z \in E,$

$$D^2 \Psi(z) \cdot (v_1,v_2) = \Psi'' (z) \odot v_1 \odot v_2 .$$

Furthermore, $(v_1 \lrcorner D^2\Psi(z))$ and $(D^2\Psi(z) \llcorner v_2)$ are both self-adjoint linear maps for any $v_1,v_2,z \in E.$
\end{theorem}

\begin{proof}
We can prove this directly:

\begin{equation}
\begin{split}
    D^2\Psi(z) \cdot (v_1,v_2) & = D(D\Psi(z) \cdot v_2) \cdot v_1 \\
    & = D(\Psi'(z) \cdot v_2) \cdot v_1\\
    & = (\Psi'' (z) \odot v_1) \odot v_2,
\end{split}
\end{equation}

where the third equality follows since $\Psi'(z) \odot v_2$ is an elementwise function in $z$. \\
Also, for any $y \in E,$

\begin{equation}
\begin{split}
    \langle y, \Big( v_1 \lrcorner D^2 \Psi(z) \Big) \cdot v_2 \rangle & = \langle y,D^2 \Psi(z) \cdot (v_1, v_2) \rangle \\
    & = \langle \Psi''(z) \odot v_1 \odot y, v_2 \rangle \\
    & = \langle \Big( v_1 \lrcorner D^2\Psi(z) \Big) \cdot y, v_2 \rangle.
\end{split}
\end{equation}

This implies that the map $\big( v_1 \lrcorner D^2\Psi(z) \big)$ is self-adjoint map for any $v_1,z \in E$. From the symmetry of the second derivative $D^2\Psi(z)$, the map $\big(D^2\Psi(z) \llcorner v_1 \big)$ is also self-adjoint for any $v_1,z \in E.$
\end{proof}

\subsection{Conclusion}

In this section, the mathematical tools for handling vector-valued functions that will arise when describing generic neural networks were presented. In particular, the notation and theory surrounding linear maps, derivatives, parameter-dependent maps, and elementwise functions were introduced. Familiarity with the material presented in this chapter is paramount for understanding the rest of this chapter.

\clearpage
\section{Generic Representation of Neural Networks}

According to Universal Approximation Theorem, a feedforward network with a single layer is sufficient to represent any function, but the layer may be infeasibly large and may fail to learn and generalize correctly (See [\cite{dl}, Chater 6, Section 4] for more details). In many circumstances, using deeper models can reduce the number of units required to represent the desired function and can reduce the amount of generalization error. Therefore, it is important to develop a solid and concise theory for repeated function composition as it pertains to neural networks; and we will use the mathematical tools described in the previous section \ref{sec:math tools} to make it happen. Also, the derivatives of these functions with respect to the parameters at each layer are computed since neural networks often learn their parameters via some form of gradient descent. The derivative maps that are computed will remain in the same vector space as the parameters, which will allow us to perform gradient descent naturally over these vector spaces.

For commencing the path to formalize NNs, a generic neural network is formulated as the composition of parameter-dependent functions. We will then introduce standard loss functions based on this composition for the task of regression. Since forecasting which was our primary problem is a regression task, we will dismiss the classification case. 
\clearpage

\subsection{Neural Network Formulation} \label{nnf}
A deep neural network with $L$ layers can be represented as the composition of $L$ functions $f_i:E_i \times H_i \to E_{i+1}$, where $E_i, H_i,$ and $E_{i+1}$ are inner product spaces for all $i \in [L]$. We will refer to the variables $x_i \in E_i$ as \emph{state variables}, and the variables $\theta_i \in H_i$ as \emph{parameters}. Throughout this thesis, the dependence of the layerwise function $f_i$ on the parameter $\theta_i$ is often suppressed for ease of composition, i.e. $f_i$ is understood as a function from $E_i$ to $E_{i+1}$ depending on $\theta_i \in H_i.$ We can then write down the output of a neural network for a generic input $x \in E_1$ using this suppression convention as a function $F: E_1 \times (H_1 \times \cdots \times H_L) \to E_{L+1}$ according to 

$$F(x;\theta) = (f_L \circ \cdots \circ f_1)(x),$$

where each $f_i$ is dependent on the parameter $\theta_i \in H_i$, and $\theta$ represents the parameter set ${\theta_1,...,\theta_L}.$ For now, we will assume that each parameter $\theta_i$ is independent of the other parameters $\{{\theta_j}\}_{j \ne i},$ but we will see how to modify this assumption when working with recurrent neural networks in the last section.
\\

Now, some maps will be introduced so as to assist in the calculation of derivatives. First, the \emph{head} map at layer $i, \alpha_i : E_1 \to E_{i+1},$ is given by 

\begin{equation} \label{eq:alpha}
    \alpha_i = f_i \circ \cdots \circ f_1
\end{equation}

for each $i \in [L] \equiv \{1,...,L\}.$ Note that $\alpha_i$ implicitly depends on the parameters $\{\theta_1,...,\theta_i\}.$ For convenience, set $\alpha_0 = id:$ the identity map on $E_1$. Similarly, \emph{tail} map at layer $i, \omega_i:E_i \to E_{L+1},$ can be defined as 

\begin{equation} \label{eq:omega}
\omega_i = f_L \circ \cdots \circ f_i
\end{equation}

for each $i \in [L].$ The map $\omega_i$ implicitly depends on $\{\theta_i,...,\theta_L\}.$ Again for convenience, set $\omega_{L+1}$ to be the identity map on $E_{L+1}.$ It can be easily show that the following holds for all $i \in [L]:$
\begin{equation}
\label{eq:4-1}
F = \omega _{i+1} \circ \alpha _i, \;\;\;\;\;    
 \omega _i = \omega _{i+1} \circ f_i, \;\;\;\;\;   
 \alpha _i = f_i \circ \alpha _{i-1} 
\end{equation}

The equations in \eqref{eq:4-1} imply that the output $F$ can be decomposed into 

$$F = \omega_{i+1} \circ f_i \circ \alpha_{i-1}$$

for all $i \in [L]$, where both $\omega _{i+1}$ and $\alpha _{i-1}$ have no dependence on the parameter $\theta _i$.

\subsection{Loss Functions and Gradient Descent}

The goal of neural network is to optimize some loss function $J$ with respect to the parameters $\theta$ over a set of $n$ network inputs $\mathcal{D} = \{(x_{(1)},y_{(1)}),...,(x_{(n)}, y_{(n)})\},$ where $x_{(j)} \in E_1$ is the $j$th input data point with associated response or target $y_{(j)} \in E_{L+1}.$ Most optimization methods are gradient-based, meaning that we must calculate the gradient of $J$ with respect to the parameters at each layer $i \in [L].$

At first, the squared loss function will be introduced along with taking derivatives of it for a single data point $(x,y) \equiv (x_{(j)},y_{(j)})$ for some $j \in [n],$ and then concisely present error backpropagation. Finally, the algorithm for performing gradient descent steps for regression will be presented.\\
%
%

Our starting point is a result to compute $\nabla_{\theta_i}^{\mbox{*}} F(x;\theta).$

\begin{lemma}
For any $x \in E_1$ and $i \in [L],$

\begin{equation} \label{eq:gradf}
\nabla_{\theta _i}^{\mbox{*}} F(x;\theta) = \nabla_{\theta_i}^{\mbox{*}} f_i(x_i) \cdot D\mbox{*}\omega_{i+1}(x_{i+1})
\end{equation}

\end{lemma}

\begin{proof}
Apply the chain rule from \eqref{eq:3-3} to $F = \omega_{i+1} \circ f_i \circ \alpha _{i-1}$ according to

\begin{equation} \label{eq:5}
\begin{split}
\nabla_{\theta _i}F(x;\theta) &  = D{\omega_{i+1}}(f_i(\alpha_{i-1}(x))) \cdot \nabla_{\theta _i}f_i(\alpha _{i-1}(x)) \\
 & =  D{\omega_{i+1}}(x_{i+1}) \cdot \nabla_{\theta_i}f_i(x_i), 
\end{split}
\end{equation}

since neither $\omega_{i+1}$ nor $\alpha_{i-1}$ depend on $\theta_i$. Then, by taking the adjoint and applying the reversing property we can obtain \eqref{eq:gradf}.

\end{proof}

\subsubsection{Regression}
The target variable $y \in E_{L+1}$ can be any generic vector of real numbers. Thus, for a single data point, the most common loss function to consider is the squared loss, given by

\begin{equation} \label{eq:sqloss} 
J_R(x,y; \theta) = \frac{1}{2} \norm{y - F(x; \theta)}^2 = \frac{1}{2} \langle y - F(x; \theta), y - F(x; \theta)\rangle.
\end{equation}

The network perdiction ${\widehat{y}}_R \in E_{L+1}$ is given by the network output $F(x; \theta).$ We can calculate the gradient of $J_R$ with respect to the parameter $\theta_i$ according to the following theorem.

\begin{theorem} \label{thm:gradj}
For any $x \in E_1, y \in E_{L+1},$ and $i \in [L]$,

\begin{equation} \label{eq:derj}
\nabla_{\theta_i}J_R(x, y; \theta) = \nabla_{\theta_i}^{\mbox{*}} f_i(x_i) \cdot D\mbox{*}\omega_{i+1}(x_{i+1}) \cdot ({\widehat{y}}_R - y),    
\end{equation}

where $x_i = \alpha_{i-1}(x)$ and ${\widehat{y}}_R = F(x; \theta)$

\end{theorem} \label{thm:derj}

\begin{proof}
By the product rule, for any $U_i \in H_i,$

\begin{equation}
\begin{split}
\nabla_{\theta_i} J_R(x, y; \theta) \cdot U_i & = \nabla_{\theta_i} \tfrac{1}{2} \langle F(x; \theta) - y, F(x; \theta) - y \rangle \cdot U \\
& = \tfrac{1}{2} \langle \nabla_{\theta_i} F(x; \theta), F(x; \theta) - y \rangle \cdot U + \tfrac{1}{2} \langle F(x; \theta) - y, \nabla_{\theta_i} F(x; \theta) \rangle \cdot U\\
& = \langle F(x; \theta) - y, \nabla_{\theta_i} F(x; \theta) \cdot U_i \rangle\\
& = \langle \nabla_{\theta_i}^{\mbox{*}}F(x; \theta) \cdot (F(x; \theta) - y), U_i \rangle,
\end{split}
\end{equation}

This implies that the derivative map above is a \emph{linear functional}, i.e. $\nabla_{\theta_i} J_R(x, y; \theta) \in L(H_i; \mathbb{R}).$ Then, by the canonical isomorphism described in [\cite{nnm3}, Chapter 5, Section 3], $\nabla_{\theta_i} J_R(x, y; \theta)$ can be represented as an element of $H_i$ as 

$$\nabla_{\theta_i} J_R(x, y; \theta) = \nabla_{\theta_i}^{\mbox{*}} F(x; \theta) \cdot (F(x; \theta) - y).$$

Since $F(x; \theta) = \widehat{y}_R$ and $\nabla_{\theta_i}^{\mbox{*}} F(x; \theta) = \nabla_{\theta_i}^{\mbox{*}} f_i(x_i) \cdot D\mbox{*}\omega_{i+1}(x_{i+1})$ by \eqref{eq:gradf}, we have thus proven \eqref{eq:derj}

\end{proof}

\subsection{Backpropagation} \label{subsec:bptt-gen}

The derivative of loss function with respect to a generic parameter $\theta _i$ \eqref{eq:derj} involves applying $D\mbox{*}\omega_{i+1}(x_{i+1})$ to an error vector, that is ${\widehat{y}}_R - y.$ This operation is commonly referred to as \emph{backpropagation}, and the procedure of calculating it recursively is demonstrated in the the next theorem. 

\begin{theorem}[\textbf{Backpropagation}] \label{thm:bckpp}
For all $x_i \in E_i,$ with $\omega_i$,

\begin{equation} \label{eq:bckppg}
D\mbox{*}\omega_i(x_i) = D\mbox{*}f_i(x_i) \cdot D\mbox{*}\omega_{i+1}(x_{i+1}),    
\end{equation}
where $x_{i+1} = f_i(x_i),$ for all $i \in [L].$

\end{theorem}

\begin{proof}
Apply the chain rule \eqref{eq:cr} to $\omega_i(x_i) = (\omega_{i+1} \circ f_i)(x_i),$ and take the adjoint to obtain \eqref{eq:bckppg}. This holds for any $i \in [L]$ since $\omega_{L+1} = id.$
\end{proof}

The reason that \ref{alg:backppg} is referred to as backpropagation will be elucidated in Algorithm \ref{alg:backppg}, since $D\mbox{*}\omega_i(x_i)$ will be applied to an error vector $e_L \in E_{L+1}$ and then sent backwards at each layer $i.$

\subsection{Gradient Descent Step Algorithm}

 The method for computing one step of gradient descent is presented for a generic layered neural network in Algorithm 3.2.1, which clarifies how the results of this subsection can be combined. The inputs are the network input point $(x,y) \in E_1 \times E_{L+1},$ the parameter set $\theta = \{\theta_1,...,\theta_L\}$ and the learning rate $\eta \in \mathbb{R}_+$. It updates the set of network parameters $\theta$ via one step of gradient descent. 

\subsubsection{Describing Algorithm}
At first the network prediction is generated using forward propagation from lines 2-5 and store the state at each layer. Then, these states are used in the backpropagation step, which begins at line 6. At the top layer $(i=L),$ the error vector $e_L$ is  initialized to ${\widehat{y}}_R - y$, since $D\mbox{*}\omega_{L+1}(x_{L+1}) = id$ and

$$\nabla_{\theta_L} J_R(x, y; \theta) = \nabla_{\theta_L}^{\mbox{*}} f_L(x_L) \cdot D\mbox{*} \omega_{L+1}(x_{L+1}) \cdot e_L = \nabla_{\theta_L}^{\mbox{*}} f_L(x_L) \cdot e_L.$$

When $i \neq L,$ the error vector in line 11 is updated through multiplication by $D\mbox{*}f_{i+1}(x_{i+1})$ in accordance with \eqref{eq:bckppg} e.g. $ e_{i} = D\mbox{*}f_{i+1}(x_{i+1}) \cdot e_{i+1}$. Then, line 13 uses $e_i = D\mbox{*}\omega_{i+1}(x_{i+1}) \cdot (F(x; \theta) - y)$ to calculate $\nabla_{\theta_i} J_R(x, y; \theta)$ as per \eqref{eq:derj}.

It will be verified that two error vectors are identical.
\begin{equation} \label{eq:errvec}
    e_i = D\mbox{*} \omega_{i+1}(x_{i+1}) \cdot (\widehat{y}_R - y),
\end{equation}
\begin{equation} \label{eq:errvec2}
    e_{i} = D\mbox{*}f_{i+1}(x_{i+1}) \cdot e_{i+1}.
\end{equation}

By substituting \eqref{eq:errvec} in \eqref{eq:errvec2} we'll obtain 

\begin{equation} \label{eq:errvec3}
\begin{split}
e_{i-1} = D\mbox{*}f_{i}(x_{i}) \cdot e_{i} & = D\mbox{*}f_{i}(x_{i}) \cdot D\mbox{*} \omega_{i+1}(x_{i+1}) \cdot (\widehat{y}_R - y) \\
& = D\mbox{*} \omega_i(x_i) \cdot (\widehat{y}_R - y).
\end{split}
\end{equation}

One can rewrite \eqref{eq:errvec3} by increasing index by one:
$$e_i = D\mbox{*} \omega_{i+1}(x_{i+1}) \cdot (\widehat{y}_R - y),$$

which will return us to \eqref{eq:errvec}. \\
Now, it is straightforward to conclude that

\begin{equation}
\nabla_{\theta_i} J_R(x, y; \theta) = \nabla_{\theta_i}^{\mbox{*}} f_i(x_i) \cdot D\mbox{*}\omega_{i+1}(x_{i+1}) \cdot (\widehat{y}_R - y) = \nabla_{\theta_i}^{\mbox{*}} f_i(x_i) \cdot e_i .
\end{equation}

The essential results we obtained thus far which were used in the algorithm are as follows:
\begin{itemize}

    \item $D\mbox{*} \omega_i(x_i) = D\mbox{*} f_i(x_i) \cdot D\mbox{*} \omega_{i+1}(x_{i+1}),$

    \item $e_i = D\mbox{*} \omega_{i+1}(x_{i+1}) \cdot (\widehat{y}_R - y) = D\mbox{*}f_{i+1}(x_{i+1}) \cdot e_{i+1}.$
    
    \item $\nabla_{\theta_i}^{\mbox{*}} F(x; \theta) = \nabla_{\theta_i}^{\mbox{*}} f_i(x_i) \cdot D\mbox{*}\omega_{i+1}(x_{i+1}),$
    
    \item $\nabla_{\theta_i} J_R(x, y; \theta) = \nabla_{\theta_i}^{\mbox{*}} F(x; \theta) \cdot (F(x; \theta) - y) = \nabla_{\theta_i}^{\mbox{*}} f_i(x_i) \cdot D\mbox{*}\omega_{i+1}(x_{i+1}) \cdot (\widehat{y}_R - y) = \nabla_{\theta_i}^{\mbox{*}} f_i(x_i) \cdot e_i,$


\end{itemize}

One can extend Algorithm \ref{alg:backppg} linearly to a batch of input points $\{(x_{(j)},y_{(j)})\}_{j \in A},$ where $A \subset [n],$ by averaging the contribution to the gradient from each point $(x_{(j)},y_{(j)})$ over the batch. One can also extend Algorithm \ref{alg:backppg} to more complex versions of gradient descent, e.g. momentum and adaptive gradient step methods. These methods, however, are not in the scope of this thesis. Once can also incorporate a simple form of regularization into this framework as described in the following remark.

\newtheorem*{remark}{Remark}

\begin{remark}
It is straightforward to incorporate a standard $\ell_2$-regularization term into this framework. Consider a new objective function $\mathcal{J}_T(x, y; \theta) = J_R(x, y, \theta) + \lambda T(\theta),$ where $\lambda \in \mathbb{R}_+$ is the regularization parameter, and

$$T(\theta) = \frac{1}{2} \norm{\theta}^2 = \frac{1}{2} \sum_{i=1}^{L} \norm{\theta_i}^2 = \frac{1}{2} \sum_{i=1}^{L} \langle \theta_i, \theta_i \rangle$$

is the regularization term. It follows that $\nabla_{\theta_i} \mathcal{J}_T(x, y; \theta) = \nabla_{\theta_i}J_R(x, y; \theta) + \lambda\theta_i,$ since $\nabla_{\theta_i}T(\theta) = \theta_i$. The reason of this can be found in [\cite{dnnmf}, Chapter 3, Section 3]. Consequently, gradient descent can be updated to include the regularizing term, i.e. line 14 can be changed in Algorithm \ref{alg:backppg} to 

$$\theta_i \leftarrow \theta_i - \eta (\nabla_{\theta_i}J_R(x, y; \theta) + \lambda\theta_i).$$

\end{remark}

\begin{algorithm}[H]
  \caption{One iteration of gradient descent for a generic neural network}
   \label{alg:backppg}
  \begin{algorithmic}[1]
    \STATE \textbf{function} \textsc{GradStepNN}%
$(x, y, \theta, \eta)$
    \INDSTATE $x_1 \leftarrow x$ 
    \begin{ALC@g}
    \FOR{$i \in \{1,...,L\}$} 
        \STATE {$x_{i+1} \leftarrow f_i(x_i)$} \COMMENT{$x_{L+1} = F(x; \theta);$ forward propagation step} 
    \ENDFOR
    
    \FOR{$i \in \{L,...,1\}$} 
        \STATE $\tilde{\theta_i} \leftarrow \theta$ \COMMENT{Store old $\theta_i$ for updating $\theta_{i-1}$}
        \IF {$i=L$} 
        \STATE $e_L \leftarrow x_{L+1} - y$
        \ELSE
        \STATE $e_i \leftarrow D{\mbox{*}}f_{i+1}(x_{i+1}) \cdot e_{i+1}$ \COMMENT{Update with $\tilde{\theta}_{i+1}$; backpropagation step}
        \ENDIF
        \STATE $\nabla_{\theta_i}J_R(x, y; \theta) \leftarrow \nabla_{\theta_i}^{\mbox{*}} f_i(x_i) \cdot e_i$ 
        \STATE $\theta_i \leftarrow \theta_i - \eta\nabla_{\theta_i} J_R(x, y; \theta)$ \COMMENT{Parameter update setup}
    \ENDFOR
    \RETURN $\theta$
   \end{ALC@g}
  \end{algorithmic}
\end{algorithm}

One can also consider a higher-order loss function that penalizes the first derivative of the network output. For more detail on this, see [\cite{dnnmf}, Chapter 3, Section 3]

\subsection{Conclusion}

In this section, a generic mathematical framework for layered neural networks was developed. Derivatives were calculated with respect to the parameters of each layer for standard loss functions, demonstrating how to do this directly over the vector space in which the parameters are defined. This generic framework will be used to represent specific network structures (such as RNNs and more specifically, gated RNNs and finally LSTM) in the next section.

\clearpage
\section{Recurrent Neural Networks} \label{sec:nnmf-rnn}

A mathematical framework for a generic layered neural network was developed in the preceding section, including a method to express error backpropagation and loss function derivatives directly over the inner product space in which the network parameters are defined. This chapter will be dedicated to expressing Recurrent Neural Networks. Afterward, we will concentrate on Vnilla RNNs and then gated RNNs and eventually narrow it down to LSTMs.

\subsection{Generic RNN Forumation}

The framework developed in \ref{nnf} will be altered to describe the RNN, as it is a different style of neural network. At first the notation for sequences is introduced, then the forward propagation of the hidden state will be discussed, and finally, the loss function and the BPTT gradient descent methods for RNN will be introduced.

\subsubsection{Sequence Data}
In the most general case, the input to an RNN, which is denoted by $\mathbf{x}$, is a sequence of bounded length, i.e.

$$\mathbf{x} \equiv (x_1,...,x_L) \in \underbrace{E_x \times \cdots \times E_x}_\text{$L$ times} \equiv E_x^L ,$$

where $E_x$ is some inner product space, $E_x^L$ is shorthand for the direct product of $L$ copies of $E_x$, and $L \in \mathbb{Z}_+$ is the maximum sequence length for the particular problem. One can also write the RNN target variables, which is denoted by $\mathbf{y}$, as a sequence of bounded length, i.e.

$$\mathbf{y} \equiv (y_1,...,y_L) \in \underbrace{E_y \times \cdots \times E_y}_\text{$L$ times} \equiv E_y^L ,$$

where $E_y$ is also an inner product space.\\

When using an RNN, the dataset will be of the form $\mathcal{D} = \{(\mathbf{x}_{(j)}, \mathbf{y}_{(j)})\}_{j=1}^n$ where $(\mathbf{x}_{(j)},\mathbf{y}_{(j)}) \in E_x^L \times E_y^L$ for all $j \in [n].$ However, sequences are generally of varying length, so any particular $\mathbf{x}_{(j)}$ may only have $\ell < L$ elements; for those points, we will simply not calculate the loss or prediction beyond the $\ell$th layer of the network. Similarly, a given $\mathbf{y}_{(j)}$ may not contain a target value for each $i \in [L]$ again, the loss will be calculated only when there is actually a target value. Thus, without loss of generality, we will only present the case where the data point we are considering, $(\mathbf{x}_{(\newoptimal{j})}, \mathbf{y}_{(\newoptimal{j})}) \equiv (\mathbf{x}, \mathbf{y}) \in \mathcal{D}$ is full, i.e. $\mathbf{x}$ is of length $L$ and $\mathbf{y}$ contains $L$ target points. 

\subsubsection{Hidden States, Parameters and Forward Propagation}
One feature that makes RNNs unique is that they contain a hidden state—initialized independently from the inputs—that is propagated forward at each layer $i$. Note that in the context of RNNs, one layer will be considered both the evolution of the hidden state and the resulting prediction generated post-evolution. The inner product space of hidden states will be referred as $E_h$. The method of propagating the hidden state forward is also the same at each layer, which is another unique property of RNNs. It is governed by the same functional form and the same set of \emph{transition} parameters $\theta \in H_T$, where $H_T$ is some inner product space. This is the \emph{recurrent} nature of RNNs: each layer performs the same operations on the hidden state, with the only difference between layers being that the input data is $x_i \in E_x$ at layer $i \in [L].$ \\

To solidify this concept, a generic layerwise function $f:E_h \times E_x \times H_T \to E_h$ is introduced that governs the propagation of the hidden state forward at each layer. One can express this for any $h \in E_h, x \in E_x$ and $\theta \in H_T$ as 

$$f(h; x; \theta) \in E_h.$$

Now consider a data point $\mathbf{x} \in E_x^L$ as described above. It is asserted that the $i$th layer of the RNN will take as input the $(i - 1)$th hidden state, which is denoted by $h_{i-1} \in E_h,$ and the $i$th value of $\mathbf{x},$ which is $x_i \in E_x,$ for all $i \in [L].$ The forward propagation of the hidden state after the $i$th layer is given by 
$$h_i \equiv f(h_{i-1}; x_i; \theta),$$

where $h_0 \in E_h$ is the initial hidden state, which can either be learned as a parameter or initialized to some fixed vector. For ease of composition, one again the parameters of $f$, will be suppressed as well as the input $x_i$ in this formulation such that

$$h_i \equiv f_i(h_{i-1})$$

for all $i \in [L].$ \footnote{A slightly different indexing convenction is adpoted in this section---notice that $f_i$ takes in $h_{i-1}$ and outputs $h_i$, as opposed to the previous section where the state variable is evolved according to $x_{i+1} = f_i(x_i).$ This indexing convention is more natural for RNNs, it will be shown that the $i$th prediction will depend on $h_i$ with this adjustment, instead of on $h_{i+1}$} Notice that $f_i$ retains implicit dependence on $x_i$ and $\theta$. The $h_i$ will be referred to as the \emph{state variable} for the RNN, as it is the quantity that we propagate forward at each layer. \\

One can define the head map as in \eqref{eq:alpha}, but with the argument corresponding to a hidden state, i.e. for all $i \in [L],$ we define $\alpha_i : E_h \to E_h$ as 

\begin{equation} \label{eq:alpha2}
    \alpha_i = f_i \circ \cdots \circ f_1,
\end{equation}

and $\alpha_0$ is defined to be the identity map on $E_h.$ If one views the RNN as a discrete-time dynamical system, he could also call $\alpha_i$ the \emph{flow} of the system. A new map will be introduced to aide in the calculation of derivatives $\mu_{j,i} : E_h \to E_h,$ which accumulates the evolution of the hidden state from layer $i \in [L]$ to $j \in \{i,..., L\}$ inclusive, i.e.

\begin{equation} \label{eq:mu}
    \mu_{j,i} = f_j \circ \cdots \circ f_i.
\end{equation}

The $\mu_{j,i}$ will be also set to be the identity on $E_h$ for $j < i,$ which we extend to include the case when $i > L,$ i.e.

$$\mu_{j,i} = id$$

whenever $i > min(j,L).$

\subsubsection{Prediction and Loss Functions}
Recall that there is a target variable at each layer $i \in [L],$ meaning that there should be a prediction at each layer. As in the previous subsection, It will be enforced that the prediction also has the same functional form and set of \emph{prediction} parameters at each layer. The prediction function $g$ takes in a hidden state $h \in E_h$ and a set of prediction parameters $\zeta \in H_P,$ and outputs an element of $E_y,$ i.e. $g : E_h \times H_P \to E_y.$ Often, the dependence of $g$ on the parameters is suppressed such that $g : E_h \to E_y$ again for ease of composition. Consequently, the prediction at layer $i \in [L]$ can be written in several ways:

\begin{equation} \label{eq:rnnoutput}
\widehat{y}_i = g(h_i) = (g \circ \mu_{i,k})(h_{k-1}) = (g \circ \alpha_i)(h) 
\end{equation}

for any $k \leq i,$ where $h_i = \alpha_i(h)$ for all $i \in [L],$ and $h \equiv h_0 \in E_h$ is the initial hidden state.

Since there is a prediction at each layer, there will be also a loss at each layer. The total loss for the entire network, $\mathcal{J}$, is the sum of these losses, i.e.

\begin{equation} \label{eq:totalloss}
\mathcal{J} = \sum_{i=1}^{L} J_R(y_i, \widehat{y}_i),    
\end{equation}

where $J : E_y \times E_y \to \mathbb{R}$ is the squared loss as in \eqref{eq:sqloss}. Recall that we can define the squared loss as 

\begin{equation} 
J_R(y,\widehat{y}) = \frac{1}{2} \langle y - \widehat{y}, y - \widehat{y} \rangle
\end{equation}

It is important to note that $\widehat{y}_i$ from \eqref{eq:totalloss} depends on the initial state $h,$ the transition parameters $\theta$, the prediction parameters $\zeta$, and the input sequence up to layer $i$, given by $\mathbf{x_i} \equiv (x_1,...,x_i).$

\subsubsection{Loss Function Gradients}
Taking derivatives of the loss function \eqref{eq:totalloss} with respect to the parameters will be required. One can easily take the derivatives of the loss with respect to the prediction parameters $\zeta$. As for the transition parameters $\theta$, there are two prevailing methods: RTRL, where derivatives will be sent forward throughout the network, and BPTT, where they will go through the entire network first and then derivatives are sent backward. In practice, basic RTRL is very slow compared to BPTT but one can derive it more intuitively than BPTT and so it serves as a good starting point. Furthermore, RTRL can sometimes be applicable to streams of data that must be processed as they arrive.

\subsubsection{Prediction Parameters}
We would like to compute $\nabla_{\zeta}\mathcal{J}$, where define $\mathcal{J}$ is defined in \eqref{eq:totalloss}. Since the differential operator $\nabla_\zeta$ is additive, we have

$$\nabla_\zeta \mathcal{J} = \sum_{i=1}^{L} \nabla_\zeta (J_R(y_i, \widehat{y}_i)),$$

where $J_R(y_i, \widehat{y}_i)$ is enclosed in parentheses to emphasize that at first $J_R(y_i, \widehat{y}_i)$ is evaluated, and then its derivative is taken with respect to $\zeta$.

\begin{theorem}
For any $y_i \in E_y, h_i \in E_h,$ and $i \in [L],$

\begin{equation} \label{eq:rnn-nabzetj}
    \nabla_\zeta (J_R(y_i, \widehat{y}_i)) = 
    \nabla_{\zeta}^{\mbox{*}} g(h_i) \cdot e_i,
\end{equation}

where $\widehat{y}_i$ is defined in \eqref{eq:rnnoutput}, and $e_i = \widehat{y}_i - y_i$.

\end{theorem}

\begin{proof}
We can prove this theorem similarly to theorem \ref{thm:derj}, although the notation is a bit different. Suppose $J_R$ is the squared loss, then for any $i \in [L]$ and $U \in H_P,$

\begin{equation} 
\begin{split}
\nabla_\zeta (J_R(y_i, \widehat{y}_i)) \cdot U 
& = \nabla_\zeta (\tfrac{1}{2} \langle \widehat{y}_i - y_i, \widehat{y}_i - y_i \rangle) \cdot U \\
& = \nabla_\zeta (\tfrac{1}{2} \langle g(h_i) - y_i, g(h_i) - y_i \rangle) \cdot U\\
& = \langle g(h_i) - y_i, \nabla_\zeta g(h_i) \cdot U \rangle \\
& = \langle \nabla_\zeta^{\mbox{*}} g(h_i) \cdot (g(h_i) - y_i), U \rangle \\
& =\langle \nabla_{\zeta}^{\mbox{*}} g(h_i) \cdot (\widehat{y}_i - y_i), U \rangle, \\ 
& =\langle \nabla_{\zeta}^{\mbox{*}} g(h_i) \cdot e_i, U \rangle,
\end{split}
\end{equation}

where the third line is true since $h_i$ has no dependence on $\zeta$, The last line implies that the derivative map above is linear functional, i.e. $\nabla_\zeta^{\mbox{*}} J_R(y_i; \widehat{y}_i) \in L(H_P; \mathbb{R}).$ Once more, by the canonical isomorphism described in [\cite{nnm3}, Chapter 5, Section 3], we can represent $\nabla_{\zeta} J_R(y_i, \widehat{y}_i)$ as an element of $H_P$ as 

$$\nabla_\zeta J_R(y_i, \widehat{y}_i) =  \nabla_\zeta^{\mbox{*}} g(h_i) \cdot e_i.$$
\end{proof}

\subsubsection{Real-Time Recurrent Learning}
We will now proceed with the presentation of the RTRL algorithm for calculating the gradient of \eqref{eq:totalloss} with respect to the transition parameters $\theta$. We will first show the forward propagation of the derivative of the head map in lemma \ref{lemma:nabalpha}, and then proceed to calculate the derivatives of \eqref{eq:totalloss} with respect to $\theta$ in Theorem \ref{thm:rtrl}.

\begin{lemma} \label{lemma:nabalpha}
For any $h \in E_h$ and $i \in [L],$ with $\alpha_i$ defined in \eqref{eq:alpha2},

\begin{equation} \label{eq:nabalplh0}
    \nabla_{\theta}^{\mbox{*}} \alpha_i(h) = \nabla_{\theta}^{\mbox{*}} \alpha_{i-1}(h) \cdot D\mbox{*}f_i(h_{i-1}) + \nabla_{\theta}^{\mbox{*}} f_i(h_{i-1}),
\end{equation}

where $h_{i-1} = \alpha_{i-1}(h)$
\end{lemma}

\begin{proof}
We know that for any $i \in [L], \alpha_i = f_i \circ \alpha_{i-1}.$ Since both $f_i$ and $\alpha_{i-1}$ depend on $\theta,$ to take the derivative of their composition we must combine the chain rule with the product rule: first hold $\alpha_{i-1}$ constant with respect to $\theta$ and differentiate $f_i$, and then hold $f_i$ constant with respect to $\theta$ and differentiate $\alpha_{i-1}$. In particular,

\begin{equation} \label{eq:nabalph00}
\nabla_\theta \alpha_i (h) = \nabla_\theta (f_i \circ \alpha_{i-1})(h) = \nabla_\theta f_i(h_{i-1}) + Df_i(h_{i-1}) \cdot \nabla_\theta \alpha_{i-1}(h) 
\end{equation}

since $h_{i-1} = \alpha_{i-1}(h).$ Then, by taking the adjoint, we recover \eqref{eq:nabalplh0}. Note that \eqref{eq:nabalplh0} still holds when $i = 1$, as $\alpha_0$ is the identity on $E_h$ with no dependence on the parameters $\theta,$ and thus $\nabla_{\theta}^{\mbox{*}} \alpha_0(h)$ is the zero operator. 

\end{proof}

\begin{theorem}[\textbf{Real-Time Recurrent Learning}] \label{thm:rtrl}
For any $h \in E_h, y_i \in E_y,$ and $i \in [L],$

\begin{equation} \label{eq:gradj-rnn}
\nabla_\theta (J_R(y_i,\widehat{y}_i)) = \nabla_\theta^{\mbox{*}} \alpha_i(h) \cdot D\mbox{*} g(h_i) \cdot e_i,     
\end{equation}

where $h_i = \alpha_i(h), \alpha_i$ is defined in \eqref{eq:alpha2}, $\widehat{y}_i$ is defined in \eqref{eq:rnnoutput} and $e_i = \widehat{y}_i - y_i$.

\end{theorem}

\begin{proof}
For any $U \in H_T,$

\begin{subequations}
\begin{align}
\nabla_\theta (J_R(y_i,\widehat{y_i})) \cdot U & = \nabla_\theta \Big ( \tfrac{1}{2} \langle g(\alpha_i(h)) - y_i, g(\alpha_i(h)) - y_i \rangle \Big ) \\ 
& = \langle g(\alpha_i(h)) - y_i, D g(h_i) \cdot \nabla_\theta \alpha_i(h) \cdot U \rangle \\ 
& = \langle \nabla_\theta^{\mbox{*}} \alpha_i (h) \cdot D\mbox{*} g(h_i) \cdot (\widehat{y_i} - y_i), U \rangle.
\end{align}
\end{subequations}

Therefore, by the canonical isomorphism described in  [\cite{nnm3}, Chapter 5, Section 3], we have proven \eqref{eq:gradj-rnn}.

\end{proof}

Note that even though we do not have access to $e_i$ and $h_i$ until layer $i$, we can still propagate the linear map $\nabla_\theta^{\mbox{*}} \alpha_i(h)$ forward without an argument at each layer $i$ according to \eqref{eq:nabalplh0}, and then use this to calculate \eqref{eq:gradj-rnn}. This is the \emph{real-time} aspect of RTRL, as it follows for exact gradient computation at each layer $i$ without knowledge of the information at future layers. Unfortunately, this forward propagation is also what makes RTRL slow compared to BPTT. Nevertheless, we present a generic algorithm for performing one step of gradient descent via RTRL in Algorithm \ref{alg:rtrl}. As input to the algorithm, we provide the sequence input $\mathbf{x}$ and associated targets $\mathbf{y}$, the initial state $h$, the transition parameters $\theta$, the prediction parameters $\zeta$ and the learning rate $\eta$. We receive, as output, a parameter set updated by a single step of gradient descent.

The essential results we obtained thus far which were used in the algorithm are as follows:

\begin{itemize}

    \item $\nabla_\zeta \mathcal{J} = \sum_{i=1}^{L} J_R(y_i, \widehat{y}_i) = \sum_{i=1}^{L} \nabla_\zeta^{\mbox{*}} g(h_i) \cdot e_i,$

    \item $\nabla_\theta^{\mbox{*}} \alpha_i(h) = \nabla_\theta^{\mbox{*}} \alpha_{i-1}(h) \cdot D\mbox{*} f_i(h_{i-1}) + \nabla_\theta^{\mbox{*}} f_i(h_{i-1}),$
    
    \item $\nabla_\theta J_R(y_i, \widehat{y}_i) = \nabla_\theta^{\mbox{*}} \alpha_i(h) \cdot D\mbox{*} g(h_i) \cdot e_i = \big( \nabla_\theta^{\mbox{*}} \alpha_{i-1} (h) \cdot D\mbox{*} f_i(h_{i-1}) + \nabla_\theta^{\mbox{*}} f_i(h_{i-1}) \big) \cdot D\mbox{*} g(h_i) \cdot e_i,$
    
    \item $\mu_{j,i+1} (h_i) = D \mu_{j, i+2} (h_{i+1}) \cdot Df_{i+1}(h_i),$
    
    \item $D_{h_i} \mathcal{J} = D\mbox{*} f_{i+1} (h_i) \cdot Dh_{i+1} \mathcal{J} + D\mbox{*} g(h_i) \cdot e_i,$

    \item $\nabla_\theta \alpha_k(h) = \sum_{j=1}^{k} D \mu_{k,j+1} (h_J) \cdot \nabla_\theta f_j(h_{j-1}),$
    
    \item $\nabla_\theta \mathcal{J} = \sum_{i=1}^{L} \nabla_\theta^{\mbox{*}} f_i(h_{i-1}) \cdot D_{h_i} \mathcal{J},$

    \item $\nabla_\theta \mathcal{J} = \sum_{i=1}^{L} \nabla_\theta^{\mbox{*}} \alpha_j(h) \cdot D\mbox{*} g(h_j) \cdot e_j = \sum_{j=1}^{L} \sum_{i=1}^{j} \nabla_\theta^{\mbox{*}} f_i(h_{i-1}) \cdot D\mbox{*} \mu_{j,i+1} (h_i) \cdot D\mbox{*} g(h_j) \cdot e_j .$
\end{itemize}

\begin{algorithm}[H]
  \caption{One iteration of gradient descent for for an RNN via RTRL}
   \label{alg:rtrl}
  \begin{algorithmic}[1]
    \STATE \textbf{function} \textsc{GradDescRTRL}%
$(\mathbf{x}, \mathbf{y}, h, \theta, \zeta, \eta)$
    \INDSTATE $h_0 \leftarrow h$ 
    \begin{ALC@g}
    \STATE $\nabla_\theta \mathcal{J} \leftarrow 0$ \COMMENT{0 in $H_T$, the inner product space in which $\theta$ resides}
    \STATE $\nabla_\zeta \mathcal{J} \leftarrow 0$ \COMMENT{0 in $H_P$, the inner product space in which $\zeta$ resides}
    \FOR{$i \in \{1,...,L\}$} 
        \STATE $h_i \leftarrow f_i(h_{i-1})$ \COMMENT{$f_i$ depends on $\theta, x_i$}
        \STATE $\widehat{y}_i \leftarrow g(h_i)$
        \STATE $\nabla_\theta^{\mbox{*}} \alpha_i(h) \leftarrow \nabla_\theta^{\mbox{*}} \alpha_{i-1}(h) \cdot D\mbox{*} f_i(h_{i-1}) + \nabla_\theta^{\mbox{*}} f_i(h_{i-1})$
        
        \STATE $e_i \leftarrow \widehat{y}_i - y_i$
        \STATE $\nabla_\theta \mathcal{J} \leftarrow \nabla_\theta \mathcal{J} + \nabla_\theta^{\mbox{*}} \alpha_i(h) \cdot D\mbox{*}g(h_i) \cdot e_i$ \COMMENT{Add accumulated gradient at each layer}
        \STATE $\nabla_\zeta \mathcal{J} \leftarrow \nabla_\zeta \mathcal{J} + \nabla_\zeta^{\mbox{*}} g(h_i) \cdot e_i$
    \ENDFOR
    \STATE $\theta \leftarrow \theta - \eta \nabla_\theta \mathcal{J}$ \COMMENT{Parameter Update Steps}
    \STATE $\zeta \leftarrow \zeta - \eta \nabla_\zeta \mathcal{J}$
    \RETURN $\theta, \zeta$
   \end{ALC@g}
  \end{algorithmic}
\end{algorithm}

\subsubsection{Backpropagation Through Time}
We can derive a more efficient method for gradient calculation with respect to the transition parameters in RNNs known as BPTT. Even though for executing BPTT, we must traverse the network both forwards and backwards, combination of the two are yet far more computationally efficient than RTRL \cite{dnnmf}. Note that we will use the notation $D_{h_i}$ to denote the action of taking the derivative with respect to the state $0h_i$ in this section, for any $i \in [L].$ We use this, as opposed to $\nabla_{h_i},$ since $h_i$ is a state variable. \\

The first part of BPTT that we will derive is the backpropagation step, which sends the error at layer $i \in [L]$ backwards throughout the network. To do this, we will calculate $D{\mu_{j,i+1}}(h_i)$ for $j \leq i + 1$ in Lemma \ref{lemma:mu}, and then use this result to derive the recurrence in Theorem \ref{thm:bptt-rec}

\begin{lemma} \label{lemma:mu}
For any $h_i \in E_h, i \in [L-1]$ and $j \in [L]$ with $j \geq i+1,$

\begin{equation}
    D{\mu_{j,i+1}}(h_i) = D{\mu_{j,i+2}}(h_{i+1}) \cdot D{f_{i+1}}(h_i) 
\end{equation}

where $h_{i+1} = f_{i+1}(h_i)$ and $\mu_{j,i}$ is defined in \eqref{eq:mu}. Furthermore, $D\mu_{i,i+1}(h_i)$ is the identity map on $E_h$.

\end{lemma}

\begin{proof}
First of all, since $\mu_{i,i+1}$ is the identity map on $E_h$, we automatically have that $D\mu_{i,i+1}(h_i)$ is the identity on $E_h$. \\
Furthermore, for $j \geq i+1$, by the definition of $\mu_{j,i+1}$ we have that 

$$\mu_{j,i+1} = \mu_{j,i+2} \circ f_{i+1}.$$

Therefore, by the chain rule, for any $h_i \in E_h$,
\begin{equation}
\begin{split}
D\mu_{j,i+1}(h_i) 
& = D(\mu_{j,i+2} \circ f_{i+1})(h_i) \\
& = D\mu_{j,i+2}(h_{i+1}) \cdot Df_{i+1}(h_i)
\end{split}
\end{equation}

since $h_{i+1} = f_{i+1}(h_i).$

\end{proof}

\begin{theorem} [\textbf{Backpropagation Through Time}] \label{thm:bptt-rec}

For any $i \in [L]$ and $h_i \in E_h$, with $\mathcal{J}$ defined as in \eqref{eq:totalloss},

\begin{equation} \label{eq:rnnlossm}
D_{h_i}\mathcal{J} = D\mbox{*}f_{i+1}(h_i) \cdot D_{h_{i+1}} \mathcal{J} + D\mbox{*}g(h_i) \cdot e_i,
\end{equation}

where we set $D_{h_{L+1}} \mathcal{J}$ to be the zero vector in $E_h$ and $e_i = \widehat{y}_i - y_i$.

\end{theorem}

\begin{proof}
We can prove this directly from the definition of $\mathcal{J}$ for the squared loss. For any $v \in E_h$,


\begin{subequations}
\begin{align}
D_{h_i}\mathcal{J} \cdot v & = D_{h_i} \Bigg( \sum_{j=1}^{L} \tfrac{1}{2} \langle \widehat{y}_j - y_j,\widehat{y}_j - y_j \rangle \Bigg) \cdot v \\
& = D_{h_i} \Bigg( \sum_{j=1}^{L} \tfrac{1}{2} \langle g(\alpha_j(h)) - y_j, g(\alpha_j(h)) - y_j \rangle \Bigg) \cdot v \\
& = D_{h_i} \Bigg( \sum_{j=i}^{L} \tfrac{1}{2} \langle g(\mu_{j,i+1}(h_i)) - y_j, g(\mu_{j,i+1}(h_i)) - y_j \rangle \Bigg) \cdot v \label{subeq:21}\\
& = \sum_{j=i}^{L} \langle g(\mu_{j,i+1}(h_i)) - y_j, D\big(g(\mu_{j,i+1}(h_i))\big) \cdot v \rangle \\
& = \sum_{j=i}^{L} \langle \underbrace{g(\mu_{j,i+1}(h_i))}_{g(h_j)} - y_j, D\big( \underbrace{g(\mu_{j,i+1}(h_i))}_{g(h_j)}\big) \cdot D\mu_{j,i+1} (h_i) \cdot v \rangle \label{subeq:22} \\
& = \sum_{j=i}^{L} \langle D\mbox{*}\mu_{j,i+1}(h_i) \cdot D\mbox{*}(g(h_j)) \cdot e_j, v \rangle
\end{align}
\end{subequations}

where \eqref{subeq:21} holds since the loss from layers $j < i$ is not impacted by $h_i$, \eqref{subeq:22} holds from the chain rule in \eqref{eq:cr}. Therefore, by the canonical isomorphism described in [\cite{nnm3}, Chapter 5, Section 3], we can represent $D_{h_i} \mathcal{J}$ as an element of $E_h$ according to 

\begin{equation} \label{eq:rnnloss}
    D_{h_i}\mathcal{J} = \sum_{j=i}^{L} D\mbox{*} \mu_{j,i+1}(h_i) \cdot D\mbox{*}g(h_j) \cdot e_j
\end{equation}
for any $i \in [L].$ We can manipulate \eqref{eq:rnnloss} as follows when $i < L:$

\begin{subequations}
\begin{align}
D_{h_i}\mathcal{J} & = D\mbox{*} \mu_{i,i+1}(h_i) \cdot D\mbox{*} g(h_i) \cdot e_i + \sum_{j=i+1}^{L} D\mbox{*} \mu_{j,i+1}(h_i) \cdot D\mbox{*} g(h_j) \cdot e_j\\
& = D\mbox{*}g(h_i) \cdot e_i + \sum_{j=i+1}^{L} D\mbox{*}f_{i+1}(h_i) \cdot D\mbox{*} \mu_{j,i+2}(h_{i+1}) \cdot D\mbox{*} g(h_j) \cdot e_j  \label{subeq:31} \\
& = D\mbox{*} g(h_i) \cdot e_i + D\mbox{*} f_{i+1}(h_i) \cdot \Bigg(\sum_{j=i+1}^{L} D\mbox{*}\mu_{j,i+2}(h_{i+1}) \cdot D\mbox{*}g(h_j) \cdot e_j \Bigg) \\
& = D\mbox{*}g(h_i) + D\mbox{*}f_{i+1}(h_i) \cdot D_{h_{i+1}}\mathcal{J}.
\end{align}
\end{subequations}

where \eqref{subeq:31} follows from Lemma \ref{lemma:mu} and the and the reversing property of the adjoint. We have proven \eqref{eq:rnnlossm} for $i < L$.

As for when $i = L,$ since we set $D_{h_{i+1}}\mathcal{J} = 0$, then:

\begin{equation}
\begin{split}
    D_{h_L}\mathcal{J} & = D\mbox{*}\mu_{L,L+1}(h_i) \cdot D\mbox{*}g(h_L) \cdot e_L \\
    & = D\mbox{*}g(h_L) \cdot e_L \\
    & = 0 + D\mbox{*}g(h_L) \cdot e_L \\
    & = id\cdot0 + D\mbox{*}g(h_L) \cdot e_L \\
    & = D\mbox{*}f_{L+1}(h_L)\cdot0 + D\mbox{*}g(h_L) \cdot e_L\\
    & = D\mbox{*}f_{L+1}(h_L) \cdot D_{h_{i+1}}\mathcal{J} + D\mbox{*}g(h_L) \cdot D\mbox{*}g(h_L) \cdot e_L
\end{split}
\end{equation}

So, we have proven \eqref{eq:rnnlossm} for all $i \in [L]$.

\end{proof}

\subsubsection{Remark} 
Here we have followed the convention that only $h_i$ s treated as an independent variable in computing the derivative of $\mathcal{J}$ with respect to $h_i$, which we can denote as $D_{h_i} \mathcal{J}$. There is some ambiguity here, however, since $h_i$ can be viewed as $\alpha_i (h_0)$. In order to avoid this ambiguity, we could just \textit{define} $D_{h_i} \mathcal{J}$ as the expression on the right-hand side in \eqref{eq:rnnloss}
without giving it the meaning of a derivative. We will see that Theorem \ref{thm:bptt-rec} will still hold under this assumption.

We will present the gradient of $\mathcal{J}$ with respect to the transition parameters for BPTT in Theorem \ref{thm:derj-rec} after first presenting a useful result in Lemma \ref{lemma:3-13}. The expression that we will derive relies heavily on the recursion from Theorem \ref{thm:bptt-rec}, similarly to how Theorem \ref{thm:gradj} depended on the recursion from Theorem \ref{thm:bckpp}

\begin{lemma} \label{lemma:3-13}
For any $k \in [L]$ and $h \in E_h$,

\begin{equation} \label{eq:nabalph}
    \nabla_\theta \alpha_k (h) = \sum_{j=1}^{k} D \mu_{k,j+1}(h_j) \cdot \nabla_\theta f_j(h_{j-1}),
\end{equation}

where $\alpha_k$ is defined in \eqref{eq:alpha2}, $h_j = \alpha_j(h)$ for all $j \in [L],$ and $\mu_{k,j+1}$ is defined in \eqref{eq:mu}.

\end{lemma}

\begin{proof}
We can prove this by induction. For $k = 1$, since $\alpha_1 = f_1$ and $h = h_0$, 
$$\nabla_\theta \alpha_1(h) = \nabla_\theta f_1(h_0).$$

Also, by Lemma \ref{lemma:mu} $D \mu_{1,2}(h_1)$ is the identity. Therefore, \eqref{eq:nabalph} is true for $k = 1$. Now assume \eqref{eq:nabalph} holds for $2 \leq k \leq L-1$. Then,

\begin{subequations}
\begin{align}
\nabla_\theta \alpha_{k+1}(h) & = D f_{k+1}(h_k) \cdot \nabla_\theta \alpha_k (h) + \nabla_\theta f_{k+1}(h_k) \label{subeq:41} \\
& = Df_{k+1}(h_k) \cdot \Bigg( \sum_{j=1}^{k} D\mu_{k,j+1}(h_j) \cdot \nabla_\theta f_j(h_{j-1}) \Bigg) \label{subeq:42} \\
& \hspace{3.5mm} + D\mu_{k+1,k+2}(h_{k+1}) \cdot \nabla_\theta f_{k+1}(h_k) \label{subeq:43}\\
& = \sum_{j=1}^{k} Df_{k+1}(h_k) \cdot D\mu_{k,j+1}(h_j) \cdot \nabla_\theta f_j(h_{j-1}) \\
& \hspace{3.5mm} + D\mu_{k+1,k+2}(h_{k+1}) \cdot \nabla_\theta f_{k+1}(h_k) \label{subeq:44}\\
& = \sum_{j=1}^{k+1} D\mu_{k+1,j+1}(h_j) \cdot \nabla_\theta f_j(h_{j-1})
\label{subeq:45}
\end{align}
\end{subequations}

where \eqref{subeq:41} follows from \eqref{eq:nabalph00}, (\eqref{subeq:42}) and \eqref{subeq:43} from the inductive hypothesis and the fact that $D\mu_{k+1,k+2}(h_{k+1})$ is the identity, and \eqref{subeq:45} comes from the fact that $f_{k+1} \circ \mu_{k,j+1} = \mu_{k+1,j+1},$ implying

$$D f_{k+1}(h_k) \cdot D\mu_{k,j+1}(h_j) = D\mu_{k+1,j+1}(h_j)$$

for $j \leq k.$ Thus, we have proven \eqref{eq:nabalph} for all $k \in [L]$ by induction.

\end{proof}

\begin{theorem} \label{thm:derj-rec}
For $\mathcal{J}$ defined as in \eqref{eq:totalloss},

\begin{equation} \label{eq:nabjtheta}
\nabla_\theta \mathcal{J} = \sum_{i=1}^{L} \nabla_\theta^{\mbox{*}} f_i(h_{i-1}) \cdot D_{h_i} \mathcal{J},   
\end{equation}

where we can write $D_{h_i} \mathcal{J}$ as an element of $E_h$ recursively according to \ref{thm:bptt-rec}
\end{theorem}

\begin{proof}
We can prove this directly using the results from earlier in this subsection:

\begin{equation}
\begin{split}
\nabla_\theta \mathcal{J} &  = \sum_{j=1}^{L} \nabla_\theta^{\mbox{*}} \alpha_j(h) \cdot D\mbox{*} g(h_j) \cdot e_j  \\
 & = \sum_{j=1}^{L} \sum_{i=1}^{j} \nabla_\theta^{\mbox{*}} f_i(h_{i-1}) \cdot D\mbox{*} \mu_{j,i+1} (h_i) \cdot D\mbox{*} g(h_j) \cdot e_j,
\end{split}
\end{equation}

where the first equality follows from summing \eqref{eq:gradj-rnn} over all $j \in [L],$ and the second line from taking the adjoint of \eqref{eq:nabalph}. We will now swap the incides to obtain the final result, since we are summing over $\{(i,j) \in [L] \times [L]: 1 \leq i \leq j \leq L \}:$

\begin{equation}
\begin{split}
\nabla_\theta \mathcal{J} &  = \sum_{i=1}^{L} \sum_{j=i}^{L} \nabla_\theta^{\mbox{*}} f_i(h_{i-1}) \cdot D\mbox{*} \mu_{j,i+1} (h_i) \cdot D\mbox{*} g(h_j) \cdot e_j \\
& = \sum_{i=1}^{L} \nabla_\theta^{\mbox{*}} f_i(h_{i-1}) \cdot \Bigg( \sum_{j=i}^{L} D\mbox{*} \mu_{j,i+1} (h_i) \cdot D\mbox{*} g(h_j) \cdot e_j \Bigg) \\
& = \sum_{i=1}^{L} \nabla_\theta^{\mbox{*}} f_i(h_{i-1}) \cdot D_{h_i} \mathcal{J},
\end{split}
\end{equation}

where the final result comes from \eqref{eq:rnnloss}

\end{proof}

We will now present an algorithm for taking one step of gradient descent
in BPTT. The inputs and outputs are the same as Algorithm \ref{alg:rtrl}, with the only difference being that we compute the gradient with respect to the transition parameters according to BPTT and not RTRL. We will denote the backpropagate error quantity in Algorithm \ref{alg:bptt-rnn} by:

$$\varepsilon_i \equiv D_{h_i} \mathcal{J}$$

for all $i \in [L+1]$. We can again extend Algorithm \ref{alg:bptt-rnn} to a batch of inputs, more complicated gradient descent algorithms, and regularization, as in Algorithm \ref{alg:backppg}. 

One important extension to the BPTT aglorithm given in Algorithm \ref{alg:bptt-rnn} is truncated BPTT, in which we run BPTT every $\ell < L$ timesteps down for a fixed $m < L$ steps \cite{lstm1}, and then reset the error vector to zero after. Truncated BPTT requires fewer computations than full BPTT and can also help with the problem of vanishing and exploding gradients, as the gradients will not be propagated back as far as in full BPTT. One potential downside is that the exact gradients will not be calculated, although this is preferable to exact gradients if they would otherwise explode. \\

\vfill

\begin{algorithm}[H]
  \caption{One iteration of gradient descent for for an RNN via BPTT}
   \label{alg:bptt-rnn}
  \begin{algorithmic}[1]
    \STATE \textbf{function} \textsc{GradDescBPTT}%
$(\mathbf{x}, \mathbf{y}, h, \theta, \zeta, \eta)$
    \INDSTATE $h_0 \leftarrow h$ 
    \begin{ALC@g}
    \STATE $\nabla_\theta \mathcal{J} \leftarrow 0$ \COMMENT{0 in $H_T$, the inner product space in which $\theta$ resides}
    \STATE $\nabla_\zeta \mathcal{J} \leftarrow 0$ \COMMENT{0 in $H_P$, the inner product space in which $\zeta$ resides}
    \FOR{$i \in \{1,...,L\}$} 
        \STATE $h_i \leftarrow f_i(h_{i-1})$ \COMMENT{$f_i$ depends on $\theta, x_i$}
        \STATE $\widehat{y}_i \leftarrow g(h_i)$
        \STATE $e_i \leftarrow \widehat{y}_i - y_i$
        \STATE $\nabla_\zeta \mathcal{J} \leftarrow \nabla_\zeta \mathcal{J} + \nabla_\zeta^{\mbox{*}} g(h_i) \cdot e_i$ \COMMENT{Add accumulated gradient at each layer}
    \ENDFOR
    \STATE $\varepsilon_{L+1} \leftarrow 0$ \COMMENT{0 in $E_h$; Initialization of $D_{h_{L+1}} \mathcal{J}$}
    \FOR{$i \in \{1,...,L\}$} 
        \STATE $\varepsilon_i \leftarrow D\mbox{*} f_{i+1} \cdot \varepsilon_{i+1} + D\mbox{*} g(h_i) \cdot e_i$ \COMMENT{BPTT update step from \eqref{eq:rnnlossm}}
        \STATE $\nabla_\theta \mathcal{J} \leftarrow \nabla_\theta \mathcal{J} + \nabla_\theta^{\mbox{*}} f_i(h_{i-1}) \cdot \varepsilon_i$ \COMMENT{Add accumulated gradient at each layer}
    \ENDFOR
    
    \STATE $\theta \leftarrow \theta - \eta \nabla_\theta \mathcal{J}$ \COMMENT{Parameter Update Steps}
    \STATE $\zeta \leftarrow \zeta - \eta \nabla_\zeta \mathcal{J}$
    \RETURN $\theta, \zeta$
   \end{ALC@g}
  \end{algorithmic}
\end{algorithm}

\clearpage

\subsection{Vanilla RNNs}
We will now formulate basic \emph{vanilla} RNN in the framework of the previous subsection. We first need to specify the hidden, input, output and parameter spaces, the layerwise function $f$, and the prediction function $g$. We will also take the derivatives of $f$ and $g$ to develop the BPTT method for vanilla RNNs. 

\subsubsection{Formulation}
Let us assume the hidden state is a vector of length $n_h$, i.e. $E_h = \mathbb{R}^{n_h}$. Suppose also that $E_x = \mathbb{R}^{n_x}$ and $E_y = \mathbb{R}^{n_y}$. We will evolve the hidden state $h \in \mathbb{R}^{n_h}$ according to a hidden-to-hidden weight matrix $W \in \mathbb{R}^{n_h \times n_h}$, an input-to-hidden weight matrix $U \in \mathbb{R}^{n_h \times n_x}$, and a bias vector $b \in \mathbb{R}^{n_h}$. We can then describe the hidden state evolution as

$$f(h;x;W,U,b) = \Psi(W \cdot h + U \cdot x + b),$$

where $\Psi : \mathbb{R}^{n_h} \to \mathbb{R}^{n_h}$ is the elementwise nonlinearity as defined in subsection \ref{subsec:elemwise funcs}. The tanh function is a particularly popular choice of elementwise nonlinearity for RNNs. If we employ the parameter and input suppression convention for each layer $i \in [L]$, we can write the layerwise function $f_i$ as

\begin{equation} \label{eq:van-f}
    f_i(h_{i-1}) = \Psi (W \cdot h_{i-1} + U \cdot x_i + b).    
\end{equation}

The prediction function $g$ is also parametrized by matrix-vector multiplication as follows for any $h \in \mathbb{R}^{n_h}:$

\begin{equation} \label{eq:van-g}
    g(h) = V \cdot h + c,
\end{equation}

where $V \in \mathbb{R}^{n_y \times n_h}$ is the hidden-to-output weight matrix, and $c \in \mathbb{R}^{n_y}$ is the output bias vector. We assume in this subsection that each vector space is equipped with the standard Euclidean inner product $\langle A,B \rangle = tr(A B^T) = tr(A^T B)$. \\

\subsubsection{Single-Layer Derivatives}
We will first derive the maps $Df$ and $\nabla_\theta f,$ for $\theta \in \{W, U, b\},$ and their adjoints. Then, we will derive $Dg$ and $\nabla_\zeta g,$ for $\zeta \in \{V, c\},$ and the adjoints of those as well. 

\begin{theorem}
For any $h_{i-1} \in \mathbb{R}^{n_h}, x_i \in \mathbb{R}^{n_x}, \widetilde{W} \in \mathbb{R}^{n_h \times n_h},$ and $\widetilde{U} \in \mathbb{R}^{n_h \times n_x},$ with $f_i$ defined as in \eqref{eq:van-f},

\begin{subequations}
\begin{align}
D f_i(h_{i-1}) & = D \Psi (z_i) \cdot W, \label{subeq:51}\\
\nabla_W f_i(h_{i-1}) \cdot \widetilde{W} & = D \Psi (z_i) \cdot \widetilde{W} \cdot h_{i-1}, \label{subeq:52} \\
\nabla_U f_i(h_{i-1}) \cdot \widetilde{U} & = D \Psi (z_i) \cdot \widetilde{U} \cdot x_i, \label{subeq:53} \\
\nabla_b f_i(h_{i-1}) & = D \Psi (z_i), \label{subeq:54},
\end{align}
\end{subequations}

where $z_i = W \cdot h_{i-1} + U \cdot x_i + b.$ Furthermore, for any $v \in \mathbb{R}^{n_h}$,

\begin{subequations}
\begin{align}
D\mbox{*} f_i(h_{i-1}) & = W^T \cdot D \Psi (z_i) , \label{subeq:61}\\
\nabla_W^{\mbox{*}} f_i(h_{i-1}) \cdot v & = (D \Psi (z_i) \cdot v) h_{i-1}^T \cdot , \label{subeq:62} \\
\nabla_U^{\mbox{*}} f_i(h_{i-1}) \cdot v & = (D \Psi (z_i) \cdot v) \cdot {x_i}^T, \label{subeq:63} \\
\nabla_b^{\mbox{*}} f_i(h_{i-1}) & = D \Psi (z_i), \label{subeq:64},
\end{align}
\end{subequations}

\end{theorem}

\begin{proof}
Equations \eqref{subeq:51} to \eqref{subeq:54} are all direct consequences of the chain rule. Equations \eqref{subeq:61} and \eqref{subeq:64} follow directly from the reversing property of the adjoint and the self-adjointness of $D\Psi$ (Theorem \ref{thm:elemfunc1stder}).

To prove Equation \eqref{subeq:62}, By \eqref{subeq:52}, for any $v \in \mathbb{R}^{n_h}$, $h_{i-1} \in \mathbb{R}^{n_h}$, and $\widetilde{W} \in \mathbb{R}^{n_h \times n_h}$,

\begin{equation}
\begin{split}
\langle \nabla_W^{\mbox{*}} f_i(h_{i-1}) \cdot v, \widetilde{W} \rangle &  = \langle v, \nabla_W f_i(h_{i-1}) \cdot \widetilde{W} \rangle \\
& = \langle v, D\Psi(z_i) \cdot \widetilde{W} \cdot h_{i-1} \rangle \\
& = \langle D\Psi(z_i) \cdot v, \widetilde{W} \cdot h_{i-1} \rangle \\
& = \langle (D\Psi(z_i) \cdot v){h_{i-1}^\intercal}, \widetilde{W} \rangle,
\end{split}
\end{equation}

where the forth quality arises from cyclic property of the trace. Since this is true for all $\widetilde{W} \in \mathbb{R}^{n_h \times n_h}$, 

$$\nabla_W^{\mbox{*}} f_i(h_{i-1}) \cdot v = (D\Psi(z_i) \cdot v) h_{i-1}^\intercal = (\Psi'_i(z_i) \odot v) h_{i-1}^\intercal .$$

The same approach can be used to obtain \eqref{subeq:63}: By \eqref{subeq:52}, for any $v \in \mathbb{R}^{n_h}$, $h_{i-1} \in \mathbb{R}^{n_h}$, and $\widetilde{U} \in \mathbb{R}^{n_h \times n_x}$,

\begin{equation}
\begin{split}
\langle \nabla_U^{\mbox{*}} f_i(h_{i-1}) \cdot v, \widetilde{U} \rangle &  = \langle v, \nabla_U f_i(h_{i-1}) \cdot \widetilde{U} \rangle \\
& = \langle v, D\Psi(z_i) \cdot \widetilde{U} \cdot x_i \rangle \\
& = \langle D\Psi(z_i) \cdot v, \widetilde{U} \cdot x_i \rangle \\
& = \langle (D\Psi(z_i) \cdot v){x_i}^{\intercal}, \widetilde{U} \rangle.
\end{split}
\end{equation}

 Since this is true for all $\widetilde{U} \in \mathbb{R}^{n_h \times n_x}$,
 
 $$ \nabla_U^{\mbox{*}} f_i(h_{i-1}) \cdot v = (D\Psi(z_i) \cdot v){x_i}^{\intercal} = (\Psi'_i(z_i) \odot v) x_i^\intercal .$$
 
\end{proof}

\begin{theorem} \label{thm:van-ders}
For any $h \in E_h$ and $\widetilde{V} \in \mathbb{R}^{n_y \times n_h}$, with $g$ defined as in \eqref{eq:van-g},
\end{theorem}

\begin{subequations}
\begin{align}
Dg(h) & = V, \label{subeq:71}\\
\nabla_V g(h) \cdot \widetilde{V} & = \widetilde{V} \cdot h, \label{subeq:72} \\
\nabla_c g(h) & = id. \label{subeq:73}
\end{align}
\end{subequations}

Furthermore, for any $v \in \mathbb{R}^{n_y},$

\begin{subequations}
\begin{align}
D\mbox{*}g(h) & = V^T, \label{subeq:81}\\
\nabla_V^{\mbox{*}} g(h) \cdot v & = v h^\intercal, \label{subeq:82} \\
\nabla_c^{\mbox{*}} g(h) & = id. \label{subeq:83}
\end{align}
\end{subequations}

\begin{proof}
Equations \eqref{subeq:71}-\eqref{subeq:73} are consequences of the chain rule and Eqs. \eqref{subeq:81}-\eqref{subeq:83} are simpler versions of their counterparts in Theorem \ref{thm:van-ders}.
\end{proof}

\subsubsection{Backpropagation Through Time}
In this subsection, we will explicitly write out the BPTT recurrence \eqref{eq:rnnlossm} and full gradient \eqref{eq:rnnlossm} for the case of vanilla RNNs. Then, we can conveniently insert these into Algorithm \ref{alg:bptt-rnn} to perform BPTT. The equations that we will bring from \cite{dnnmf} derive bear a strong resemblance to those found in [\cite{dl}, Chapter 10]; however, \cite{dnnmf} have explicitly shown the derivation here and have carefully defined the maps and vectors that the book is using.

\begin{theorem} \label{thm:van-derj}
For any $i \in [L]$,

\begin{equation}
    D_{h_i}\mathcal{J} = W^T \cdot D\Psi(z_{i+1}) \cdot D_{h_{i+1}}\mathcal{J} + V^T \cdot e_i,    
\end{equation}

where $\mathcal{J}$ is defined in \eqref{eq:totalloss}, $z_{i+1} = W \cdot h_i + U \cdot x_{i+1} + b, e_i$ is $\widehat{y}_i - y_i,$ and we set $D_{h_{L+1}}\mathcal{J}$ to be the zero vector in $\mathbb{R}^{n_h}.$
\end{theorem}

\begin{proof}
We can prove this simply by inserting the definitions of $D\mbox{*}f_i$ and $D\mbox{*}g$ from \eqref{subeq:61} and \eqref{subeq:81}, respectively, into \eqref{eq:rnnlossm}.
\end{proof}

\begin{theorem} \label{thm:van-nabj}
For $\mathcal{J}$ defined as in \eqref{eq:totalloss},

\begin{equation}
\begin{split}
\nabla_W \mathcal{J} &  = \sum_{i=1}^{L} \big( D\Psi(z_i) \cdot D_{h_i} \mathcal{J} \big) h_{i-1}^\intercal, \\
\nabla_U \mathcal{J} &  = \sum_{i=1}^{L} \big( D\Psi(z_i) \cdot D_{h_i} \mathcal{J} \big) x_i^\intercal, \\
\nabla_b \mathcal{J} &  = \sum_{i=1}^{L} D\Psi(z_i) \cdot D_{h_i} \mathcal{J},
\end{split}
\end{equation}

where $h_i = \alpha_i(h)$ for all $i \in [L],$ and $D_{h_i}\mathcal{J}$ can be calculated recursively according to Theorem (\ref{thm:van-derj}).

\end{theorem}

\begin{proof}
As with Theorem (\ref{thm:van-derj}), we can prove this by inserting $\nabla_W^{\mbox{*}} f_i(h_{i-1})$ from \eqref{subeq:62}, $\nabla_U^{\mbox{*}} f_i(h_{i-1})$ from \eqref{subeq:63}, and $\nabla_b^{\mbox{*}} f_i(h_{i-1})$ from \eqref{subeq:64} into \eqref{eq:nabjtheta}.
\end{proof}

We can use the results from Theorems \ref{thm:van-derj} and \ref{thm:van-nabj} to create a specific BPTT algorithms for vanilla RNNs, which we present in Algorithm \ref{alg:bptt-van}. We have the same inputs and outputs as Algorithm \ref{alg:bptt-rnn}, although our transition parameters $\theta$ are now $\theta = \{W,U,b\}$, and our prediction parameters $\zeta$ are now $\zeta = \{V,c\}$.

\begin{algorithm}[H]
  \caption{One iteration of gradient descent for for a vanilla RNN via BPTT}
   \label{alg:bptt-van}
  \begin{algorithmic}[1]
    \STATE \textbf{function} \textsc{GradDescVanillaBPTT}%
$(\mathbf{x}, \mathbf{y}, h, \theta, \zeta, \eta)$
    \INDSTATE $h_0 \leftarrow h$ 
    \begin{ALC@g}
    \STATE $\nabla_W \mathcal{J}, \nabla_U \mathcal{J}, \nabla_b \mathcal{J} \leftarrow 0$ \COMMENT{0 in their respective spaces}
    \STATE $\nabla_V \mathcal{J}, \nabla_c \mathcal{J} \leftarrow 0$
    \FOR{$i \in \{1,...,L\}$} 
        \STATE $z_i \leftarrow W \cdot h_{i-1} + U \cdot x_i + b$ 
        \STATE $h_i \leftarrow \Psi(z_i)$ \COMMENT{Specific defnition of $f_i$}
        \STATE $\widehat{y}_i \leftarrow V \cdot h_i + c$ \COMMENT{Specific defnition of $g$}
        \STATE $e_i \leftarrow \widehat{y}_i - y_i$
        \STATE $\nabla_c \mathcal{J} \leftarrow \nabla_c \mathcal{J} + e_i$ \COMMENT{Inserted \eqref{subeq:83} into \eqref{eq:rnn-nabzetj} to accumulate gradient}
        \STATE $\nabla_V \mathcal{J} \leftarrow \nabla_V \mathcal{J} + e_i \cdot h_i^\intercal$ \COMMENT{Inserted \eqref{subeq:82} into \eqref{eq:rnn-nabzetj} to accumulate gradient}
    \ENDFOR
    \STATE $\varepsilon_{L+1} \leftarrow 0$ \COMMENT{0 in $E_h$; Initialization of $D_{h_{L+1}} \mathcal{J}$}
    \FOR{$i \in \{1,...,L\}$} 
        \STATE $\varepsilon_i \leftarrow W^T \cdot D\Psi(z_{i+1}) \cdot \varepsilon_{i+1} + V^T \cdot e_i$ \COMMENT{BPTT Update Step with \eqref{subeq:61} and \eqref{subeq:81}}
        
        \STATE $\nabla_b \mathcal{J} \leftarrow \nabla_b \mathcal{J} + D\Psi(z_i) \cdot \varepsilon_i$ \COMMENT{Inserted \eqref{subeq:64} into \eqref{eq:nabjtheta}}
        \STATE $\nabla_W \mathcal{J} \leftarrow \nabla_W \mathcal{J} + \big(D\Psi(z_i) \cdot \varepsilon_i \big) h_{i-1}^\intercal$ \COMMENT{Inserted \eqref{subeq:62} into \eqref{eq:nabjtheta}}
        \STATE $\nabla_U \mathcal{J} \leftarrow \nabla_U \mathcal{J} + \big(D\Psi(z_i) \cdot \varepsilon_i \big) x_i^\intercal$  \COMMENT{Inserted \eqref{subeq:63} into \eqref{eq:nabjtheta}}
    \ENDFOR
    
    \STATE $\theta \leftarrow \theta - \eta \nabla_\theta \mathcal{J}$ \COMMENT{Parameter Update Steps for all $\theta, \zeta$}
    \STATE $\zeta \leftarrow \zeta - \eta \nabla_\zeta \mathcal{J}$
    \RETURN $\theta, \zeta$
   \end{ALC@g}
  \end{algorithmic}
\end{algorithm}

\clearpage
\subsection{Gated RNNs}
Beyond just the vanilla RNN, there exist numerous variants. The challenge of long-term dependency leads to vanishing and exploding gradients which are prevalent in vanilla RNNs, necessitating the development of \emph{gated} RNN architectures. This issue is explained in \ref{subsec:ann-rnn}.
 The standard techniques of BPTT and RTRL can be applied in gated RNNs. \\

We can understand the success of the LSTM by referring to \cite{lstm3}, particularly section 2, where the transition and prediction equations are defined. We notice that the \emph{cell state} at layer $t$, denoted $c^t$---one of the hidden states in the LSTM---is updated such that the norm of the Jacobian of the evolution from layer $t - 1$ is close to 1. This adds stability to the calulation of gradients, allowing longer-term dependencies to propagate further backwards through the network and forgoing the need for truncated BPTT.\\

We notice from \cite{lstm3} that he update and prediction steps for the LSTM are quite complicated, requiring six equations in total. Thus, a simpler gating mechanism requiring fewer parameters and update equations than the LSTM---now referred to as the Gated Recurrent Unit (GRU) \cite{lstm2}--- was introduced in \cite{lstm5}. The GRU state update still maintains an additive component, as in the LSTM, but does not explicitly contain a memory state. Introducing a GRU has been shown to be at least as effective as the LSTM on certain tasks while converging faster \cite{lstm2}. Another interesting comparison between LSTM and GRU is given in \cite{lstm6}, where the authors demonstrate empirically that the performance between the two is almost equal.

\subsubsection{Functions of LSTM}
The most important component of LSTM is the cell state $c_i^{(t)},$ which is illustrated in \ref{subsec:ann-rnn}. $c_i^{(t)}$ has a linear self-loop similar to the leaky units described in \ref{subsec:ann-rnn}. Here, however, the self-loop weight (or the associated time constant) is controlled by a foget gate unit $f_t^{t}$ (for time step $t$ and cell $i$), which sets this weight to a value between $0$ and $1$ via a sigmoid unit:

$$f_i^{(t)} = \sigma \Bigg( b_i^f + \sum_{j}^{} U_{i,j}^f x_j^{(t)} + \sum_{j}^{} W_{i,j}^f h_j^{(t-1)} \Bigg),$$

where $x^{(t)}$ is the current input vector $h^{(t)}$ is the current hidden layer vector, containing the outputs of all LSTM cells, and $b^f, U^f, W^f$ are respectively biases, input weights, and recurrent weights for the forget gates. The LSTM cell internal state is thus updates as follows, but with a conditional self-loop weight $f_i^{(t)}$:

$$c_i^{(t)} = f_i^{(t)} c_i^{(t-1)} + g_i^{(t)} \sigma \Bigg( b_i + \sum_{j}^{} U_{i,j} x_j^{(t)} + \sum_{j}^{} W_{i,j} h_j^{(t-1)} \Bigg),$$

where $b, U and W$ respectively denote the biases, input weights, and recurrent weights into the LSTM cell. The \textit{external input gate} unit $g_i^{(t)}$ is computed similarly to forget gate (with a sigmoid unit to obtain a gating value between 0 and 1), but with its own parameters: 

$$g_i^{(t)} = \sigma \Bigg( b_i^g + \sum_{j}^{} U_{i,j}^g x_j^{(t)} + \sum_{j}^{} W_{i,j}^g h_j^{(t-1)} \Bigg).$$

The output $h_i^{(t)}$ of the LSTM cell can also be shut off, via the \textit{output gate} $q_i^{(t)}$, which also uses a sigmoid unit for gating:

$$h_i^{(t)} = tanh \Big(c_i^{(t)}\Big) q_i^{(t)},$$
$$q_i^{(t)} = \sigma \Bigg( b_i^o + \sum_{j}^{} U_{i,j}^o x_j^{(t)} + \sum_{j}^{} W_{i,j}^o h_j^{(t-1)} \Bigg),$$

which has parameters $b^o, U^o, W^o$ for its biases, input weights and recurrent weights, respectively. Among the variants, one can choose to use the cell state $c_i^{(t)}$ as an extra input (with its weight) into the three gates of the i-th unit. This would require three additional parameters.

\section{Conclusion}
In final chapter, we introduced a mathematical framework on which neural networks can be defined. By dint of this framework, neural networks can be formalized and gradient descent (GD) algorithm (which incorporates backpropagation) is expressed through the framework. At first, the framework was developed for a generic NN and GD was expressed through it. Afterward, we concentrated on specific structures, from which we chose RNNs. Finally, we narrow our scope to specific RNN structures, which are Vanilla RNN and LSTM.


\appendix
\renewcommand*{\thesection}{\Alph{section}}\textbf{}




\bibliographystyle{ieeetr}
\bibliography{references}
\addcontentsline{toc}{chapter}{Bibliography}

\end{document}